\documentclass[12pt]{article}

\usepackage{tqian-latex-style}

\usepackage[toc,page]{appendix}

\numberwithin{equation}{section}
\numberwithin{figure}{section}
\numberwithin{table}{section}

\usepackage{thmtools}

\declaretheorem[
    name=Theorem,
    Refname={Theorem,Theorems},
    numberwithin=section]{thm}
\declaretheorem[
    name=Lemma,
    Refname={Lemma,Lemmas},
    sibling=thm]{lem}
\declaretheorem[
    name=Corollary,
    Refname={Corollary,Corollarys},
    sibling=thm]{cor}
\declaretheorem[
    name=Assumption,
    Refname={Assumption,Assumptions},
    numberwithin=section]{asu}
\declaretheorem[
    name=Remark,
    Refname={Remark,Remarks}]{rmk}
\declaretheorem[
    name=Definition,
    Refname={Definition,Definitions},
    numberwithin=section]{defn}
\declaretheorem[
    name=Example,
    Refname={Example,Examples},
    numberwithin=section]{ex}

\usepackage{etoolbox}
\apptocmd{\leqm}{\addcontentsline{toc}{subsubsection}{Lemma \thelem}}{}{}
\apptocmd{\thm}{\addcontentsline{toc}{subsubsection}{Theorem \thethm}}{}{}
\apptocmd{\asu}{\addcontentsline{toc}{subsubsection}{Assumption \theasu}}{}{}

\begin{document}

\title{\bf Dynamic Causal Mediation Analysis for Intensive Longitudinal Data}
\author{Tianchen Qian\footnote{Department of Statistics, University of California, Irvine. t.qian@uci.edu}}
\date{}
\maketitle

\spacingset{1.4}
\begin{abstract}
    Intensive longitudinal data, characterized by frequent measurements across numerous time points, are increasingly common due to advances in wearable devices and mobile health technologies. We consider evaluating causal mediation pathways between time-varying exposures, time-varying mediators, and a final, distal outcome using such data. Addressing mediation questions in these settings is challenging due to numerous potential exposures, complex mediation pathways, and intermediate confounding. Existing methods, such as interventional and path-specific effects, become impractical in intensive longitudinal data. We propose novel mediation effects termed natural direct and indirect excursion effects, which quantify mediation through the most immediate mediator following each treatment time. These effects are identifiable under plausible assumptions and decompose the total excursion effect. We derive efficient influence functions and propose multiply-robust estimators for these mediation effects. The estimators are multiply-robust and accommodate flexible machine learning algorithms and optional cross-fitting. In settings where the treatment assignment mechanism is known, such as the micro-randomized trial, the estimators are doubly-robust. We establish the consistency and asymptotic normality of the proposed estimators. Our methodology is illustrated using real-world data from the HeartSteps micro-randomized trial and the SleepHealth observational study.
\end{abstract}

\noindent%
{\it Keywords:} longitudinal observational study, micro-randomized trial, multiple robustness, natural mediated excursion effect, semiparametric efficiency theory
\vfill

\newpage

\tableofcontents

\newpage

\spacingset{1.9} 

\section{Introduction}
\label{sec:introduction}

Wearable and sensor technologies have made intensive longitudinal data common across fields like behavioral science, sleep research, and mobile health \citep{mcneish2021measurement}. These data are characterized by frequent, repeated measurements that generate hundreds or thousands of time points per individual. A common scenario includes time-varying exposures, time-varying mediators, and a distal outcome measured at the end of the study. For example, in the SleepHealth observational study, the daily exposure might be sleep quality, the daily mediator might be next-day sleepiness, and the distal outcome might be a long-term measure of cognitive functioning or alertness \citep{deering2020real}.

Intensive longitudinal data also arises in experimental settings, such as the micro-randomized trial \citep[MRT;][]{dempsey2015,liao2016sample}. MRTs are widely used in mobile health research for assessing proximal and distal effects of adaptive interventions \citep{walton2018optimizing,leong2023participant}. In an MRT, participants are repeatedly randomized to receive interventions such as notifications or activity prompts, generating intensive longitudinal data with time-varying treatments, time-varying mediators, and a distal outcome. For example, in the HeartSteps study \citep{klasnja2015microrandomized}, individuals were randomized five times daily for 42 days to receive activity suggestions. Immediate step counts were measured following each prompt, and long-term physical activity levels were also recorded to assess sustained behavioral change.

Our central question concerns the mediation pathways linking time-varying exposures ($A_t$), time-varying mediators ($M_t$), and a distal outcome ($Y$) in intensive longitudinal data. In the HeartSteps MRT, this corresponds to asking whether the effect of an activity suggestion on long-term activity levels is mediated by the short-term step count. In the SleepHealth observational study, it corresponds to asking whether daily sleep quality affects long-term alertness via next-day sleepiness.

Addressing mediation questions in experimental or observational intensive longitudinal data presents substantial challenges. The extensive number of time points yields numerous potential exposures and mediation pathways, and this makes it infeasible to model and estimate traditional causal effects that contrast fixed treatment trajectories \citep{robins2000marginalstructural}. Another challenge is intermediate confounding; for example, time-varying mediators are influenced by prior treatments and subsequently confound future mediator-outcome relationships. Existing causal mediation methods for time-varying exposures, mediators, and a distal outcome primarily include interventional effects \citep{vanderweele2017mediation,diaz2022efficient} and path-specific effects \citep{avin2005identifiability,miles2020semiparametric}. However, these methods become impractical when dealing with the extensive temporal structure of intensive longitudinal data due to the curse of dimensionality.

Our first contribution is to propose a novel class of mediation effects, termed \textit{natural (in)direct excursion effects}. These effects quantify how the influence of a time-varying exposure $A_t$ on a distal outcome $Y$ is mediated through the immediate mediator $M_t$, rather than through future mediators or directly. Focusing on the immediate mediator reduces dimensionality substantially and results in interpretable summaries of the causal mechanism. Moreover, these natural mediation effects are identifiable because the focus on the immediate mediator makes the no intermediate confounding assumption plausible. Additionally, these effects decompose the total excursion effect \citep{qian2025distal}; in contrast, interventional effects generally do not sum to the total effect.

Our new effects definition builds on three key ideas. First, we adopt the concept of excursion effects, which contrast policies that differ from the behavior policy only at a specific decision point but match elsewhere. While excursion effects are well-established for quantifying proximal effects in MRTs with numerous treatment occasions \citep{boruvka2018,dempsey2020stratified}, we extend this framework to mediation analysis. Second, we split the total effect based on the immediate mediator alone, which reduces dimensionality and enables identification of natural mediation effects. Third, we restrict attention to treatment policies that adhere to a time-varying eligibility criterion, thereby circumventing common positivity violations in time-varying exposure studies \citep{loewinger2024nonparametric}.

Our second contribution is to propose \textit{multiply-robust estimators} for these natural (in)direct excursion effects. We derive the efficient influence function for the mediation functionals at a given time point and propose multiply-robust estimators that summarize these time-varying effects via projections onto linear spaces. In settings such as MRTs where the treatment assignment mechanism is known, our estimators are doubly-robust. The estimators accommodate data-adaptive machine learning algorithms for estimating nuisance functions, along with optional cross-fitting to ensure valid inference even when using complex, high-dimensional prediction algorithms \citep{chernozhukov2018double}. We establish consistency and asymptotic normality for both non-cross-fitted and cross-fitted estimators.

Finally, we illustrate the methodology in two real-world applications. In the HeartSteps MRT, we find that the indirect effect of activity suggestions on long-term activity level (mediated through short-term step counts) increases in relative importance over time. In the SleepHealth observational study, we find that daily sleep quality influences long-term alertness both through next-day sleepiness and through additional pathways.

\cref{sec:notation} describes data structure and presents the definition of natural (in)direct excursion effects. \cref{sec:identification} presents causal assumptions and identification result. \cref{sec:methods} presents the multiply-robust estimators and the asymptotic theory. \cref{sec:simulation} presents the simulation study. \cref{sec:heartsteps} and \cref{sec:sleephealth} present the real data application to HeartSteps and SleepHealth, respectively. \cref{sec:discussion} concludes with discussion.


\section{Mediated Causal Excursion Effects}
\label{sec:notation}

\subsection{Data Structure}
\label{subsec:mrt-data-structure}

We first describe the structure of a micro-randomized trial (MRT) with time-varying covariates, time-varying treatments, time-varying mediators, and a distal outcome. We then describe a general observational longitudinal study where our method also applies. The two designs differ by whether the time-varying treatments are sequentially randomized. For consistency, we use ``treatment'' to also refer to ``exposure'' in observational contexts.

Let $n$ be the number of participants in the MRT, with participant $i$ observed for $T_i$ pre-scheduled time points. Variables without subscripts $i$ refer to a generic participant, whose data structure is shown in \cref{fig:data-structure}. At each $t = 1,\ldots,T$, we observe a triplet $(X_t, A_t, M_t)$ in temporal order. $X_t$ denotes time-varying covariates recorded after mediator $M_{t-1}$ but before treatment $A_t$. $A_t$ denotes a binary treatment at time point $t$, with $1$ indicates treatment and $0$ indicates no treatment; $A_t$ is sequentially randomized in a MRT. $M_t$ denotes the mediator observed after $A_t$ and before $X_{t+1}$, which can be the proximal outcome in an MRT. We focus on settings where $M_t$ is measured immediately after $A_t$; this makes the no-intermediate-confounding assumption more plausible (see \cref{sec:identification}). $Y$ denotes the distal outcome measured at the end of the study, which usually captures the primary health outcome. The observed history for a participant up to $t$ (right before $A_t$) is $H_t := (X_1,A_1,M_1,\ldots,X_{t-1},A_{t-1},M_{t-1},X_t)$. The full data for participant $i$ is $O_i := (X_{i1},A_{i1},M_{i1},X_{i2},A_{i2},M_{i2},\ldots,X_{iT_i},A_{iT_i},M_{iT_i},Y_i)$.

  
\begin{figure}[htbp]
    \centering  
        \begin{tikzpicture}[
        -Latex,auto,node distance =0.5 cm and 0.5 cm,semithick,
        state/.style ={circle, draw, minimum size = 0.7 cm, inner sep = 1pt},
        point/.style = {circle, draw, inner sep=0.04cm,fill,node contents={}},
        bidirected/.style={Latex-Latex,dashed},
        el/.style = {inner sep=2pt, align=left, sloped}]

        \node[state] (x1) at (0,0) {$X_1$};
        \node[state] (a1) [right =of x1] {$A_1$};
        \node[state] (m1) [right =of a1] {$M_1$};

        \path (x1) edge (a1);
        \path (a1) edge (m1);
        \path (x1) edge[bend left=40] (m1);




        \node[draw=none] (dots1) [right =0.7cm of m1] {$\cdots$};
        \path [arrows = {-Computer Modern Rightarrow[line cap=round]}] (m1) edge[line width = 1mm] (dots1);

        \node[state] (xt) [right =0.7cm of dots1] {$X_t$};
        \node[state] (at) [right =of xt] {$A_t$};
        \node[state] (mt) [right =of at] {$M_t$};

        \path (xt) edge (at);
        \path (at) edge (mt);
        \path (xt) edge[bend left=40] (mt);

        \path [arrows = {-Computer Modern Rightarrow[line cap=round]}] (dots1) edge[line width = 1mm] (xt);

        \node[draw=none] (dots2) [right =0.7cm of mt] {$\cdots$};
        \path [arrows = {-Computer Modern Rightarrow[line cap=round]}] (mt) edge[line width = 1mm] (dots2);

        \node[state] (xT) [right =0.7cm of dots2] {$X_T$};
        \node[state] (aT) [right =of xT] {$A_T$};
        \node[state] (mT) [right =of aT] {$M_T$};

        \path (xT) edge (aT);
        \path (aT) edge (mT);
        \path (xT) edge[bend left=40] (mT);

        \path [arrows = {-Computer Modern Rightarrow[line cap=round]}] (dots2) edge[line width = 1mm] (xT);

        \node[state] (y) [right =0.7cm of mT] {$Y$};
        \path [arrows = {-Computer Modern Rightarrow[line cap=round]}] (mT) edge[line width = 1mm] (y);

        \draw [decorate,decoration={brace,amplitude=3mm,mirror,raise=3mm}, -] (x1.south west) -- (xt.south east);

        \node[draw = none] (ht) [below = 0.5cm of m1] {$H_t$};
    \end{tikzpicture}
    \caption{\footnotesize Intensive longitudinal data structure with time-varying treatments, time-varying mediators, and a distal outcome. To simplify notation, a thick arrow is used to denote arrows from all nodes on the left to all nodes on the right.}
    \label{fig:data-structure}
\end{figure}
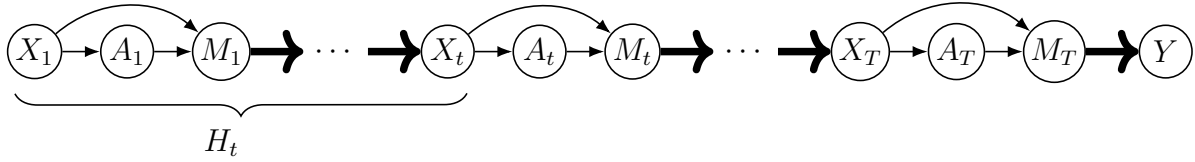

When designing an MRT, researchers may prohibit intervention delivery at times deemed unsafe or unethical---e.g., sending push notifications while a participant is driving. At such decision points, the participant is marked as ineligible, and no treatment is delivered. Formally, $X_t$ includes an indicator $I_t$, where $I_t = 1$ indicates being eligible for treatment at $t$, and $I_t = 0$ otherwise. If $I_t = 0$, then $A_t = 0$ by design. Eligibility is determined separately at each decision point for each participant. In the MRT literature $I_t = 0$ is also referred to as ``unavailable'' \citep{boruvka2018}.

The proposed method is also applicable to the general setting of observational intensive longitudinal data with data structure in \cref{fig:data-structure}. In such settings, $A_t$ are time-varying treatments that are not randomized, and $I_t = 0$ encodes times when a participant has zero probability to receive $A_t = 1$. For example, in a sleep study, a participant may have zero probability to have a good sleep ($A_t = 1$) if they are on night shift or responsible for overnight infant care for that night.

\subsubsection{Additional Notation}

Overbars denote sequences of variables from 1 to $t$; for example, $\bA_t :=(A_1,A_2,\ldots,A_t)$. Underbars denote sequences of variables from $t$ to $T$; for example, $\uM_t :=(M_t,M_{t+1}, \ldots, M_T)$. Although $I_t$ is contained in $H_t$, we sometimes write $I_t$ alongside $H_t$ in conditional expectations: For example, $\EE(\cdot | H_t, I_t = 1)$ means $\EE(\cdot | H_t \setminus \{I_t\}, I_t = 1)$. Data from participants are assumed to be independent and identically distributed from some distribution $\cP_0$, and expectation $\EE$ and probability $P$ are taken with respect to $\cP_0$ unless stated otherwise. 

Let $\PP_n$ represent the empirical mean across all participants. For any positive integer $k$, we define $[k] := \{1, 2, \ldots, k\}$. The superscript $\star$ indicates quantities related to the true data-generating distribution $\cP_0$. For a vector $\alpha$ and a vector-valued function $f(\alpha)$, we use $\partial_\alpha f(\alpha) := \partial f(\alpha) / \partial \alpha^\top$ to denote the Jacobian matrix, where the $(i,j)$-th entry is the partial derivative of the $i$-th component of $f$ with respect to the $j$-th component of $\alpha$.

\subsection{Mediation Functional Under Excursion Policies}
\label{subsec:def-mediation-functional}

To define causal effects, we use the potential outcomes notation \citep{rubin1974estimating,robins1986new}. Lowercase letters represent specific instantiations (non-random values) of the corresponding capitalized random variables. For example, $a_t$ and $m_t$ are instantiations of $A_t$ and $M_t$, respectively. For every participant, denote by $X_t(\ba_{t-1}, \bm_{t-1})$ the $X_t$ that would have been observed if the participant were assigned a treatment sequence of $\ba_{t-1}$ and a mediator sequence $\bm_{t-1}$. $M_t(\ba_t,\bm_{t-1})$ is defined similarly but it depends additionally on $a_t$. The potential outcome of $Y$ under $(\ba_T,\bm_T)$ is $Y(\ba_T,\bm_T)$.

We define a \textit{policy} as any decision rule that assigns $A_t$ given the history $H_t$ for each $t \in [T]$ \citep{murphy2003optimal}. We define $D := (d_1,\ldots, d_T)$ as the \textit{behavior policy}, i.e., the policy under the observed data distribution $\cP_0$: for MRTs, $D$ represents the policy designed by researchers to sequentially randomize $A_t$; for observational studies, $D$ represents the policy nature uses to sequentially assign $A_t$. Specifically, each $d_t(\cdot)$ is a stochastic mapping that maps $H_t$ to $\{0,1\}$, such that $d_t(H_t) = a$ with probability $P(A_t = a \mid H_t)$ for $a \in \{0,1\}$ (recall $P$ is with respect to the truth $\cP_0$). For a fixed $t$, we consider two alternative rules for assigning $A_t$ that will be contrasted in defining the causal effect: $d_t^1$ always assigns $A_t = 1$ except when $I_t = 0$ (in which case $A_t = 0$), and $d_t^0$ always assigns $A_t = 0$. That is, $d_t^1(H_t) = I_t$, and $d_t^0(H_t) = 0$. For $a\in\{0,1\}$, let $D_t^a$ denote the policy obtained by replacing $d_t$ with $d_t^a$ in the behavior policy $D$:
\begin{align}
    D_t^a := (d_1,\ldots,d_{t-1},d_t^a,d_{t+1},\ldots,d_T). \label{eq:def-policy_Dta}
\end{align}
$D_t^1$ and $D_t^0$ represents two \textit{excursions} \citep{boruvka2018,guo2021discussion} from the behavior policy $D$; they deviate from $D$ at decision point $t$ by always assigning $A_t = I_t$ or $A_t = 0$, respectively.
\begin{defn}[Mediation functional]
    For $a,b \in \{0,1\}$, define
    \begin{align}
        \theta_t^{ab} := \EE \Big[ Y \big\{D_t^a, \bM_{t-1}(D_t^a), M_t(D_t^b), \uM_{t+1}(D_t^a) \big\} \Big]. \label{eq:def-theta_t_ab}
    \end{align}
\end{defn}
When $a = b \in \{0,1\}$, $\theta_t^{aa} = \EE[Y\{D_t^a, \bM(D_t^a)\}]$ is the expected potential outcome of $Y$ under excursion policy $D_t^a$. When $a \neq b$, $\theta_t^{ab}$ is the mediation functional, i.e., the expected potential outcome of $Y$ in a hypothetical world where the treatments $\bA_T$ are assigned according to excursion policy $D_t^a$, but $M_t$ for the specific $t$ is set to its potential value under another excursion policy, $D_t^b$.

\begin{ex}
    \label{ex:theta_t_ab-heartsteps}
    \normalfont
    Consider the HeartSteps MRT, which evaluates an mHealth intervention package to promote physical activity. Here, $A_t$ is the randomized activity suggestion. $M_t$ is the proximal outcome (30-minute step count after decision point $t$), and $Y$ is a distal outcome reflecting long-term activity level. The quantity $\theta_t^{10}$ can be interpreted as follows. Imagine a participant for whom the activity suggestion delivery follows the MRT policy up to time $t-1$, and then always receives an activity suggestion at time $t$ (unless ineligible at $t$). However, suppose that their subsequent 30-minute step count $M_t$ is forced to what it would have been under no suggestion at $t$, through a hypothetical intervention that neutralizes the suggestion's immediate effect without influencing future behavior (i.e., future $\uM_{t+1}$ and $Y$). From time $t+1$ onward, the activity suggestion delivery again follows the MRT policy. Then $\theta_t^{10}$ is the expected value of $Y$ under this scenario. The interpretation of $\theta_t^{01}$ is analogous. \qed
\end{ex}

\begin{rmk}
    \label{rmk:theta_t_ab}
    \normalfont
    In cross-sectional settings, $\EE[Y\{a, M(b)\}]$ for $a \neq b$ is central to defining natural (in)direct effects and is known as the expected cross-world potential outcome or mediation functional \citep{tchetgen2012semiparametric}. Our $\theta_t^{ab}$ generalizes this to longitudinal data with time-varying treatments and mediators. In the definition \eqref{eq:def-theta_t_ab}, only $M_t$ is set to its natural value under $D_t^b$, i.e., ``in the other world''; all the other variables ($Y,\bM_{t-1},\uM_{t+1}$) all follow their natural values under $D_t^a$. This serves three purposes, which we elaborate on later: it reduces the dimensionality (\cref{rmk:decomposition}), aligns with scientific goals (\cref{rmk:decomposition}), and facilitates identification (\cref{sec:identification}). \qed
\end{rmk}

\subsection{Natural (In)direct Excursion Effects}
\label{subsec:def-mediation-effect}

We now define the total and mediation effects and discuss their interpretation.

\begin{defn}[Total excursion effect and natural (in)direct excursion effects]
    The total excursion effect (TEE) contrasting excursions at decision point $t$ is $\tee_t := \theta_t^{11} - \theta_t^{00}$. The natural indirect excursion effects (NIEE) and the natural direct excursion effects (NDEE) are defined as
    \begin{align*}
        \nieet_t := \theta_t^{11} - \theta_t^{10}, & & \nieep_t := \theta_t^{01} - \theta_t^{00}, \\
        \ndeep_t := \theta_t^{10} - \theta_t^{00}, & & \ndeet_t := \theta_t^{11} - \theta_t^{01},
    \end{align*}
    with the decomposition
    \begin{align}
        \tee_t = \nieet_t + \ndeep_t = \nieep_t + \ndeet_t. \label{eq:effect-decomposition}
    \end{align}
\end{defn}

\begin{rmk}[Excursion aspect of the effects]
    \label{rmk:excursion-aspect}
    \normalfont
    The total excursion effect $\tee_t$ captures the effect of changing treatment at a single time point $t$, while holding the treatment policy at other time points fixed at the behavior policy. It is a special case of the distal causal excursion effect in \citet{qian2025distal}, with their effect modifiers $S_t = \emptyset$. This effect is scientifically relevant because the effect is contextualized in a realistic scenario (i.e., the actual behavior policy) rather than an arbitrary or unrealistic policy. This also allows the effect to incorporate nuanced yet important aspects in intensive longitudinal studies and MRTs, such as user burden (e.g., Example 2 in \citet{qian2025distal}) and treatment effects on future eligibility (e.g., Example 3 in \citet{qian2025distal}). Statistically, anchoring the effect in the behavior policy reduces the number of parameters. These features apply similarly to NIEE and NDEE. Further discussion of excursion effects appears in \citet{guo2021discussion,zhang2021discussion}. \qed
\end{rmk}

\begin{rmk}[Effect decomposition based on the immediate mediator]
    \label{rmk:decomposition}
    \normalfont
    The decomposition \eqref{eq:effect-decomposition} splits the total effect of $A_t$ on $Y$ into indirect effects mediated through the immediate mediator $M_t$ and direct effects encompassing all other pathways. For brevity, we refer to the latter as ``direct'', though it includes indirect pathways such as $A_t \to M_{t+1} \to Y$.

    Scientifically, focusing on the immediate mediator aligns with key research questions. For instance, mHealth interventions ($A_t$) designed to encourage healthy behaviors often target proximal mediators ($M_t$) as short-term indicators of progress toward an ultimate health goal ($Y$). A concrete example is promoting immediate physical activity to achieve sustained long-term exercise behavior \citep{thompson2014texting,nahum2018just}. Our decomposition explicitly tests hypotheses regarding whether affecting these immediate mediators ultimately improves distal health outcomes. Similar logic applies in observational studies, such as assessing whether better sleep quality improves long-term alertness via reduced next-day sleepiness (see \cref{sec:sleephealth}).

    Statistically, focusing on the immediate mediator reduces the complexity of the parameter space. Moreover, it enables the identifiability of natural mediation effects, because the no intermediate confounding assumption is more plausible. We elaborate on identifiability in \cref{sec:identification}. \qed
    
\end{rmk}


\setcounter{ex}{0}
\begin{ex}[\textbf{continued}]
    \normalfont
    Consider again the HeartSteps MRT from \cref{ex:theta_t_ab-heartsteps}. Here, $\nieet_t$ and $\nieep_t$ quantify the indirect effect of a prompt at decision point $t$ on the distal outcome through immediate changes in the subsequent 30-minute step count. $\ndeep_t$ and $\ndeet_t$ capture the direct impact of the prompt on the distal outcome, assuming the immediate step count remains unchanged at the no-prompt natural level.
\end{ex}

\begin{ex}
    \normalfont
    Suppose $T = 2$ and $I_t = 1$ with probability 1 for $t = 1, 2$. Then we have
    \begin{align*}
        \theta_1^{a_1 b_1} & = \EE\Big[ Y\big\{a_1, A_2(a_1), M_1(b_1), M_2(a_1, A_2(a_1))\big\}\Big], \\
        \theta_2^{a_2 b_2} & = \EE\Big[ Y\big\{A_1, a_2, M_1(A_1), M_2(A_1, b_2)\big\}\Big].
    \end{align*}
    The $\nieet_t$ and $\ndeep_t$ for $t = 1$ are
    \begin{align*}
        \nieet_{t=1} & = \EE\Big[ Y\big\{1, A_2(1), M_1(1), M_2(1, A_2(1))\big\}\Big] - \EE\Big[ Y\big\{1, A_2(1), M_1(0), M_2(1, A_2(1))\big\}\Big], \\
        \ndeep_{t=1} & = \EE\Big[ Y\big\{1, A_2(1), M_1(0), M_2(1, A_2(1))\big\}\Big] - \EE\Big[ Y\big\{0, A_2(0), M_1(0), M_2(0, A_2(0))\big\}\Big].
    \end{align*}
    The other NIEEs and NDEEs can be expressed similarly. \cref{fig:direct-indirect-effects} illustrates the pathways that constitute NIEE and NDEE, respectively, for $t=1$. \qed
\end{ex}

\begin{figure}[htbp]
    \centering
        \begin{tikzpicture}[>=Stealth, node distance=1cm and 1cm]
        \definecolor{directBlue}{HTML}{0072B2}      
        \definecolor{indirectOrange}{HTML}{E69F00}  
        \definecolor{customGray}{HTML}{A0A0A0}      

        \node (X1) {\(X_1\)};
        \node[right=of X1] (A1) {\(A_1\)};
        \node[below=of A1] (M1) {\(M_1\)};
        
        \node[right=of A1] (X2) {\(X_2\)};
        \node[right=of X2] (A2) {\(A_2\)};
        \node[below=of A2] (M2) {\(M_2\)};
        
        \node[right=of A2] (Y) {\(Y\)};
        
        \draw[->, customGray] (X1) to (A1);
        \draw[->, customGray] (X1) to (M1);
        \draw[->, indirectOrange, thick] (A1) to (M1);
        
        \draw[->, indirectOrange, thick] (X2) to (A2);
        \draw[->, indirectOrange, thick] (X2) to (M2);
        \draw[->, indirectOrange, thick] (A2) to (M2);
        
        \draw[->, customGray, bend left=20] (X1) to (X2);
        \draw[->, customGray, bend left=25] (X1) to (A2);
        \draw[->, customGray] (X1) to (M2);
        \draw[->, customGray] (A1) to (X2);
        \draw[->, customGray, bend left=20] (A1) to (A2);
        \draw[->, customGray] (A1) to (M2);
        \draw[->, indirectOrange, thick] (M1) to (X2);
        \draw[->, indirectOrange, thick] (M1) to (A2);
        \draw[->, indirectOrange, thick] (M1) to (M2);
        
        \draw[->, customGray, bend left=30] (X1) to (Y);
        \draw[->, customGray, bend left=30] (A1) to (Y);
        \draw[->, indirectOrange, thick] (M1) to (Y);

        \draw[->, indirectOrange, thick, bend left=20] (X2) to (Y);
        \draw[->, indirectOrange, thick] (A2) to (Y);
        \draw[->, indirectOrange, thick] (M2) to (Y);

        \node[anchor=south west, font=\bfseries] at (current bounding box.south west) {(a)};

    \end{tikzpicture}
    \hfill
    \begin{tikzpicture}[>=Stealth, node distance=1cm and 1cm]
        \definecolor{directBlue}{HTML}{0072B2}      
        \definecolor{indirectOrange}{HTML}{E69F00}  
        \definecolor{customGray}{HTML}{A0A0A0}      

        \node (X1) {\(X_1\)};
        \node[right=of X1] (A1) {\(A_1\)};
        \node[below=of A1] (M1) {\(M_1\)};
        
        \node[right=of A1] (X2) {\(X_2\)};
        \node[right=of X2] (A2) {\(A_2\)};
        \node[below=of A2] (M2) {\(M_2\)};
        
        \node[right=of A2] (Y) {\(Y\)};
        
        \draw[->, customGray] (X1) to (A1);
        \draw[->, customGray] (X1) to (M1);
        \draw[->, customGray] (A1) to (M1);
        
        \draw[->, directBlue, thick] (X2) to (A2);
        \draw[->, directBlue, thick] (X2) to (M2);
        \draw[->, directBlue, thick] (A2) to (M2);
        
        \draw[->, customGray, bend left=20] (X1) to (X2);
        \draw[->, customGray, bend left=25] (X1) to (A2);
        \draw[->, customGray] (X1) to (M2);
        \draw[->, directBlue, thick] (A1) to (X2);
        \draw[->, directBlue, thick, bend left=20] (A1) to (A2);
        \draw[->, directBlue, thick] (A1) to (M2);
        \draw[->, customGray] (M1) to (X2);
        \draw[->, customGray] (M1) to (A2);
        \draw[->, customGray] (M1) to (M2);
        
        \draw[->, customGray, bend left=30] (X1) to (Y);
        \draw[->, directBlue, thick, bend left=30] (A1) to (Y);
        \draw[->, customGray] (M1) to (Y);

        \draw[->, directBlue, thick, bend left=20] (X2) to (Y);
        \draw[->, directBlue, thick] (A2) to (Y);
        \draw[->, directBlue, thick] (M2) to (Y);

        \node[anchor=south west, font=\bfseries] at (current bounding box.south west) {(b)};
    \end{tikzpicture}
    \caption{\footnotesize NIEE and NDEE of $A_1$ on $Y$ with number of decision points $T=2$. (a) Paths in orange collectively illustrate the indirect effect from $A_1$ to $Y$ mediated by $M_1$. (b) Paths in blue collectively illustrate the ``direct'' effect from $A_1$ to $Y$ not mediated by $M_1$.}
    \label{fig:direct-indirect-effects}
\end{figure}
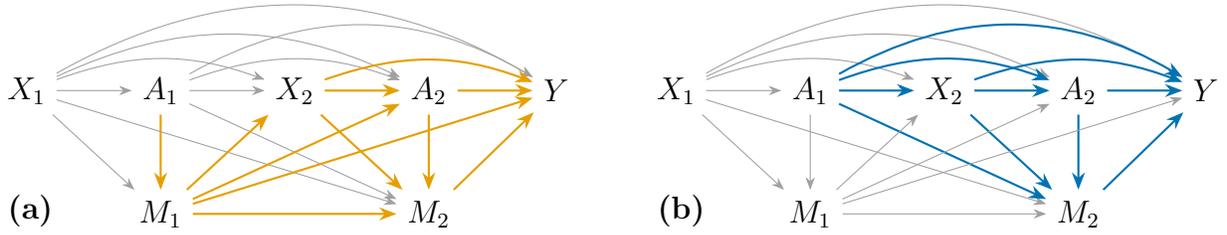

\section{Causal Assumptions and Identification}
\label{sec:identification}

We make the following causal assumptions.
\begin{asu}
    \label{asu:causal-assumptions}
    \normalfont
    \spacingset{1.5}
    \begin{asulist}
        \item \label{asu:consistency} (SUTVA.) There is no interference across participants, and the observed data equals the potential outcome under the observed treatment. Specifically, for each $t$, $X_t(\ba_{t-1}, \bm_{t-1}) = X_t$ if $\bA_{t-1} = \ba_{t-1}$ and $\bM_{t-1} = \bm_{t-1}$, and $M_t(\ba_t, \bm_{t-1}) = M_t$ if $\bA_t = \ba_t$ and $\bM_{t-1} = \bm_{t-1}$. In addition, $Y(\ba_T, \bm_T) = Y$ if $\bA_T = \ba_T$ and $\bM_T = \bm_T$.

        \item \label{asu:positivity} (Positivity when eligible.) For each $t$, there exists a positive constant $c > 0$, such that if $P(H_t = h_t, I_t = 1) > 0$ then $c < P(A_t = a \mid H_t = h_t, I_t = 1) < 1-c$ for $a \in \{0, 1\}$. Furthermore, if $P(H_t = h_t, A_t = a) > 0$, then $P(M_t = m \mid H_t = h_t, A_t = a) > 0$ for all $m$ in the support of $M_t$.

        \item \label{asu:seq-ign-A} (Sequential ignorability of $A_t$.) For each $t$, the potential outcomes $\big\{M_t(\ba_t)$, $X_{t+1}(\ba_t)$, $M_{t+1}(\ba_{t+1})$, $X_{t+2}(\ba_{t+1})$, $\ldots, X_T(\ba_{T-1})$, $M_T(\ba_T)$, $Y(\ba_T, M_T(\ba_T)): \ba_T \in \{0,1\}^{\otimes T} \big\}$ are conditionally independent of $A_t$ given $H_t$.

        \item \label{asu:seq-ign-M} (Sequential ignorability of $M_t$.) For each $t$, for any $a,b \in \{0,1\}$, and for any $m_t$ in the support of $M_t$, we have
        \begin{align*}
            M_t(D_t^b) \perp Y\{D_t^a, \bM_{t-1}(D_t^a), m_t, \uM_{t+1}(D_t^a)\} \mid H_t, A_t = d_t^b
        \end{align*}
    \end{asulist}
\end{asu}

\cref{asu:consistency} may be violated in the presence of interference, i.e., when one participant's treatment affects another's outcome. In such settings, alternative frameworks that account for causal interference would be needed \citep{hudgens2008toward,shi2022assessing}. In MRTs, the first part of \cref{asu:positivity} is guaranteed by design, while the second part can be checked empirically. In observational studies, both parts need to be checked empirically. The eligibility indicator $I_t$ allows one direction of positivity violation: $P(A_t = 1 \mid I_t = 0) = 0$. The probability $P(\cdot)$ in \cref{asu:positivity} denotes a density with respect to a dominating measure, so the same notation is used for discrete and continuous random variables. \cref{asu:seq-ign-A} holds by design in MRTs and must be assumed in observational settings.

\cref{asu:seq-ign-M} extends the cross-world independence assumption in causal mediation framework to the intensive longitudinal setting \citep{nguyen2022clarifying}. It assumes the independence between the mediator $M_t$ under $D_t^b$ and the outcome $Y$ under $D_t^a$. While multiple versions of this assumption exist with subtle distinctions \citep{pearl2001direct,robins2003semantics,petersen2006estimation}, ours generalizes Assumption 1 from \citet{imai2010identification}, implying that $M_t$ is conditionally ignorable given treatment $A_t$ and history $H_t$. This rules out ``intermediate confounding'', i.e., post-$A_t$ variables affected by $A_t$ that influence both $M_t$ and subsequent variables ($\uM_{t+1}$ or $Y$). \cref{asu:seq-ign-M} is untestable and must be justified using subject-matter knowledge. In intensive longitudinal studies, $M_t$ is often measured shortly after $A_t$, making the assumption more plausible. However, when there is a substantial delay between $A_t$ and $M_t$, this assumption may be violated, rendering NIEEs and NDEEs become unidentifiable. In such cases, alternative estimands such as interventional or path-specific mediation effects may be more appropriate \citep{vanderweele2017mediation,diaz2022efficient,avin2005identifiability,miles2020semiparametric}.

A reassuring nuance is that \cref{asu:seq-ign-M} does not preclude confounders affected by $A_t$ that influence future mediators (e.g., $M_{t+1}$) and the outcome $Y$. Such confounders are common in intensive longitudinal studies; in fact, $M_t$ itself often serves as one. NIEE and NDEE remain identifiable despite these confounders, because we focus on mediation through the immediate mediator $M_t$, rather than attempting to decompose effects through mediators at future decision points.

In Supplementary Material \ref{A-sec:proof-identification}, we prove the following identification result.
\begin{thm}[Identification]
    \label{thm:identification}
    Under \cref{asu:causal-assumptions}, $\theta_t^{aa}$ for $a \in \{0,1\}$ is identified as 
    \begin{align}
        \theta_t^{aa}
        & = \EE \Big\{ \EE\big(Y \mid H_t, A_t = d_t^a\big) \Big\} \label{eq:identify-theta-aa-ie} \\
        & = \EE \bigg\{ \frac{\indic(A_t = d_t^a)}{P(A_t = d_t^a \mid H_t)} Y \bigg\}. \label{eq:identify-theta-aa-ipw}
    \end{align}
    The mediation functional $\theta_t^{ab}$ for $a,b\in\{0,1\}$ and $a \neq b$ is identified as
    \begin{align}
        \theta_t^{ab} 
        & = \EE\Big[ \EE \big\{ \EE(Y \mid H_t, A_t = d_t^a, M_t) \mid H_t, A_t = d_t^b \big\} \Big] \label{eq:identify-theta-ab-ie} \\
        & = \EE \bigg\{ \frac{\indic(A_t = d_t^a)}{P(A_t = d_t^b \mid H_t)} \frac{P(A_t = d_t^b \mid H_t, M_t)}{P(A_t = d_t^a \mid H_t, M_t)} Y \bigg\}. \label{eq:identify-theta-ab-ipw}
    \end{align}
\end{thm}

\section{Estimation of Natural (In)direct Excursion Effects}
\label{sec:methods}

\subsection{Efficient Influence Function for Mediation Functional}
\label{subsec:eif}

We define a few nuisance functions, and we use superscript $\star$ to denote the truth. For $a \in \{0,1\}$, let
\begin{align}
    p_t^\star(a | H_t) & := P(A_t = d_t^a \mid H_t), &
    q_t^\star(a | H_t, M_t) & := P(A_t = d_t^a \mid H_t, M_t), \nonumber \\
    \eta_t^\star(a, H_t) & := \EE(Y \mid A_t = d_t^a, H_t), &
    \mu_t^\star(a, H_t, M_t) & := \EE(Y \mid A_t = d_t^a, H_t, M_t), \nonumber \\
    \nu_t^\star(a, H_t) & := \EE\{ \mu_t^\star(a, H_t, M_t) \mid A_t = d_t^{1-a}, H_t \}. \label{eq:def-true-nuisance-function}
\end{align}
Versions without superscript $\star$ denote working models that are not necessarily correct. For example, $p_t(a | H_t)$ denotes a working model that does not necessarily equal to $p_t^\star(a | H_t)$. Furthermore, for $a,b \in \{0,1\}$ and $a \neq b$, define
\begin{align}
    \phi_t^{aa}(p_t, \eta_t) & := \frac{\indic(A_t = d_t^a)}{p_t(a | H_t)} Y - \frac{\indic(A_t = d_t^a) - p_t(a | H_t)}{p_t(a | H_t)} \eta_t(a, H_t), \label{eq:def-phi-aa} \\
    \phi_t^{ab}(p_t, q_t, \mu_t, \nu_t) & := \frac{\indic(A_t = d_t^a) ~q_t(b | H_t, M_t)}{p_t(b | H_t) ~q_t(a | H_t, M_t)} \Big\{ Y - \mu_t (a, H_t, M_t) \Big\} \nonumber \\
    & ~~~~ + \frac{\indic(A_t = d_t^b)}{p_t(b | H_t)} \Big\{\mu_t (a, H_t, M_t) - \nu_t(a, H_t)\Big\} + \nu_t(a, H_t). \label{eq:def-phi-ab}
\end{align}
The following theorem gives the efficient influence functions (EIFs), with proof in in Supplementary Material \ref{A-sec:proof-eif}.
\begin{thm}[Efficient influence functions]
    \label{thm:eif}
    The EIF for $\theta_t^{aa}$ is $\phi_t^{aa}(p_t^\star, \eta_t^\star) - \theta_t^{aa}$. The EIF for $\theta_t^{ab}$ ($a \neq b$) is $\phi_t^{ab}(p_t^\star, q_t^\star, \mu_t^\star, \nu_t^\star) - \theta_t^{ab}$.
\end{thm}
\cref{thm:eif} generalizes results in the literature to excursion effects in intensive longitudinal settings with eligibility constraints. Specifically, the EIF for $\theta_t^{aa}$ generalizes the classic result in \citet{hahn1998role}. The EIF for $\theta_t^{ab}$ generalizes \citet[Lemma 3.2]{farbmacher2022causal}, which is based on the seminal work of \citet{tchetgen2012semiparametric}. The EIFs are used to construct multiply-robust estimators for NDEEs and NIEEs next.

\subsection{Multiply-Robust Estimators for NDEE and NIEE}
\label{subsec:dr-estimator}

Let $f(t)$ be a pre-specified $p$-dimensional feature vector of index $t$. We consider estimating the best linear projections of $\ndeep_t$ and $\nieet_t$ onto $f(t)$. Specifically, the true estimand is $\gamma^\star := ((\alpha^\star)^\top, (\beta^\star)^\top)^\top$, with $\alpha^\star \in \RR^p$ and $\beta^\star \in \RR^p$ defined as:
\begin{align}
    \alpha^\star &:= \arg\min_{\alpha \in \RR^p} \sumt \omega(t) \Big\{ \ndeep_t - f(t)^\top\alpha \Big\}^2, \nonumber \\
    \beta^\star &:= \arg\min_{\beta \in \RR^p} \sumt \omega(t) \Big\{ \nieet_t - f(t)^\top\beta \Big\}^2, \label{eq:alphabeta-star-def}
\end{align}
where $\omega(t)$ is a pre-specified weight function with $\sum_{t=1}^T \omega(t) = 1$. These projections yield interpretable estimands even if the true time trends of NDEE and NIEE are not linear in $f(t)$.

The choice of $f(t)$ and $\omega(t)$ should reflect the scientific goal. For example, $f(t) = (1,t)^\top$ or $f(t) = (1,t,t^2)^\top$ models linear or quadratic trends in the effects over time, while more flexible basis functions of $t$ can also be used (see \cref{sec:heartsteps}). Setting $\omega(t) = 1/T$ gives equal weight to each time point. In the simplest case, with $f(t) = 1$ and $\omega(t) = 1/T$, the scalars $\alpha^\star$ and $\beta^\star$ correspond to the average $\ndeep_t$ and $\nieet_t$ over $t \in [T]$. To estimate effects at a specific time point $t_0$, one can set $\omega(t_0) = 1$ and $\omega(t) = 0$ for $t \neq t_0$.

We estimate $\gamma^\star$ using an estimating equation based on the EIFs in \cref{thm:eif}. Let $\zeta_t := (p_t, q_t, \eta_t, \mu_t, \nu_t)$ denote the nuisance functions corresponding to time point $t$, and $\zeta := (\zeta_1, \zeta_2, \ldots, \zeta_T)$. The proposed estimating function is
\begin{align}
    \psi(\gamma; \zeta) := \sum_{t=1}^T \omega(t)
    \left[\begin{matrix}
        \big\{ \phi_t^{10}(p_t, q_t, \mu_t, \nu_t) - \phi_t^{00}(p_t, \eta_t) - f(t)^\top\alpha \big\} f(t) \\
        \big\{ \phi_t^{11}(p_t, \eta_t) - \phi_t^{10}(p_t, q_t, \mu_t, \nu_t) - f(t)^\top\beta \big\} f(t)
    \end{matrix}\right]. \label{eq:def-psi}
\end{align}
We propose two estimators for $\gamma^\star$ based on $\psi(\gamma;\zeta)$, one without cross-fitting ($\hat\gamma$) and one with cross-fitting ($\tilde\gamma$), detailed in Algorithms \ref{algo:estimator-ncf} and \ref{algo:estimator-cf}. Their asymptotic properties are given in \cref{thm:estimator-consistency,thm:estimator-normality}.

\begin{algorithm}[htbp]
    \caption{A two-stage estimator $\hat\gamma$ (without cross-fitting)}
    \label{algo:estimator-ncf}
    \spacingset{1.5}
    \vspace{0.3em}
    \textbf{Stage 1:} Fit nuisance models $\hat\zeta_t = (\hat{p}_t, \hat{q}_t, \hat\eta_t, \hat\mu_t, \hat\nu_t)$ as conditional expectations defined in \eqref{eq:def-true-nuisance-function} for all $t$ using the entire sample. In practice we often pool across $t\in[T]$. Use $\hat\zeta$ to denote $(\hat\zeta_1, \ldots, \hat\zeta_T)$.

    \textbf{Stage 2:} Obtain $\hat\gamma$ by solving $\PP_n \psi(\gamma;\hat\zeta) = 0$.
    \vspace{0.3em}
\end{algorithm}

\begin{algorithm}[htbp]
    \caption{A two-stage estimator $\tilde\gamma$ (with $K$-fold cross-fitting)}
    \label{algo:estimator-cf}
    \spacingset{1.5}
    \vspace{0.3em}
    \textbf{Stage 1:} Take a $K$-fold equally-sized random partition $(B_k)_{k=1}^K$ of observation indices $[n] = \{1,\ldots,n\}$. Define $B_k^c = [n] \setminus B_k$ for $k \in [K]$. For each $k \in [K]$, use solely observations from $B_k^c$ and estimate nuisance functions $\zeta_t = (p_t, q_t, \eta_t, \mu_t, \nu_t)$ for all $t$. The fitted models using $B_k^c$ are denoted by $\hat\zeta_{kt}$, and let $\hat\zeta_k := (\hat\zeta_{k1},\ldots,\hat\zeta_{kT})$.

    \textbf{Stage 2:} Obtain $\tilde\gamma$ by solving $K^{-1} \sum_{k=1}^K \PP_{n,k} \psi(\gamma, \hat\zeta_k) = 0$. Here $\PP_{n,k}$ denotes empirical average over observations from $B_k$.
    \vspace{0.3em}
\end{algorithm}

\begin{thm}[Consistency of $\hat\gamma$ and $\tilde\gamma$]
    \label{thm:estimator-consistency}
    Suppose \cref{asu:causal-assumptions} hold and consider $\gamma^\star = ((\alpha^\star)^\top, (\beta^\star)^\top)^\top$ defined in \eqref{eq:alphabeta-star-def}. For each $t \in [T]$, suppose that $\zeta_t'$, which is the $L_2$-limit of the fitted nuisance function $\hat\zeta_t$, satisfies one of the four conditions: (i) $p_t' = p_t^\star$ and $q_t' = q_t^\star$; or (ii) $p_t' = p_t^\star$ and $\mu_t' = \mu_t^\star$; or (iii) $\eta_t' = \eta_t^\star$, $\nu_t' = \nu_t^\star$ and $q_t' = q_t^\star$; or (iv) $\eta_t' = \eta_t^\star$, $\nu_t' = \nu_t^\star$ and $\mu_t' = \mu_t^\star$. Then under regularity conditions, $\hat\gamma \pto \gamma^\star$ as $n\to\infty$. For $\tilde\gamma$, if for all $k\in[K]$ the $L_2$-limit of $\hat\zeta_{kt}$ satisfies one of these four conditions, then $\tilde\gamma \pto \gamma^\star$ as $n\to\infty$.
\end{thm}

\begin{thm}[Asymptotic normality of $\hat\gamma$ and $\tilde\gamma$]
    \label{thm:estimator-normality}
    \spacingset{1.5}
    Suppose \cref{asu:causal-assumptions} hold and consider $\gamma^\star = ((\alpha^\star)^\top, (\beta^\star)^\top)^\top$ defined in \eqref{eq:alphabeta-star-def}. For each $t \in [T]$, suppose that the fitted nuisance function $\hat\zeta_t$ satisfy the following rates:
    \begin{align}
        \|\hat{p}_t - p_t^\star\| \cdot \big( \|\hat\eta_t - \eta_t^\star\| + \|\hat\nu_t - \nu_t^\star\| \big) = o_P(n^{-\frac{1}{2}}) \text{ and } \|\hat{q}_t - q_t^\star\| \cdot \|\hat\mu_t - \mu_t^\star\| = o_P(n^{-\frac{1}{2}}). \label{eq:nuisance-product-rate}
    \end{align}
    Let $\zeta'$ denote the $L_2$-limit of $\hat\zeta$. Under regularity conditions, we have $\sqrt{n} (\hat\gamma - \gamma^\star) \dto N(0, V)$ as $n\to\infty$, with 
    $V := [\EE\{ \partial_\gamma \psi(\gamma^\star, \zeta') \}]^{-1} ~ \EE\{ \psi(\gamma^\star, \zeta') \psi(\gamma^\star, \zeta')^\top\} ~ [\EE\{ \partial_\gamma \psi(\gamma^\star, \zeta') \}]^{-1 \top}$.
    $V$ can be consistently estimated by $[\PP_n\{ \partial_\gamma \psi(\hat\gamma, \hat\zeta) \}]^{-1} ~ \PP_n\{ \psi(\hat\gamma, \hat\zeta) \psi(\hat\gamma, \hat\zeta)^\top\} ~ [\PP_n\{ \partial_\gamma \psi(\hat\gamma, \hat\zeta) \}]^{-1 \top}$.

    For the estimator $\tilde\gamma$ that uses cross-fitting, assume that $\hat\zeta_k$ converges to the same limit $\zeta'$ in $L_2$ for all $k\in[K]$. For each $t \in [T]$, suppose that \eqref{eq:nuisance-product-rate} with $(\hat{p}_t, \hat{q}_t, \hat\eta_t, \hat\mu_t, \hat\nu_t)$ replaced by $(\hat{p}_{kt}, \hat{q}_{kt}, \hat\eta_{kt}, \hat\mu_{kt}, \hat\nu_{kt})$ is satisfied for all $k\in[K]$. Under regularity conditions, we have $\sqrt{n} (\tilde\gamma - \gamma^\star) \dto N(0, V)$ as $n\to\infty$, and $V$ can be consistently estimated by 
    \begin{align*}
        \bigg[\frac{1}{K}\sum_{k=1}^K \PP_{n,k}\big\{ \partial_\gamma \psi(\hat\gamma, \hat\zeta_k) \big\} \bigg]^{-1} \bigg[\frac{1}{K}\sum_{k=1}^K \PP_{n,k} \big\{ \psi(\hat\gamma, \hat\zeta_k) \psi(\hat\gamma, \hat\zeta_k)^\top \big\} \bigg]
        \bigg[\frac{1}{K}\sum_{k=1}^K \PP_{n,k}\big\{ \partial_\gamma \psi(\hat\gamma, \hat\zeta_k) \big\} \bigg]^{-1 \top},
    \end{align*}
    where $\PP_{n,k}$ denotes the empirical average over observations from the $k$-th partition $B_k$.
\end{thm}

The regularity conditions and the proofs are provided in Supplementary Material \ref{A-sec:proof-asymptotic}.

\begin{rmk}[Multiple robustness]
    \label{rmk:estimator-robustness}
    \normalfont
    The estimators $\hat\gamma$ and $\tilde\gamma$ are multiply-robust: they remain consistent if one of four specific pairs of the four working models are correctly specified (\cref{thm:estimator-consistency}). \cref{thm:estimator-normality} further shows that they are  rate-multiply-robust \citep{smucler2019unifying}, requiring each nuisance parameter to converge at only $o_P(n^{-1/4})$ rate to ensure $o_P(n^{-1/2})$ convergence of their products. Many data-adaptive prediction algorithms, such as generalized additive models and the highly adaptive lasso, can achieve this rate under mild conditions \citep{wood2017generalized,benkeser2016highly}. Furthermore, the asymptotic variance $V$ depends only on the $L_2$-limits of the fitted nuisance functions and does not depend on the specific prediction algorithm used, because the bias term is negligible due to rate-multiple-robustness \citep{kennedy2016semiparametric}. Lastly, the convergence condition \eqref{eq:nuisance-product-rate} implies that the consistency conditions in \cref{thm:estimator-consistency} are satisfied. \qed
\end{rmk}

In MRTs, the known randomization probabilities $p_t^\star$ yield a stronger form of robustness.

\begin{cor}[Double robustness of $\hat\gamma$ and $\tilde\gamma$ with known $p_t^\star$]
    \label{cor:normality-known-pt}
    \spacingset{1.5}
    Suppose \cref{asu:causal-assumptions} hold and consider $(\alpha^\star, \beta^\star)$ defined in \eqref{eq:alphabeta-star-def}. Suppose the treatment assignment mechanism $p_t^\star$ is known for $t \in [T]$ (such as in an MRT), and that $\hat{p}_t$ in Algorithms \ref{algo:estimator-ncf} and \ref{algo:estimator-cf} is replaced with the known $p_t^\star$. For each $t \in [T]$, suppose that the fitted nuisance functions satisfy $\|\hat{q}_t - q_t^\star\| \cdot \|\hat\mu_t - \mu_t^\star\| = o_P(n^{-1/2})$. Then \cref{thm:estimator-normality} conclusions hold.
\end{cor}

Parallel methods and asymptotic theory for $\ndeet_t$ and $\nieep_t$ are straightforward. One simply revises the estimating function $\psi$ by replacing $\phi_t^{10}$ and $\phi_t^{00}$ in the first row of \eqref{eq:def-psi} with $\phi_t^{11}$ and $\phi_t^{01}$, and replacing $\phi_t^{11}$ and $\phi_t^{10}$ in the second row of \eqref{eq:def-psi} with $\phi_t^{01}$ and $\phi_t^{00}$. The asymptotic theory for for $\ndeet_t$ and $\nieep_t$ holds with the newly defined $\psi$.


\section{Simulations}
\label{sec:simulation}

We conduct two simulation studies to evaluate the proposed estimator. The first verifies the convergence rate condition \eqref{eq:nuisance-product-rate} and the rate-multiple-robustness by correctly specifying all nuisance models and precisely controlling their convergence rates. The second assesses estimator performance in a more realistic setting with misspecified nuisance models. 

\subsection{Simulation Study 1: Verifying Rate Multiple Robustness}
\label{subsec:simulation-dgm_simple}

We describe GM-1, the generative model used in the first simulation study. We set $T = 5$. Variables $(X_t, I_t, A_t, M_t)_{1 \leq t \leq T}$ and the distal outcome $Y$ are generated sequentially. Each $X_t \sim N(0, \sigma^2_X)$, and $I_t = 1$. Treatment $A_t$ and mediator $M_t$ are both binary, with $(A_t, M_t) \mid H_t$ following a multinomial distribution with $P(A_t = a, M_t = m \mid H_t) = s_{am} / s$, where $s_{00} := 1$, $s_{10} := \exp\{\kappa_1 + h_1(t,X_t)\}$, $s_{01} := \exp\{\kappa_2 + h_2(t,X_t)\}$, $s_{11} := \exp\{\kappa_0 + \kappa_1 + \kappa_2 + h_1(t,X_t) + h_2(t,X_t)\}$, and $s := s_{00} + s_{10} + s_{01} + s_{11}$. The distal outcome $Y \sim N\Big(\sum_{t=1}^T (\xi_t X_t + \rho_t M_t + \lambda_t A_t + \tau_t A_t M_t), \sigma^2_Y\Big)$. Parameters are set as $\sigma_X = \sigma_Y = 2$, $\kappa_0 = 2$, $\kappa_1 = \kappa_2 = -1.5$, $\xi_t = \rho_t = \lambda_t = \tau_t = 0.5 + 0.25(t-1)/T$. The nonlinear $h_1(t, X_t)$ and $h_2(t, X_t)$ are:
\begin{align*}
    h_1(t, X_t) := \frac{g_{2,5}(t/T) + g_{2,5}(\expit(X_t))}{2}, \quad h_2(t, X_t) = \frac{g_{5,2}(t/T) + g_{5,2}(\expit(X_t))}{2},
\end{align*}
where $g_{2,5}$ and $g_{5,2}$ are $\text{Beta}(2,5)$ and $\text{Beta}(5,2)$ densities. With $f(t) = 1$ and $\omega(t) = 1/T$, the true marginal estimands are $\alpha^\star = 1.381$ and $\beta^\star = 0.822$. We derive closed-form expressions for the true nuisance functions $p_t^\star$, $q_t^\star$, $\eta_t^\star$, $\mu_t^\star$, and $\nu_t^\star$ (see Supplementary Material \ref{A-subsec:deriving_the_true_nuisance_functions_for_gm_1}).

We use perturbation parameters $(r_p, r_q, r_\eta, r_\mu, r_\nu)$ to control the convergence rates of the nuisance estimators to their truths. Specifically, the estimator for each nuisance function equals the true function multiplied by a random variable following $\text{Uniform}[1-n^{-r}, 1]$, where $r$ is the perturbation parameter for that nuisance function. For example, with $U_n \sim \text{Uniform}[1-n^{-r_p}, 1]$, we set $\hat{p}_t(a | H_t) = U_n ~p_t^\star(a | H_t)$, and this yields $\|\hat{p}_t - p_t^\star\| \leq \|U_n - 1\| \cdot \|p_t^\star\| = O_P(n^{-r_p})$. We set $r_p = r_\eta = r_\mu = r_1$ and $r_q = r_\mu = r_2$, and we vary $r_1, r_2 \in \{0.1, 0.2, 0.3, 0.4, 0.5\}$. By \cref{thm:estimator-normality}, $\hat\alpha$ and $\hat\beta$ are asymptotically normal when $r_1 > 0.25$ and $r_2 > 0.25$. 

\begin{figure}[htbp]
    \includegraphics[width = 0.49\textwidth]{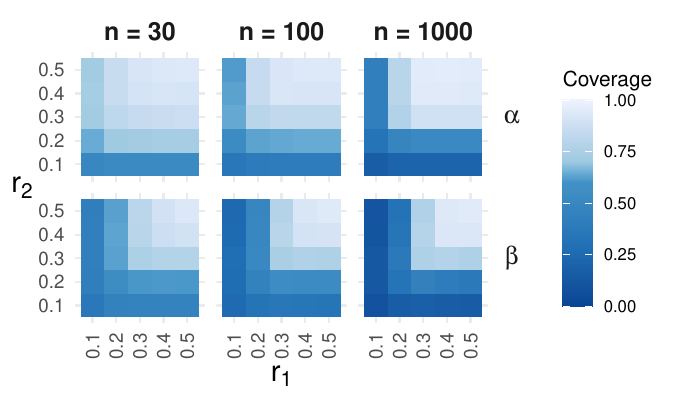}
    \includegraphics[width = 0.49\textwidth]{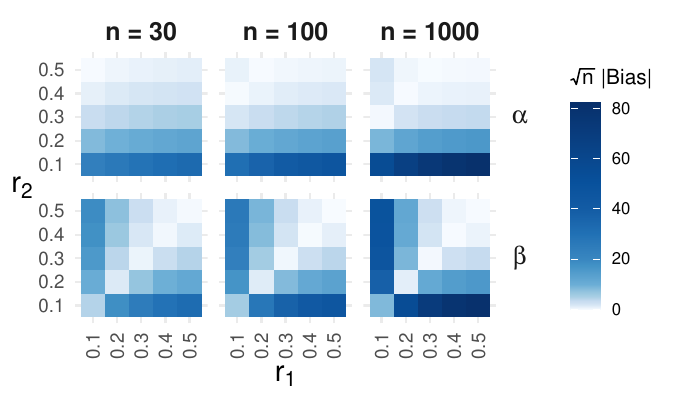}
    \caption{\footnotesize   Simulation results under GM-1: Coverage of 95\% confidence intervals (left) and $\sqrt{n}~|\text{Bias}|$ (right) for $\hat\alpha$ and $\hat\beta$ under various convergence rates of the nuisance parameters, $r_1$ and $r_2$.}
    \label{fig:simulation-1}
\end{figure}

\cref{fig:simulation-1} shows the coverage of 95\% confidence intervals and $\sqrt{n}~|\text{Bias}|$ for $\hat\alpha$ and $\hat\beta$ across different values of $r_1$ and $r_2$. The covarage is nominal when $r_1, r_2 \geq 0.4$ (top right $2\times 2$ tiles in each panel of the coverage plot) and near nominal when $r_1, r_2 \geq 0.3$. The latter suggests that larger sample sizes may be needed when nuisance functions converge slowly. The bias plot, though noisier, shows a similar pattern: $\sqrt{n}~|\text{Bias}|$ stabilizes with increasing $n$ for $r_1, r_2 \geq 0.3$ (top right $3\times 3$ tiles in each panel of the bias plot), demonstrating $\sqrt{n}$-rate convergence of the estimators which validates \cref{thm:estimator-normality}. A detailed version of \cref{fig:simulation-1}, along with additional results using $f(t) = (1,t)^\top$, is provided in Supplementary Material \ref{A-subsec:additional_results_from_simulation_study_1}.

\subsection{Simulation Study 2: Complex and Realistic MRT Setting}
\label{subsec:simulation-dgm_complex}

The second simulation uses a more complex generative model, GM-2, designed to mimic a MRT with serial dependence and nontrivial eligibility (i.e., $I_t$ not always 1). We set $T = 30$. Variables $(X_t, I_t, A_t, M_t)_{1 \leq t \leq T}$ and the distal outcome $Y$ are generated sequentially as follows:
\begin{align}
    X_t & \sim N\Big(0.3 X_{t-1} + 0.2 A_{t-1} + 0.2 M_{t-1}, 1\Big), \nonumber\\
    I_t & \sim \bern\Big\{\expit(1.5 - 0.3 A_{t-1} - 0.3 M_{t-1} + 0.3 X_t)\Big\}, \nonumber\\
    A_t & \sim \bern\Big[I_t ~ \expit\big\{0.2 A_{t-1} + 0.2 h_3(t, M_{t-1}) + 0.3 h_3(t, X_t)\big\}\Big], \label{eq:GM2-true-p}\\
    M_t & \sim N\Big\{0.4 A_{t-1} + 0.4 h_3(t, M_{t-1}) + 0.3 h_3(t, X_t) + 0.6 A_t, 1\Big\}, \nonumber\\
    Y & \sim N\Big[\sum_{t=1}^T \big\{0.3 h_3(t, X_t) + 0.4 h_3(t, M_t) + 0.2 A_t + 0.1 A_t h_3(t, M_t)\big\}, 1\Big]. \nonumber
\end{align}
Here, $h_3(t, z) := \text{tanh} \{3\times(2 t - T) / T\} + \text{sin}(z)$ is nonlinear in both arguments. With $f(t) = 1$ and $\omega(t) = 1/T$, the true marginal estimands are $\alpha_0^\star = 0.285$, $\beta_0^\star = 0.121$. With $f(t) = (1,t)^\top$ and $\omega(t) = 1/T$, the true estimands are $\alpha_1^\star = 0.247, \beta_1^\star = 0.253$ (intercepts) and $\alpha_2^\star = 0.002, \beta_2^\star = -0.008$ (slopes for $t$).

Under GM-2, $p_t^\star$ is given by \eqref{eq:GM2-true-p}. $q_t^\star$ is a nonlinear function of $(t, A_{t-1}, M_{t-1}, X_t, M_t)$ with no closed form. The other three nuisance functions are complex, history-dependent functions with no closed forms: $\eta_t^\star$ and $\nu_t^\star$ depend on $t$ and all of $H_t$, and $\mu_t^\star$ depends on $(t, M_t)$ and all of $H_t$. Because $q_t^\star$ only depends on a small subset of variables, it can be reasonably approximated with generalized additive models. In contrast, the full-history dependence of $\eta_t^\star$, $\mu_t^\star$, and $\nu_t^\star$ makes their correct specification infeasible when $T = 30$.

To capture these nuances, we simulate four scenarios under GM-2 (\cref{tab:simulation-2-scenarios}), varying the working models for $q_t$ and $\mu_t$. Model specification is color-coded: green (correct), light red (moderate misspecification), and dark red (severe misspecification). Mimicking an MRT, we use the known truth $p_t^\star$ in all scenarios. Because the estimator's robustness with known $p_t^\star$ depends on $q_t$ and $\mu_t$ (\cref{cor:normality-known-pt}), we vary their specifications accordingly. For clarity, we fix the other two nuisance functions $\eta_t = \nu_t = 0$ (severely misspecified) across all scenarios. \cref{cor:normality-known-pt} guarantees good performance in Scenarios 1 and 3.

\begin{table}[tb]
    \spacingset{1.5}
    \scriptsize
    \centering
    \begin{tabular}{c|ccccc}
    \hline\hline
    \textbf{Sce.} & Model for $p_t$ & Model for $q_t$ & Model for $\mu_t$ & Model for $\eta_t$ & Model for $\nu_t$ \\
    \hline
    1 & \cellcolor{NatureLightGreen} use true value $p_t^\star$
      & \cellcolor{NatureLightGreen} $s(t)\!+\!A_{t-1}\!+\!s(M_{t-1})\!+\!s(X_t)\!+\!s(M_t)$
      & \cellcolor{NatureLightRed} $s(t)\!+\!s(X_t)\!+\!s(M_t)$
      & \cellcolor{NatureMidRed} set to 0
      & \cellcolor{NatureMidRed} set to 0 \\

    2 & \cellcolor{NatureLightGreen} use true value $p_t^\star$
      & \cellcolor{NatureMidRed} $s(t)$
      & \cellcolor{NatureLightRed} $s(t)\!+\!s(X_t)\!+\!s(M_t)$
      & \cellcolor{NatureMidRed} set to 0
      & \cellcolor{NatureMidRed} set to 0 \\

    3 & \cellcolor{NatureLightGreen} use true value $p_t^\star$
      & \cellcolor{NatureLightGreen} $s(t)\!+\!A_{t-1}\!+\!s(M_{t-1})\!+\!s(X_t)\!+\!s(M_t)$
      & \cellcolor{NatureMidRed} $s(t)$
      & \cellcolor{NatureMidRed} set to 0
      & \cellcolor{NatureMidRed} set to 0 \\

    4 & \cellcolor{NatureLightGreen} use true value $p_t^\star$
      & \cellcolor{NatureMidRed} $s(t)$
      & \cellcolor{NatureMidRed} $s(t)$
      & \cellcolor{NatureMidRed} set to 0
      & \cellcolor{NatureMidRed} set to 0 \\
    \hline\hline
    \end{tabular}
    \caption{\footnotesize  Four simulation scenarios (Sce.) under GM-2 that differ in how nuisance parameters are estimated. Expressions like $s(t) + A_{t-1} + s(M_{t-1}) + s(X_t)$ denote generalized additive models with penalized spline terms for $t$, $M_{t-1}$, $X_t$, using appropriate link functions (identity for $\mu_t$, logistic for $q_t$). Cells are colored green, light red, or dark red to indicate correctly specified, moderately misspecified, or severely misspecified models, respectively.}
    \label{tab:simulation-2-scenarios}
\end{table}

\cref{fig:simulation-2} displays $\sqrt{n}~|\text{Bias}|$, RMSE (root mean squared error), ASE/SD (ratio of average standard error to empirical standard deviation), and the coverage of 95\% confidence intervals. Scenarios 1 and 3 perform well across all metrics as expected. Scenario 2 shows slight overcoverage and ASE/SD slightly greater than 1 due to misspecification of both $q_t$ and $\mu_t$. Scenario 4 performs poorly because both $q_t$ and $\mu_t$ are severely misspecified. Although $q_t$ and $\mu_t$ are misspecified in both Scenarios 2 and 4, Scenario 2 performs better because its $\mu_t$ model includes more relevant variables.

\begin{figure}[htbp]
    \spacingset{1.5}
    \centering
    \includegraphics[width = \textwidth]{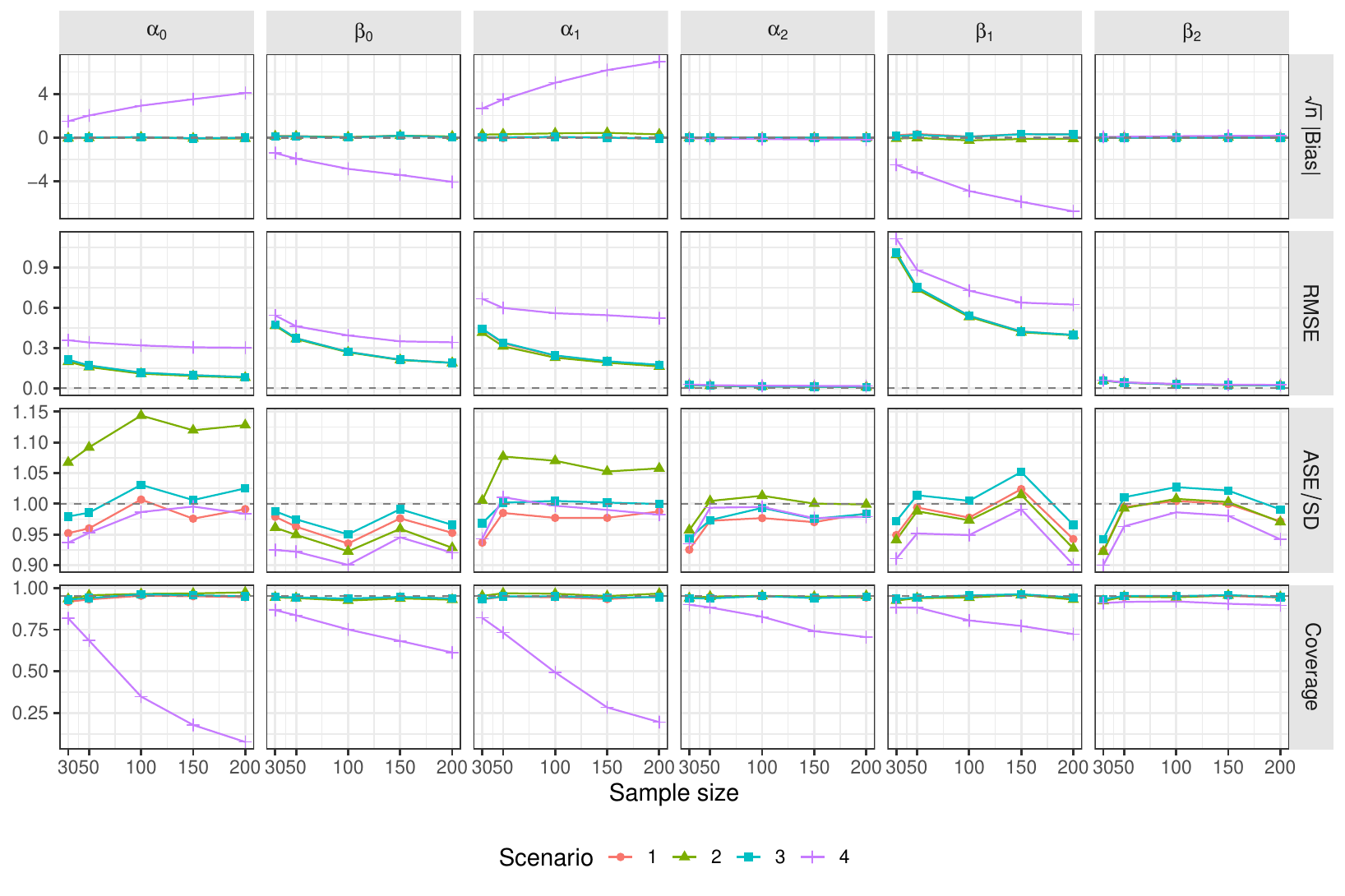}
    \caption{\footnotesize Simulation results under GM-2: $\sqrt{n}~|\text{Bias}|$, RMSE (root mean squared error), ASE/SD (ratio of average standard error to empirical standard deviation), and the coverage of 95\% confidence intervals across the four model specification scenarios in \cref{tab:simulation-2-scenarios}.}
    \label{fig:simulation-2}
\end{figure}

\section{Micro-Randomized Trial Application: HeartSteps I}
\label{sec:heartsteps}

\subsection{HeartSteps I: Background and Methods}

We analyze data from HeartSteps I, the first in a sequence of HeartSteps MRTs, for developing mHealth interventions to promote sustained physical activity in sedentary adults \citep{klasnja2015microrandomized}. We focus on the micro-randomized activity suggestions intervention. A total of $n = 37$ participants were enrolled for six weeks, with five pre-specified decision points per day (210 decision points total). We used the first five weeks ($T = 175$ decision points) to define the time-varying treatment $A_t$ and mediator $M_t$, reserving Week 6 to define the distal outcome $Y$. Participants were eligible for randomization ($I_t = 1$) if not driving or walking and had stable internet; approximately 80\% of decision points were eligible.
At each eligible decision point, a suggestion ($A_t = 1$) encouraging brief activity was delivered with probability 0.6; otherwise, no suggestion ($A_t = 0$) was given. Step counts were continuously recorded via wristband trackers, and the mediator $M_t$ is defined as the step count in the 30 minutes following decision point $t$. The distal outcome $Y$ is the average daily step count during Week 6, serving as a proxy measure for sustained physical activity habits. We discuss implications of this outcome choice later in the section.

Researchers hypothesized that activity suggestions ($A_t$) would boost immediate physical activity ($M_t$) and, in turn, promote long-term habits. While prior work has examined the effects of $A_t$ on $M_t$ \citep{klasnja2018efficacy} and $A_t$ on $Y$ \citep{qian2025distal}, the mediated pathway from $A_t$ through $M_t$ to $Y$ remains unexplored, and we address it in this analysis.

We conducted three causal mediation analyses using different specifications of $f(t)$: (i) a constant $f(t)=1$ (average effect); (ii) a linear $f(t)=(1, t-1)^\top$ (linear moderation by time); and (iii) a B-spline basis with 6 degrees of freedom (nonlinear moderation). Choices (ii) and (iii) were motivated by prior evidence that the effect of $A_t$ on $M_t$ declines over time \citep{klasnja2018efficacy}, and the term $t-1$ in (ii) improves interpretability. Given the known randomization probability in HeartSteps MRT, we used the true $p_t^\star$ and estimated nuisance functions ($\hat{q}_t, \hat\eta_t, \hat\mu_t, \hat\nu_t$) using generalized additive models (\texttt{gam} from the R package \texttt{mgcv} \citep{wood2017generalized}). For $\hat\eta_t$ and $\hat\nu_t$, covariates included a spline on prior 30-minute step count and an indicator for being at home/work. For $\hat{q}_t$ and $\hat\mu_t$, we additionally included a spline on the mediator $M_t$.

\subsection{HeartSteps I: Analysis Results}

In Analysis (i) with $f(t)=1$, we estimated an indirect effect of $\hat\beta_0=25$ steps (95\% CI $[-3,53]$) and a direct effect of $\hat\alpha_0=72$ steps (95\% CI $[-110,174]$) (see \cref{tab:heartsteps}). This suggests that, when averaged over Weeks 1--5, the activity suggestions had a relatively small but nearly statistically significant indirect effect---a 25-steps increase in Week 6 daily step count---mediated through increases in the subsequent 30-minute step count. In contrast, the direct effect was larger in magnitude (a 72-step increase) but with greater uncertainty.

In Analysis (ii) with $f(t) = (1,t-1)^\top$, the indirect and direct effects are parameterized by $\beta = (\beta_1, \beta_2)^\top$ and $\alpha = (\alpha_1, \alpha_2)^\top$. The estimated effects were $\hat\beta_1 = 73$ (95\% CI $[26, 120]$), $\hat\beta_2 = -0.5$ (95\% CI $[-1.0, -0.1]$), $\hat\alpha_1 = 208$ (95\% CI $[-124, 539]$), and $\hat\alpha_2 = -2.0$ (95\% CI $[-5.0, 0.9]$). These results suggests a statistically significant indirect effect at the beginning of the study that declines over time, and a larger but more uncertain direct effect that also declines. Notably, both effects are strongest early on, despite being temporally distant from the distal outcome in Week 6. This suggests that early interventions had a greater impact on long-term habit formation---both mediated through short-term physical activity and via direct pathways---compared to those delivered later in the study.

\begin{table}[htbp]
    \spacingset{1.5}
    \centering
    \caption{\footnotesize Estimated mediation effects of activity suggestion on final week average daily step count from HeartSteps I MRT. NDEE: natural direct excursion effect; NIEE: natural indirect excursion effect; SE: standard error; CI: confidence interval.}
    \footnotesize
    \label{tab:heartsteps}
    \begin{tabular}{lcccc}
        \toprule
        & Marginal NDEE & Marginal NIEE & Moderated NDEE by $t$ & Moderated NIEE by $t$ \\
        \midrule
        Coefficient & $\alpha_0$ & $\beta_0$ & $\alpha_1$, $\alpha_2$ & $\beta_1$, $\beta_2$ \\
        Estimate & 31.8 & 24.9 & 207.7, $-2.0$ & 73.0, $-0.6$ \\
        SE & 72.4 & 14.3 & 169.2, 1.5 & 23.9, 0.2 \\
        95\% CI & $(-110.0, -173.7)$ & $(-3.2, 53.0)$ & $(-123.9, 539.4)$, $(-5.0, 0.9)$ & $(26.1, 119.9)$, $(-1.0, -0.12)$ \\
        \bottomrule
    \end{tabular}
\end{table}

In Analysis (iii), we estimated non-linear mediation effects using the spline-based model (\cref{fig:heartsteps-analysis}). The indirect effect (NIEE) remained relatively stable and positive until around the 100th decision point (i.e., mid-study). In contrast, the direct effect (NDEE) rose sharply during the first week, peaking near the 25th decision point (Day 5) at an estimated 400-step increase in the distal outcome, then declined rapidly. These findings suggest that early in the study, the impact of activity suggestions on long-term behavior was primarily through direct pathways, independent of immediate activity. Later in the study, the impact shifted toward indirect pathways via short-term activity increases. This pattern is intuitive: early short-term changes are temporally distant from the distal outcome and less likely to persist, whereas later short-term changes are more likely to translate to sustained behavior. Notably, although the direct effect was generally larger, it showed greater variability. The total effect closely followed the direct effect pattern, while the indirect effect, though smaller, has less uncertainty and remained statistically significant over a longer portion of the study.

\begin{figure}[htbp]
    \spacingset{1.5}
    \centering
    \includegraphics[width = 0.7\textwidth]{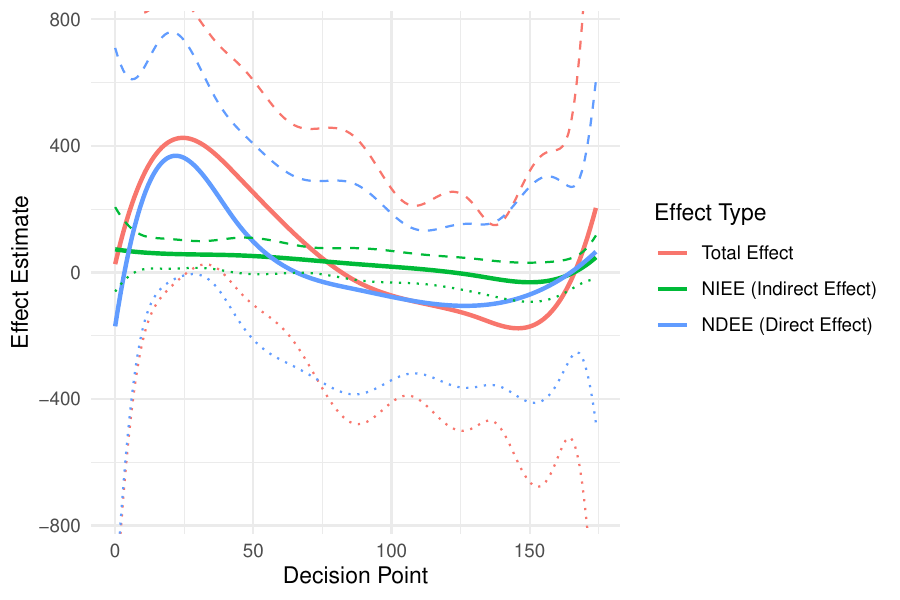}
    \caption{\footnotesize HeartSteps analysis: total effect, NIEE, and NDEE over the $T=175$ decision points. Solid, dashed, and dotted lines in each color represent point estimates, upper and lower bounds of pointwise 95\% confidence interval for the corresponding effect.}
    \label{fig:heartsteps-analysis}
\end{figure}

The mediation analysis findings have implications for designing just-in-time adaptive interventions \citep{nahum2018just}. Strong early direct effects suggest that initial interventions, i.e., those delivered when participants are newly engaged and most receptive, can shape long-term behavior even without immediate activity increases. This supports front-loading more engaging or motivational content to capitalize on heightened receptivity and motivation. In contrast, the sustained indirect effects later in the study suggest that continued delivery of activity suggestions help maintain behavior change through short-term activity boosts. Together, these results suggest a dynamic intervention strategy: intensify early efforts to initiate habit formation, then taper to lighter, well-timed nudges that sustain behavior change.

Lastly, we note that participants continued receiving randomized suggestions in Week 6, which directly influenced the distal outcome. While our analysis remains valid, interpretation of results should account for the ongoing interventions. In other words, had interventions stopped by the end of Week 5, the observed effects on the Week 6 outcomes could differ due to possible interactions between earlier and later interventions.

\section{Observational Study Application: SleepHealth}
\label{sec:sleephealth}

The SleepHealth study was conducted to investigate sleep and daily activity patterns using real-world data collected through wearables and daily self-reported measures \citep{deering2020real}. Participants completed baseline surveys and provided daily self-reported measures on sleep quality and daytime sleepiness, alongside measures collected via Apple HealthKit such as step count. This illustrative analysis examines whether sleep quality affects distal sleepiness, and whether this relationship is mediated by immediate next-day sleepiness.

We restricted our analysis to 280 participants who provided at least six consecutive days of complete data on sleep quality and daytime sleepiness. The distal outcome was daytime sleepiness reported on Day 6. For each of the first five days ($T = 5$), prior-night sleep quality was dichotomized into a treatment indicator $A_t$ (1 if sleep quality score $\geq 4$, indicating ``good'' or ``very good''; 0 otherwise). Daily daytime sleepiness served as a mediator $M_t$, measured on a 1-to-9 scale (1 representing extremely alert and 9 extremely sleepy). Working models for $\hat{p}_t$, $\hat{q}_t$, $\hat\eta_t$, and $\hat\mu_t$ included baseline covariates (demographics, health behaviors, and sleep habits), time-varying daily step counts, and the current $M_t$ in $\hat{q}_t$ and $\hat\mu_t$. Full details on data preprocessing and covariate definitions are provided in Supplementary Material \ref{A-sec:sleephealth}.

Our causal mediation analysis estimated the average-over-five-days direct and indirect excursion effects by setting $f(t)=1$ and $\omega(t) = 1/5$ for $t =1,\ldots,5$. As shown in the first two columns in \cref{tab:sleephealth}, both effects were statistically significant. The estimated indirect effect was $\hat\beta_0 = -0.34$ (95\% CI [–0.48, –0.21]), indicating that better sleep quality reduced distal sleepiness via improvements in next-day sleepiness. The direct effect was similar: $\hat\alpha_0 = -0.31$ (95\% CI [–0.54, –0.08]), suggesting additional pathways beyond immediate next-day sleepiness. The presence of both direct and indirect effects suggest that while short-term sleepiness mediates part of the effect, other mechanisms---such as cumulative sleep, psychological states, or physiological factors---also contribute. A secondary analysis using $f(t) = (1, t - 1)^\top$ found no significant time trends (last two columns in \cref{tab:sleephealth}).

\begin{table}[htbp]
    \spacingset{1.5}
    \centering
    \caption{\footnotesize Estimated mediation effects of sleep quality on daytime sleepiness from SleepHealth Study. NDEE: natural direct excursion effect; NIEE: natural indirect excursion effect; SE: standard error; CI: confidence interval.}
    \footnotesize
    \label{tab:sleephealth}
    \begin{tabular}{lcccc}
        \toprule
        & Marginal NDEE & Marginal NIEE & Moderated NDEE by $t$ & Moderated NIEE by $t$ \\
        \midrule
        Coefficient & $\alpha_0$ & $\beta_0$ & $\alpha_1$, $\alpha_2$ & $\beta_1$, $\beta_2$ \\
        Estimate & $-0.31$ & $-0.34$ & $-0.28$, $-0.01$ & $-0.32$, $-0.01$ \\
        SE       & $0.12$ & $0.07$ & $0.18$, $0.06$ & $0.11$, $0.04$ \\
        95\% CI  & $(-0.54, -0.08)$ & $(-0.48, -0.21)$ & $(-0.65, 0.08)$, $(-0.14, 0.12)$ & $(-0.53, -0.11)$, $(-0.09, 0.06)$ \\
        \bottomrule
    \end{tabular}
\end{table}

\section{Discussion}
\label{sec:discussion}

We studied causal mediation in intensive longitudinal settings with time-varying exposures, time-varying mediators, and a distal outcome, covering both MRTs and observational studies. We introduced the natural (in)direct excursion effects, a novel class of mediation effects that decompose the total effect through the most immediate mediator. The estimands are identifiable under plausible assumptions. We derived the EIF for the mediation functionals and proposed multiply-robust estimators that accommodate flexible machine learning algorithms and optional cross-fitting. In MRTs, the estimators are doubly-robust.

Applying our method to the HeartSteps MRT revealed that early activity suggestions primarily influenced long-term activity through direct pathways, while later effects were increasingly mediated by short-term steps. In the SleepHealth observational study, daily sleep quality impacted long-term alertness both via next-day sleepiness and through other pathways.

Our approach has limitations and points to directions for future work. First, identification of the natural (in)direct excursion effects requires the no intermediate confounding assumption (an implication of \cref{asu:seq-ign-M}). While plausible in MRTs or observational studies where the mediator is measured immediately after treatment, this may be violated when there's a delay or when the mediator is aggregated over a window that includes subsequent treatments, as in some MRTs \citep{battalio2021sense2stop}. In such cases, sensitivity analyses for \cref{asu:seq-ign-M} warrant further exploration. Alternatively, extensions of interventional or path-specific mediation effects to intensive longitudinal settings may yield identifiable estimands under weaker assumptions.

Second, although we decompose the total effect into components mediated through and not through the immediate mediator, we do not further separate more granular pathways within each component. For example, identifying mediation via a mediator $k$ time points into the future would require new methodological development.

Third, because natural (in)direct excursion effects are defined relative to the behavior policy, their interpretation depends on how treatments are assigned in the study. In MRTs, this dependence is often desirable \citep{boruvka2018,dempsey2020stratified}, as it contextualizes the effects within realistic treatment policies. In observational settings, however, this dependence may complicate interpretation if the behavior policy is poorly understood. Future work could examine the sensitivity of these effects to perturbations in the behavior policy.

Fourth, our mediation effects marginalizes over covariates except for time index $t$. Extensions could explore effect modification by time-varying covariates, particularly endogenous ones.

Lastly, extending our approach to accommodate binary or survival distal outcomes and to handle missing data is an important direction for future research.

Code to reproduce all simulations and applications can be downloaded at \url{https://anonymous.4open.science/r/paper_mediationMRT-451F/}.

\section*{Acknowledgement}

The author gratefully acknowledges support from the ICS Research Award at the University of California, Irvine. The author thanks Linda Valeri and Hengrui Cai for their valuable discussions and insights. The author also thanks the support from the coach and group members of the Faculty Success Program by the National Center for Faculty Development and Diversity (NCFDD).

\section*{Supplementary Materials}
\label{sec:supp}

Supplementary Material \ref{A-sec:proof-identification} contains the proof of the identification result (\cref{thm:identification}). Supplementary Material \ref{A-sec:proof-eif} contains the derivation of the EIF (\cref{thm:eif}). Supplementary Material \ref{A-sec:proof-asymptotic} contains the regularity conditions and the proofs for the asymptotic theory (\cref{thm:estimator-consistency,thm:estimator-normality}). Supplementary Material \ref{A-sec:additional_details_and_results_for_simulation} contains additional details and results for the simulation studies (\cref{sec:simulation}). Supplementary Material \ref{A-sec:sleephealth} contains additional details on data preprocessing and covariate definitions for the SleepHealth application (\cref{sec:sleephealth}).

\newpage

\bibliographystyle{agsm}

\bibliography{mhealth-ref}

\newpage

\begin{appendices}

\spacingset{1.9} 

\section{Proof of Identification Result (\texorpdfstring{\cref{thm:identification}}{Theorem 1})}
\label{A-sec:proof-identification}

\subsection{Assumptions for Identification}

Here we restate the identification assumptions.

\begin{asu}[\cref{asu:causal-assumptions} in the main paper]
    \label{A-asu:causal-assumptions}
    \normalfont
    \spacingset{1.5}
    \begin{asulist}
        \item \label{A-asu:consistency} (SUTVA.) There is no interference across participants, and the observed data equals the potential outcome under the observed treatment. Specifically, for each $t$, $X_t(\ba_{t-1}, \bm_{t-1}) = X_t$ if $\bA_{t-1} = \ba_{t-1}$ and $\bM_{t-1} = \bm_{t-1}$, and $M_t(\ba_t, \bm_{t-1}) = M_t$ if $\bA_t = \ba_t$ and $\bM_{t-1} = \bm_{t-1}$. In addition, $Y(\ba_T, \bm_T) = Y$ if $\bA_T = \ba_T$ and $\bM_T = \bm_T$.

        \item \label{A-asu:positivity} (Positivity when eligible.) For each $t$, there exists a positive constant $c > 0$, such that if $P(H_t = h_t, I_t = 1) > 0$ then $c < P(A_t = a \mid H_t = h_t, I_t = 1) < 1-c$ for $a \in \{0, 1\}$. Furthermore, if $P(H_t = h_t, A_t = a) > 0$, then $P(M_t = m \mid H_t = h_t, A_t = a) > 0$ for all $m$ in the support of $M_t$.

        \item \label{A-asu:seq-ign-A} (Sequential ignorability of $A_t$.) For each $t$, the potential outcomes $\big\{M_t(\ba_t)$, $X_{t+1}(\ba_t)$, $M_{t+1}(\ba_{t+1})$, $X_{t+2}(\ba_{t+1})$, $\ldots, X_T(\ba_{T-1})$, $M_T(\ba_T)$, $Y(\ba_T, M_T(\ba_T)): \ba_T \in \{0,1\}^{\otimes T} \big\}$ are conditionally independent of $A_t$ given $H_t$.

        \item \label{A-asu:seq-ign-M} (Sequential ignorability of $M_t$.) For each $t$, for any $a,b \in \{0,1\}$, and for any $m_t$ in the support of $M_t$, we have
        \begin{align*}
			M_t(D_t^b) \perp Y\{D_t^a, \bM_{t-1}(D_t^a), m_t, \uM_{t+1}(D_t^a)\} \mid H_t, A_t = d_t^b
		\end{align*}
    \end{asulist}
\end{asu}

        

\subsection{Lemmas for Identificaftion}

We state and prove a few useful lemmas.

\begin{lem}
	\label{A-lem:causal-consistency}
	Given any $t \in [T]$ and any $a \in \{0,1\}$, if $A_t = d_t^a(H_t)$ then $M_t(D_t^a) = M_t$ and $Y \big\{D_t^a, \bM(D_t^a)\big\} = Y$.
\end{lem}

\begin{proof}[Proof of \cref{A-lem:causal-consistency}]
	By \cref{A-asu:consistency}, $M_t(\bA_t) = M_t$ and $Y \big\{\bA, \bM(\bA)\big\} = Y$. When $A_t = d_t^a(H_t)$, we have $D_t^a = \bA$, and therefore the lemma statement holds. 
\end{proof}

\begin{lem}
	\label{A-lem:id-cond-exp}
	Given any $t \in [T]$, for any $a \neq b \in \{0,1\}$ and any fixed $m_t$ in the support of $M_t(D_t^b)$, we have
	\begin{align*}
		& ~~~~ \EE \Big[ Y \big\{D_t^a, \bM_{t-1}(D_t^a), M_t(D_t^b), \uM_{t+1}(D_t^a) \big\} ~\Big|~ H_t(\bA_{t-1}), M_t(D_t^b) = m_t \Big] \\
		& = \EE(Y \mid H_t, A_t = d_t^a, M_t = m_t).
	\end{align*}
\end{lem}

\begin{proof}[Proof of \cref{A-lem:id-cond-exp}]
	First, note that \cref{A-asu:seq-ign-A} implies that for any $m_t$ and $m_t'$ in the support of $M_t$ and any $a,b\in\{0,1\}$, we have
	\begin{align}
		Y\{D_t^a, \bM_{t-1}(D_t^a), m_t, \uM_{t+1}(D_t^a)\} \perp A_t \mid H_t, M_t(D_t^b) = m_t'. \label{A-eq:id-cond-exp-proofuse0}
	\end{align}

	Now, for any $a,b\in\{0,1\}$, we have
	\begin{align}
		& ~~~~ \EE \Big[ Y \big\{D_t^a, \bM_{t-1}(D_t^a), M_t(D_t^b), \uM_{t+1}(D_t^a) \big\} ~\Big|~ H_t(\bA_{t-1}), M_t(D_t^b) = m_t \Big] \nonumber \\
		& = \EE \Big[ Y \big\{D_t^a, \bM_{t-1}(D_t^a), m_t, \uM_{t+1}(D_t^a) \big\} ~\Big|~ H_t, M_t(D_t^b) = m_t \Big] \label{A-eq:id-cond-exp-proofuse1} \\
		& = \EE \Big[ Y \big\{D_t^a, \bM_{t-1}(D_t^a), m_t, \uM_{t+1}(D_t^a) \big\} ~\Big|~ H_t, A_t = d_t^a, M_t(D_t^b) = m_t \Big] \label{A-eq:id-cond-exp-proofuse2} \\
		& = \EE \Big[ Y \big\{D_t^a, \bM_{t-1}(D_t^a), m_t, \uM_{t+1}(D_t^a) \big\} ~\Big|~ H_t, A_t = d_t^a \Big] \label{A-eq:id-cond-exp-proofuse3} \\
		& = \EE \Big[ Y \big\{D_t^a, \bM_{t-1}(D_t^a), m_t, \uM_{t+1}(D_t^a) \big\} ~\Big|~ H_t, A_t = d_t^a, M_t(D_t^a) = m_t \Big] \label{A-eq:id-cond-exp-proofuse4} \\
		& = \EE \Big[ Y \big\{D_t^a, \bM_{t-1}(D_t^a), M_t(D_t^a), \uM_{t+1}(D_t^a) \big\} ~\Big|~ H_t, A_t = d_t^a, M_t(D_t^a) = m_t \Big] \nonumber \\
		& = \EE(Y \mid H_t, A_t = d_t^a, M_t = m_t). \label{A-eq:id-cond-exp-proofuse5}
	\end{align}
	Here, \cref{A-eq:id-cond-exp-proofuse1} follows from \cref{A-asu:consistency}, \cref{A-eq:id-cond-exp-proofuse2} follows from \cref{A-eq:id-cond-exp-proofuse0}, \cref{A-eq:id-cond-exp-proofuse3,A-eq:id-cond-exp-proofuse4} both follow from \cref{A-asu:seq-ign-M}, and \cref{A-eq:id-cond-exp-proofuse5} follows from \cref{A-asu:consistency}. This completes the proof.
\end{proof}

\begin{lem}
	\label{A-lem:id-cond-exp-equals-ipw}
	Fix any $t \in [T]$ and any $a \neq b \in \{0,1\}$. Let $f(m_t \mid H_t, A_t = d_t^a)$ denote the conditional density of $M_t \mid H_t, A_t = d_t^a$ with respect to a dominating measure; likewise for $f(m_t \mid H_t, A_t = d_t^b)$. We have
	\begin{align}
		& ~~~~ \EE \Big\{ \EE (Y \mid H_t, A_t = d_t^a, M_t) ~\Big|~ H_t, A_t = d_t^b \Big\} \nonumber \\
		& = \EE \bigg\{ \frac{\indic(A_t = d_t^a)}{P(A_t = d_t^a \mid H_t)} \frac{f(M_t \mid H_t, A_t = d_t^b)}{f(M_t \mid H_t, A_t = d_t^a)} ~ Y ~\bigg|~ H_t \bigg\} \label{A-eq:lem:id-cond-exp-equals-ipw-1} \\
		& = \EE \bigg\{ \frac{\indic(A_t = d_t^a)}{P(A_t = d_t^b \mid H_t)} \frac{P(A_t = d_t^b \mid H_t, M_t)}{P(A_t = d_t^a \mid H_t, M_t)} ~ Y ~\bigg|~ H_t \bigg\}. \label{A-eq:lem:id-cond-exp-equals-ipw-2}
	\end{align}
\end{lem}

\begin{proof}[Proof of \cref{A-lem:id-cond-exp-equals-ipw}]
	We have
	\begin{align}
		& ~~~~ \EE \Big\{ \EE (Y \mid H_t, A_t = d_t^a, M_t) ~\Big|~ H_t, A_t = d_t^b \Big\} \nonumber \\
		& = \int \EE (Y \mid H_t, A_t = d_t^a, M_t = m_t) ~ f(m_t \mid H_t, A_t = d_t^b) ~ \dd m_t \nonumber \\
		& = \int \EE (Y \mid H_t, A_t = d_t^a, M_t = m_t) ~ \frac{f(m_t \mid H_t, A_t = d_t^b)}{f(m_t \mid H_t, A_t = d_t^a)} ~ f(m_t \mid H_t, A_t = d_t^a) ~ \dd m_t \nonumber \\
		& = \int \EE \bigg\{ \frac{f(M_t \mid H_t, A_t = d_t^b)}{f(M_t \mid H_t, A_t = d_t^a)} ~ Y ~\bigg|~ H_t, A_t = d_t^a, M_t = m_t \bigg\} ~ f(m_t \mid H_t, A_t = d_t^a) ~ \dd m_t \nonumber \\
		& = \EE \bigg\{ \frac{f(M_t \mid H_t, A_t = d_t^b)}{f(M_t \mid H_t, A_t = d_t^a)} ~ Y ~\bigg|~ H_t, A_t = d_t^a \bigg\} \nonumber \\
		& = \EE \bigg\{ \frac{\indic(A_t = d_t^a)}{P(A_t = d_t^a \mid H_t)} \frac{f(M_t \mid H_t, A_t = d_t^b)}{f(M_t \mid H_t, A_t = d_t^a)} ~ Y ~\bigg|~ H_t \bigg\}, \label{A-eq:lem:id-cond-exp-equals-ipw-proofuse1}
	\end{align}
	where the last equality follows from the law of total probability. This proves \cref{A-eq:lem:id-cond-exp-equals-ipw-1}.

	Furthermore, using Bayes' Theorem we have
	\begin{align}
		f(m_t \mid H_t, A_t = d_t^a) = \frac{f(m_t, A_t = d_t^a \mid H_t)}{P(A_t = d_t^a \mid H_t)} = \frac{P(A_t = d_t^a \mid M_t = m_t, H_t) ~ f(m_t \mid H_t)}{P(A_t = d_t^a \mid H_t)}, \label{A-eq:lem:id-cond-exp-equals-ipw-proofuse2}
	\end{align}
	where $f(m_t, A_t = d_t^a \mid H_t)$ and $f(m_t \mid H_t)$ represent the corresponding conditional densities with respect to a dominating measure. Using \cref{A-eq:lem:id-cond-exp-equals-ipw-proofuse2} and its parallel version for $d_t^b$, and we have
	\begin{align}
		& ~~~~ \EE \bigg\{ \frac{\indic(A_t = d_t^a)}{P(A_t = d_t^a \mid H_t)} \frac{f(M_t \mid H_t, A_t = d_t^b)}{f(M_t \mid H_t, A_t = d_t^a)} ~ Y ~\bigg|~ H_t \bigg\} \nonumber \\
		& = \EE \bigg\{ \frac{\indic(A_t = d_t^a)}{P(A_t = d_t^b \mid H_t)} \frac{P(A_t = d_t^b \mid H_t, M_t)}{P(A_t = d_t^a \mid H_t, M_t)} ~ Y ~\bigg|~ H_t \bigg\}. \label{A-eq:lem:id-cond-exp-equals-ipw-proofuse3}
	\end{align}
	Plugging \cref{A-eq:lem:id-cond-exp-equals-ipw-proofuse3} into \cref{A-eq:lem:id-cond-exp-equals-ipw-proofuse1} yields \cref{A-eq:lem:id-cond-exp-equals-ipw-2}. This completes the proof.
\end{proof}

\subsection{Proof of the Identification Result}

The following is \cref{thm:identification} in the main paper.
\begin{thm}[Identification, \cref{thm:identification} in the main paper]
    \label{A-thm:identification}
    Under \cref{A-asu:causal-assumptions}, $\theta^{aa}(t)$ for $a \in \{0,1\}$ is identified as 
    \begin{align}
        \theta_t^{aa}
        & = \EE \Big\{ \EE\big(Y \mid H_t, A_t = d_t^a\big) \Big\} \label{A-eq:identify-theta-aa-ie} \\
        & = \EE \bigg\{ \frac{\indic(A_t = d_t^a)}{P(A_t = d_t^a \mid H_t)} Y \bigg\}. \label{A-eq:identify-theta-aa-ipw}
    \end{align}
    The mediation functional $\theta_t^{ab}$ for $a,b\in\{0,1\}$ and $a \neq b$ is identified as
    \begin{align}
        \theta_t^{ab} 
        & = \EE\Big[ \EE \big\{ \EE(Y \mid H_t, A_t = d_t^a, M_t) \mid H_t, A_t = d_t^b \big\} \Big] \label{A-eq:identify-theta-ab-ie} \\
        & = \EE \bigg\{ \frac{\indic(A_t = d_t^a)}{P(A_t = d_t^b \mid H_t)} \frac{P(A_t = d_t^b \mid H_t, M_t)}{P(A_t = d_t^a \mid H_t, M_t)} Y \bigg\}. \label{A-eq:identify-theta-ab-ipw}
    \end{align}
\end{thm}

\begin{proof}[Proof of \cref{A-thm:identification}]
	\cref{A-eq:identify-theta-aa-ie,A-eq:identify-theta-aa-ipw} are a special case of Theorem 1 in \citet{qian2025distal}.

	Now we show \cref{A-eq:identify-theta-ab-ie,A-eq:identify-theta-ab-ipw}. We present the proof for the setting where the support of $M_t$ is discrete. For settings where the support of $M_t$ is continuous, the proof is similar with $\sum_{m_t}$ replaced by the integral.

	We have
	\begin{align}
		\theta_t^{ab} & := \EE \Big[ Y \big\{D_t^a, \bM_{t-1}(D_t^a), M_t(D_t^b), \uM_{t+1}(D_t^a) \big\} \Big] \nonumber \\
		& = \EE \Big( \EE \Big[ Y \big\{D_t^a, \bM_{t-1}(D_t^a), M_t(D_t^b), \uM_{t+1}(D_t^a) \big\} ~\Big|~ H_t(\bA_{t-1}) \Big] \Big) \nonumber \\
		& = \EE \bigg( \sum_{m_t} P\big\{M_t (D_t^b) = m_t \mid H_t(\bA_{t-1}) \big\} \nonumber \\
		& ~~~~ \times \EE \Big[ Y \big\{D_t^a, \bM_{t-1}(D_t^a), M_t(D_t^b), \uM_{t+1}(D_t^a) \big\} ~\Big|~ H_t(\bA_{t-1}), M_t(D_t^b) = m_t \Big] \bigg) \nonumber \\
		& = \EE \bigg[ \sum_{m_t} P\big\{M_t (D_t^b) = m_t \mid H_t \big\} ~ \EE (Y \mid H_t, A_t = d_t^a, M_t = m_t) \bigg] \label{A-eq:identification-proofuse1} \\
		& = \EE \bigg[ \sum_{m_t} P\big\{M_t (D_t^b) = m_t \mid H_t, A_t= d_t^b \big\} ~ \EE (Y \mid H_t, A_t = d_t^a, M_t = m_t) \bigg] \label{A-eq:identification-proofuse2} \\
		& = \EE \bigg[ \sum_{m_t} P\big( M_t = m_t \mid H_t, A_t= d_t^b \big) ~ \EE (Y \mid H_t, A_t = d_t^a, M_t = m_t) \bigg] \label{A-eq:identification-proofuse3} \\
		& = \EE \Big[ \EE \big\{ \EE (Y \mid H_t, A_t = d_t^a, M_t = m_t) \mid H_t, A_t = d_t^b \big\} \Big]. \label{A-eq:identification-proofuse4}
	\end{align}
	Here, \cref{A-eq:identification-proofuse1} follows from \cref{A-lem:causal-consistency,A-lem:id-cond-exp}, \cref{A-eq:identification-proofuse2} follows from \cref{A-asu:seq-ign-A}, \cref{A-eq:identification-proofuse3} follows from \cref{A-lem:causal-consistency}, and \cref{A-eq:identification-proofuse4} follows from the definition of conditional expectation.

	Thus, we have shown \cref{A-eq:identify-theta-ab-ie}. \cref{A-eq:identify-theta-ab-ipw} follows directly from \cref{A-lem:id-cond-exp-equals-ipw}. This completes the proof.
\end{proof}

\section{Deriving the Efficient Influence Function (EIF) for $\theta_t^{aa}$ and $\theta_t^{ab}$}
\label{A-sec:proof-eif}

We prove \cref{thm:eif} in this section. We derive the EIF for $\theta_t^{aa}$ in \cref{A-subsec:proof-eif-theta-t-aa} and the EIF for $\theta_t^{aa}$ in \cref{A-subsec:proof-eif-theta-t-ab}. In both derivations, we first construct an influence function (and verify that it is indeed an influence function), then show that the constructed influence function lies in the tangent space (and therefore its projection onto the tangent space is itself), and thus it is the efficient influence function.

\subsection{EIF for $\theta_t^{aa}$}
\label{A-subsec:proof-eif-theta-t-aa}

Fix $t \in [T]$. For notation simplicity, within this subsection we denote $Y_0 := Y(\bA_{t-1}, 0, \uA_{t+1})$, $Y_1 := Y(\bA_{t-1}, 1, \uA_{t+1})$, $A := A_t$, $I := I_t$ and $H := H_t$. By the causal consistency assumption, $Y = Y_0$ if $A_t = 0$ and $Y = Y_1$ if $A_t = 1$. We consider the semiparametric model, $\cM$, that consists of all probability distributions of $(H, A, Y_0, Y_1)$ that satisfy the constraints that $A$ must be 0 if $I = 0$ (recall $I \in H$) and that $A \perp (Y_0, Y_1) \mid H$. Suppose $P^\star \in \cM$ is the truth that generates the data. The following derivation extends the proof of Theorem 1 in \citet{hahn1998role} to causal excursion effects, where we tackle the technical challenges brought by the eligibility indicator $I$.

\subsubsection{Step 1: Derive the score function in a parametric submodel}

Denote a parametric submodel by $\{P_\vartheta: \vartheta \in \Theta\}$ such that there exists $\vartheta_0 \in \Theta$ with $P_{\vartheta_0} = P^\star$. Note that we use $\vartheta$ to parameterize the submodel to be consistent with \citet{hahn1998role}, and $\vartheta$ is distinct from the $\theta$ notation in $\theta_t^{aa}$. The density of $(H, A, Y_0, Y_1)$ (with respect to some $\sigma$-finite measure) in this submodel is given by
\begin{align*}
	f(h; \vartheta) ~ \pi(h; \vartheta)^a ~ \{1-\pi(h; \vartheta)\}^{1-a} ~ f(y_0, y_1 | h; \vartheta),
\end{align*}
where $f(h; \vartheta)$ denotes the marginal density of $H$, $\pi(h; \vartheta) = P_\vartheta(A = 1 | h)$, and $f(y_0, y_1 \mid h; \vartheta)$ denotes the conditional density of $(Y_0, Y_1)$ given $H$, respectively. The density of $(H, A, Y)$ is then equal to
\begin{align*}
	L(\vartheta) = f(h; \vartheta) ~ \{f_1(y | h; \vartheta) ~ \pi(h; \vartheta)\}^a ~ \{f_0(y | h; \vartheta) ~ (1 - \pi(h; \vartheta) )\}^a,
\end{align*}
where $f_1(\cdot | h) := \int f(y_0, \cdot | h) dy_0$ and $f_0(\cdot | h) := \int f(\cdot, y_1 | h) dy_1$. Equivalently, we rewrite the density as (recall $i$ is the dummy argument for $I \equiv I_t$, the eligibility indicator)
\begin{align*}
	L(\vartheta) = \begin{cases}
		f(h; \vartheta) ~ \{f_1(y | h; \vartheta) ~ \pi(h; \vartheta)\}^a ~ \{f_0(y | h; \vartheta) ~ (1 - \pi(h; \vartheta) )\}^a & \text{ if } i = 1,\\
		f(h; \vartheta) ~ f_0(y | h; \vartheta) & \text{ if } i = 0.
	\end{cases}
\end{align*}
The log-likelihood $\ell(\vartheta)$ equals
\begin{align}
	\ell(\vartheta) = \begin{cases}
		\log f(h; \vartheta) + a \{ \log f_1(y | h; \vartheta) + \log \pi(h; \vartheta)\} \\ \qquad + (1-a) \{ \log f_0(y | h; \vartheta) + \log (1 - \pi(h; \vartheta) )\} & \text{ if } i = 1,\\
		\log f(h; \vartheta) + \log f_0(y | h; \vartheta) & \text{ if } i = 0.
	\end{cases} \label{A-eq:eif-aa-proofuse1}
\end{align}

Therefore, when $i = 1$, it follows from \cref{A-eq:eif-aa-proofuse1} that the score function at the truth $\theta = \vartheta_0$ equals
\begin{align}
	S_\vartheta(\vartheta_0) & := \frac{\partial \ell(\vartheta_0)}{\partial \vartheta} \nonumber \\
	& = \frac{\dot{f}(h; \vartheta_0)}{f(h; \vartheta_0)} + a \bigg\{ \frac{\dot{f}_1(y|h; \vartheta_0)}{f_1(y|h; \vartheta_0)} + \frac{\dot\pi(h; \vartheta_0)}{\pi(h; \vartheta_0)} \bigg\} + (1-a) \bigg\{ \frac{\dot{f}_0(y|h; \vartheta_0)}{f_0(y|h; \vartheta_0)} - \frac{\dot\pi(h; \vartheta_0)}{1 - \pi(h; \vartheta_0)} \bigg\} \nonumber \\
	& = t(h) + a ~ s_1(y|h) + (1-a)~s_0(y|h) + \frac{a - \pi(h)}{\pi(h) \{1 - \pi(h)\}} ~ \dot\pi(h), \label{A-eq:eif-aa-proofuse2}
\end{align}
where
\begin{align}
	& t(h) := t(h; \vartheta_0), & & t(h; \vartheta) := \frac{\partial}{\partial \vartheta} \log f(h; \vartheta), \nonumber \\
	& s_1(y|h) := s_1(y|h; \vartheta_0), & & s_1(y|h; \vartheta) := \frac{\partial}{\partial \vartheta} \log f_1(y|h; \vartheta), \nonumber \\
	& s_0(y|h) := s_0(y|h; \vartheta_0), & & s_0(y|h; \vartheta) := \frac{\partial}{\partial \vartheta} \log f_0(y|h; \vartheta), \nonumber \\
	& \pi(h) := \pi(h; \vartheta_0), & & \dot\pi(h; \vartheta) := \frac{\partial}{\partial \vartheta} \pi(h; \vartheta). \label{A-eq:eif-aa-proofuse2.1}
\end{align}
When $i = 0$, it follows from \cref{A-eq:eif-aa-proofuse1} that the score function at the truth $\theta = \vartheta_0$ equals
\begin{align}
	S_\vartheta(\vartheta_0) := \frac{\partial \ell(\vartheta_0)}{\partial \vartheta} = t(h) + s_0(y|h). \label{A-eq:eif-aa-proofuse3}
\end{align}

\cref{A-eq:eif-aa-proofuse2} and \cref{A-eq:eif-aa-proofuse3} imply that
\begin{align}
	S_\vartheta(\vartheta_0) & = t(h) + i \bigg\{ a ~ s_1(y|h) + (1-a)~s_0(y|h) + \frac{a - \pi(h)}{\pi(h) \{1 - \pi(h)\}} ~ \dot\pi(h;\vartheta_0) \bigg\} \nonumber \\
	& ~~~~ + (1-i) s_0(y|h). \label{A-eq:eif-aa-proofuse4}
\end{align}
Note that $t(h), s_1(y|h), s_0(y|h), \dot\pi(h;\vartheta_0)$ all depend on the choice of the specific parametric submodel, whereas $\pi(h) = P(A = 1 | H = h)$ does not depend on the choice of the parametric submodel.

\subsubsection{Step 2: Derive the tangent space}

Let $\cH$ denote the Hilbert space of 1-dimensional mean zero finite-variance measurable funtions of $(H,A,Y)$, equipped with the covariance inner product). (1-dimensional because $\theta_t^{aa}$ is a scalar.) The tangent space $\cT$ is a linear subspace of the $\cH$ defined as the mean-square closure of parametric submodel tangent spaces, the latter being the 1-dimensional subspaces spanned by linear combinations of $S_\vartheta(\vartheta_0)$. Using the same logic as the proof of Theorem 1 in \citet{hahn1998role}, we obtain based on \cref{A-eq:eif-aa-proofuse4} that
\begin{align}
	\cT = \bigg\{ \tilde{t}(H) + I \Big[ A \tilde{s}_1(Y | H) + (1-A) \tilde{s}_0(Y | H) + \{A - \pi(H)\} \tilde{b}(H)\Big] + (1-I) \tilde{s}_0(Y | H) \bigg\}, \label{A-eq:eif-aa-proofuse5}
\end{align}
where $\tilde{t}(\cdot)$, $\tilde{s}_1(\cdot|\cdot)$, and $\tilde{s}_0(\cdot|\cdot)$ satisfy
\begin{align}
	\EE\{\tilde{t}(H)\} = 0, \quad \EE\{\tilde{s}_1(Y(1)|H) \mid H\} = 0, \quad \EE\{\tilde{s}_0(Y(0)|H) \mid H\} = 0, \label{A-eq:eif-aa-proofuse5.1}
\end{align}
$\tilde{b}(\cdot)$ is any square-integrable measurable function of $h$, and $\pi(h) = P(A = 1 \mid H = h)$. Here, $\EE$ and $P$ are with respect to the true $P^\star$.

\subsubsection{Step 3: Prove pathwise differentiability}
\label{A-subsubsec:proof-eif-theta-t-aa-pathwise-diff}

By definition of pathwise differentiability, it suffices to show that there exists a function $\varphi_t^{aa}$ such that for any parametric submodel $\{P_\vartheta: \vartheta \in \Theta\}$ with $P_{\vartheta_0} = P^\star$, 
\begin{align}
	\frac{\partial \theta_t^{aa}(\vartheta_0)}{\partial \vartheta} = \EE \{\varphi_t^{aa} S_\vartheta(\vartheta_0)\}, \label{A-eq:eif-aa-proofuse6}
\end{align}
with $\EE$ taken under $P^\star$. (Notation remark: $\varphi_t^{aa}$ is distinct from but closely related to $\phi_t^{aa}$ in the main paper, as we will demonstrate shortly.)

We first derive $\frac{\partial \theta_t^{11}(\vartheta_0)}{\partial \vartheta}$. Consider the same parametric submodel as in Step 1. By \cref{thm:identification} we have
\begin{align}
	\theta_t^{11}(\vartheta) & = \EE_\vartheta \{ \EE_\vartheta ( Y \mid H, A = I)\} \nonumber \\
	& = \EE_\vartheta \{ I ~ \EE_\vartheta ( Y \mid H, A = 1) + (1-I) ~ \EE_\vartheta ( Y \mid H, A = 0) \} \nonumber \\
	& = \EE_\vartheta \{ I ~ \EE_\vartheta ( Y_1 \mid H) + (1-I) ~ \EE_\vartheta ( Y_0 \mid H) \} \nonumber \\
	& = \int y f_1(y | h; \vartheta) i f(h;\vartheta) dydh + \int y f_0(y | h; \vartheta) (1-i) f(h;\vartheta) dydh. \label{A-eq:eif-aa-proofuse7}
\end{align}
So using the identity $\frac{\partial f}{\partial\vartheta} = \frac{\partial \log f}{\partial \vartheta} f$, we obtain from \cref{A-eq:eif-aa-proofuse7} that
\begin{align}
	\frac{\partial \theta_t^{11}(\vartheta)}{\partial \vartheta} & = \int y s_1(y | h; \vartheta) f_1(y | h; \vartheta) i f(h; \vartheta) dy dh + \int y f_1(y | h; \vartheta) i t(h; \vartheta) f(h; \vartheta) dy dh \nonumber \\
	& ~~~~ + \int y s_0(y | h; \vartheta) f_0(y | h; \vartheta) (1-i) f(h; \vartheta) dy dh + \int y f_0(y | h; \vartheta) (1-i) t(h; \vartheta) f(h; \vartheta) dy dh. \label{A-eq:eif-aa-proofuse8}
\end{align}
Evaluating \cref{A-eq:eif-aa-proofuse8} at $\vartheta = \vartheta_0$ and we obtain
\begin{align}
	\frac{\partial \theta_t^{11}(\vartheta_0)}{\partial \vartheta} & = \EE [Y(1) s_1\{Y(1)|H\} I] + \EE \{ \eta_1(H) I t(H)\} \nonumber \\
	& ~~~~ + \EE [Y(0) s_0\{Y(0)|H\} (1-I)] + \EE \{ \eta_0(H) (1-I) t(H)\}, \label{A-eq:eif-aa-proofuse8.1}
\end{align}
where $s_1(y|h), s_0(y|h)$ and $t(h)$ are defined in \cref{A-eq:eif-aa-proofuse2.1}, $\eta_1(h) := \EE \{Y(1) | H = h\}$, and $\eta_0(h) := \EE \{Y(0) | H = h\}$.

Define
\begin{align}
	\varphi_t^{11} := \frac{\indic(A = I)}{P(A = I | H)}\{Y - \EE(Y | A = I, H)\} + \EE(Y | A = I, H) - \theta_t^{11}. \label{A-eq:eif-aa-proofuse9}
\end{align}
\cref{A-lem:eif-theta-11-pathwise-diff}, which is formally stated and proven later, implies that
\begin{align*}
	\frac{\partial \theta_t^{11}(\vartheta_0)}{\partial \vartheta} = \EE \{\varphi_t^{11} S_\vartheta(\vartheta_0)\}.
\end{align*}
Therefore, $\theta_t^{11}$ is pathwise differentiable.

We now derive $\frac{\partial \theta_t^{00}(\vartheta_0)}{\partial \vartheta}$. Consider the same parametric submodel as in Step 1. By \cref{thm:identification} we have
\begin{align*}
	\theta_t^{00}(\vartheta) = \EE_\vartheta \{ \EE_\vartheta ( Y \mid H, A = 0)\} = \int y f_0(y | h; \vartheta) f(h; \vartheta) dydh.
\end{align*}
So using the identity $\frac{\partial f}{\partial\vartheta} = \frac{\partial \log f}{\partial \vartheta} f$, we have
\begin{align}
	\frac{\partial \theta_t^{00}(\vartheta)}{\partial \vartheta} & = \int y s_0(y|h; \vartheta) f_0(y | h; \vartheta) f(h; \vartheta) dydh + \int y f_0(y | h; \vartheta) t(h; \vartheta) f(h; \vartheta) dydh. \label{A-eq:eif-aa-proofuse9.1}
\end{align}
Evaluating \cref{A-eq:eif-aa-proofuse9.1} at $\vartheta = \vartheta_0$ and we obtain
\begin{align}
	\frac{\partial \theta_t^{00}(\vartheta_0)}{\partial \vartheta} = \EE [Y(0) s_0\{Y(0) | H\}] + \EE \{ \eta_0(H) t(H)\} \label{A-eq:eif-aa-proofuse9.2}.
\end{align}

Define
\begin{align}
	\varphi_t^{00} := \frac{\indic(A = 0)}{P(A = 0 | H)}\{Y - \EE(Y | A = 0, H)\} + \EE(Y | A = 0, H) - \theta_t^{00}. \label{A-eq:eif-aa-proofuse9.3}
\end{align}
\cref{A-lem:eif-theta-00-pathwise-diff}, which is formally stated and proven later, implies that
\begin{align*}
	\frac{\partial \theta_t^{00}(\vartheta_0)}{\partial \vartheta} = \EE \{\varphi_t^{00} S_\vartheta(\vartheta_0)\}.
\end{align*}
Therefore, $\theta_t^{00}$ is pathwise differentiable.

\subsubsection{Step 4: Verify that the pathwise derivative is the EIF}

We first verify that $\varphi_t^{11} \in \cT$. For $\varphi_t^{11}$ defined in \cref{A-eq:eif-aa-proofuse9}, when $I = 1$ we have
\begin{align*}
	\varphi_t^{11} = \frac{A}{\pi(H)}\{Y - \eta_1(H)\} + \eta_1(H) - \theta_t^{11}.
\end{align*}
When $I = 0$, we have $\pi(H) = 0$ and $A = 0$, and thus
\begin{align*}
	\varphi_t^{11} = Y - \theta_t^{11}.
\end{align*}
Therefore, using the identities $I \eta_1(H) = I ~ \EE(Y | A = I, H)$ and $(1-I) \eta_0(H) = (1-I) ~\EE(Y | A = I, H)$, we have
\begin{align*}
	\varphi_t^{11} & = I \frac{A}{\pi(H)}\{Y - \eta_1(H)\} + I ~ \eta_1(H) + (1-I)Y - \theta_t^{11} \\
	& = I A \frac{Y - \eta_1(H)}{\pi(H)} + (1 - I)\{Y - \eta_0(H)\} + \EE(Y | A = I, H) - \theta_t^{11}.
\end{align*}
Therefore, $\varphi_t^{11} \in \cT$ with $\tilde{s}_1(y|h) = \frac{y - \eta_1(h)}{\pi(h)}$, $\tilde{s}_0(y|h) = y - \eta_0(h)$, $\tilde{t}(h) = \EE(Y | A = I, H) - \theta_t^{11}$, and $\tilde{b}(h) = 0$. (We can easily verify that \eqref{A-eq:eif-aa-proofuse5.1} holds for these $\tilde{s}_1$, $\tilde{s}_0$, and $\tilde{t}$. $\EE\{\tilde{t}(H)\} = 0$ is true because of \cref{thm:identification}.) Thus, by \citet[][Section 3.3, Theorem 1 and Definition 2]{bkrw1993}, $\varphi_t^{11}$ is the EIF for $\theta_t^{11}$.

We now verify that $\varphi_t^{00} \in \cT$. For $\varphi_t^{00}$ defined in \cref{A-eq:eif-aa-proofuse9.2}, using the fact that $(1-I)(1-A) = 1-I$, we have
\begin{align*}
	\varphi_t^{00} & = I\frac{1-A}{1 - \pi(H)}\{Y - \eta_0(H)\} + (1-I) \frac{1-A}{1 - \pi(H)}\{Y - \eta_0(H)\} + \eta_0(H) - \theta_t^{00} \\
	& = I\frac{1-A}{1 - \pi(H)}\{Y - \eta_0(H)\} + (1-I) \frac{1}{1 - \pi(H)}\{Y - \eta_0(H)\} + \eta_0(H) - \theta_t^{00}.
\end{align*}
Therefore, $\varphi_t^{00} \in \cT$ with $\tilde{s}_1(y|h) = 0$, $\tilde{s}_0(y|h) = \frac{1}{1 - \pi(h)}\{y - \eta_0(h)\}$, $\tilde{t}(h) = \eta_0(h) - \theta_t^{00}$, and $\tilde{b}(h) = 0$. (We can easily verify that \eqref{A-eq:eif-aa-proofuse5.1} holds for these $\tilde{s}_1$, $\tilde{s}_0$, and $\tilde{t}$. $\EE\{\tilde{t}(H)\} = 0$ is true because of \cref{thm:identification}.) Thus, by \citet[][Section 3.3, Theorem 1 and Definition 2]{bkrw1993}, $\varphi_t^{11}$ is the EIF for $\theta_t^{11}$.

This completes the derivation of the EIFs for $\theta_t^{11}$ and $\theta_t^{00}$.

\subsubsection{Lemmas for deriving the EIF of $\theta_t^{aa}$}

First, we state and prove a few auxiliary lemmas.

\begin{lem}
	\label{A-lem:eif-auxiliary-identities}
	\begin{lemlist}
		\item \label{A-lem:eif-auxiliary-identities-1} For any function $g(H,A,Y)$, we have
		\begin{align*}
			\EE\bigg\{ \frac{\indic(A = I)}{P(A = I | H)} g(H,A,Y) \bigg\} = \EE\Big[ \EE \big\{ g(H,A,Y) \mid H, A = I \big\} \Big].
		\end{align*}
		\item \label{A-lem:eif-auxiliary-identities-2} For any function $g(H,A)$, we have
		\begin{align*}
			\EE\bigg[ \frac{\indic(A = I)}{P(A = I | H)} \Big\{ Y - \EE(Y \mid A = I, H) \Big \}  g(H,A) \bigg] = 0.
		\end{align*}
		\item \label{A-lem:eif-auxiliary-identities-3} For any function $g(H)$, we have $\EE\{s_1(Y|H) g(H) \mid H, A = 1\} = 0$ and $\EE\{s_0(Y|H) g(H) \mid H, A = 0\} = 0$.
	\end{lemlist}
\end{lem}

\begin{proof}[Proof of \cref{A-lem:eif-auxiliary-identities}]
	\cref{A-lem:eif-auxiliary-identities-1} follows from iterated expectation given $H,A$:
	\begin{align*}
		& ~~~~ \EE\bigg\{ \frac{\indic(A = I)}{P(A = I | H)} g(H,A,Y) \bigg\} \\
		& = \EE \bigg[ \EE\bigg\{ \frac{\indic(A = I)}{P(A = I | H)} g(H,A,Y) ~\bigg|~ H,A = I \bigg\} P(A = I | H) \bigg] \\
		& = \EE\Big[ \EE \big\{ g(H,A,Y) \mid H, A = I \big\} \Big].
	\end{align*}
	\cref{A-lem:eif-auxiliary-identities-2} follows from iterated expectation given $H,A$:
	\begin{align*}
		& ~~~~ \EE\bigg[ \frac{\indic(A = I)}{P(A = I | H)} \Big\{ Y - \EE(Y \mid A = I, H) \Big \}  g(H,A) \bigg] \\
		& = \EE \bigg( \EE\bigg[ \frac{\indic(A = I)}{P(A = I | H)} \Big\{ Y - \EE(Y \mid A = I, H) \Big \}  g(H,A) ~\bigg|~ H,A = I \bigg] P(A = I | H) \bigg) \\
		& = \EE\bigg[ \Big\{ Y - \EE(Y \mid A = I, H) ~\Big|~ H, A = I  \Big \}  g(H,A) \bigg] \\
		& = 0.
	\end{align*}
	\cref{A-lem:eif-auxiliary-identities-3} follows from the following: For $a \in \{0,1\}$, we have
	\begin{align*}
		& ~~~~ \EE\{s_a(Y|H) \mid H, A = a\} \\
		& = \EE[s_a\{Y(a)|H\} \mid H, A = a] & \text{(causal consistency)} \\
		& = \EE[s_a\{Y(a)|H\} \mid H] & \text{(sequential ignorability)} \\
		& = 0. & \text{(property of the score function)}
	\end{align*}
	This completes the proof.
\end{proof}

 Now we state and prove the lemmas used in \cref{A-subsubsec:proof-eif-theta-t-aa-pathwise-diff}. 

\begin{lem}
	\label{A-lem:eif-theta-11-pathwise-diff}
	With $\dfrac{\partial \theta_t^{11}(\vartheta_0)}{\partial \vartheta}$ derived in \cref{A-eq:eif-aa-proofuse8.1}, $\varphi_t^{11}$ defined in \cref{A-eq:eif-aa-proofuse9}, and $S_\vartheta(\vartheta_0)$ derived in \cref{A-eq:eif-aa-proofuse4}, we have
	\begin{align*}
		\frac{\partial \theta_t^{11}(\vartheta_0)}{\partial \vartheta} = \EE \{\varphi_t^{11} S_\vartheta(\vartheta_0)\}.
	\end{align*}
\end{lem}

\begin{proof}[Proof of \cref{A-lem:eif-theta-11-pathwise-diff}]
	Let
	\begin{align*}
		\cnone & := \frac{\indic(A = I)}{P(A = I | H)}\{Y - \EE(Y | A = I, H)\}, \\
		\cntwo & := \EE(Y | A = I, H), \\
		\cnthree & := \theta_t^{11}, \\
		\cna & := t(H), \\
		\cnb & := I A s_1(Y|H), \\
		\cnc & := I(1-A)s_0(Y|H), \\
		\cnd & := I \frac{A - \pi(H)}{\pi(H) \{1 - \pi(H)\}} ~ \dot\pi(H;\vartheta_0), \\
		\cne & := (1-I) s_0(Y|H),
	\end{align*}
	then we have $\varphi_t^{11} = \cnone + \cntwo + \cnthree$ and $S_\vartheta(\vartheta_0) = \cna + \cnb + \cnc + \cnd + \cne$, so it suffices to show that
	\begin{align*}
		\EE\Big\{\Big(\cnone + \cntwo + \cnthree\Big) \Big(\cna + \cnb + \cnc + \cnd + \cne\Big)\Big\} = \dfrac{\partial \theta_t^{11}(\vartheta_0)}{\partial \vartheta}	
	\end{align*}
	with the latter given in \cref{A-eq:eif-aa-proofuse8.1}.

	We have $\EE\big(\cnone \cna\big) = 0$ by \cref{A-lem:eif-auxiliary-identities-2}.

	For $\EE\big(\cnone \cnb\big)$, we have:
	\begin{align*}
		& ~~~~ \EE\big(\cnone \cnb\big) \\
		& = \EE\bigg[ \frac{\indic(A = I)}{P(A = I | H)}\{Y - \EE(Y | A = I, H)\} I A s_1(Y|H) \bigg] \\
		& = \EE\bigg( \EE \Big[ \{Y - \EE(Y | A = I, H)\} I A s_1(Y|H) ~\Big|~ H, A = I \Big] \bigg) & \text{(due to \cref{A-lem:eif-auxiliary-identities-1})} \\
		& = \EE [ \EE \{Y I A s_1(Y|H) \mid H, A = I \} ] \\
		& ~~~~ - \EE [ \EE \{\EE(Y | A = I, H) I A s_1(Y|H) \mid H, A = I \} ] \\
		& = \EE [ I \EE \{Y s_1(Y|H) \mid H, A = 1 \} ] - 0 & \text{(due to \cref{A-lem:eif-auxiliary-identities-3})} \\
		& = \EE [ I \EE \{Y(1) s_1\{Y(1)|H\} \mid H] & \text{(consistency, ignorability)} \\
		& = \EE [I Y(1) s_1\{Y(1)|H\}].  & \text{(iterated expectation)}
	\end{align*}

	We have $\EE\big(\cnone \cnc\big) = 0$ because $\indic(A = I) I (1-A) \equiv 0$.

	We have $\EE\big(\cnone \cnd\big) = 0$ by \cref{A-lem:eif-auxiliary-identities-2}.

	For $\EE\big(\cnone \cne\big)$, we have:
	\begin{align*}
		& ~~~~ \EE\big(\cnone \cne\big) \\
		& = \EE\bigg[ \frac{\indic(A = I)}{P(A = I | H)}\{Y - \EE(Y | A = I, H)\} (1-I) s_0(Y|H) \bigg] \\
		& = \EE\bigg( \EE \Big[ \{Y - \EE(Y | A = I, H)\} (1-I) s_0(Y|H) ~\Big|~ H, A = I \Big] \bigg) & \text{(due to \cref{A-lem:eif-auxiliary-identities-1})} \\
		& = \EE [ (1-I) \EE \{Y s_0(Y|H) \mid H, A = 0 \} ] \\
		& ~~~~ - \EE [ (1-I) \EE \{\EE(Y | A = I, H) s_0(Y|H) \mid H, A = 0 \} ] \\
		& = \EE [ (1-I) \EE \{Y s_0(Y|H) \mid H, A = 0 \} ] - 0 & \text{(due to \cref{A-lem:eif-auxiliary-identities-3})} \\
		& = \EE [ (1-I) \EE \{Y(0) s_0\{Y(0)|H\} \mid H] & \text{(consistency, ignorability)} \\
		& = \EE [(1-I) Y(0) s_0\{Y(0)|H\}].  & \text{(iterated expectation)}
	\end{align*}

	For $\EE\big(\cntwo \cna\big)$, we have:
	\begin{align*}
		& ~~~~ \EE\big(\cntwo \cna\big) \\
		& = \EE \{ \EE(Y \mid A = I, H) t(H) \} \\
		& = \EE \{ I \EE(Y \mid A = I, H) t(H) \} + \EE \{ (1-I) \EE(Y \mid A = I, H) t(H) \} \\
		& = \EE \{ I \eta_1(H) t(H) \} + \EE \{ (1-I) \eta_0(H) t(H) \},
	\end{align*}
	where the last step used the identities $I \eta_1(H) = I ~ \EE(Y | A = I, H)$ and $(1-I) \eta_0(H) = (1-I) ~\EE(Y | A = I, H)$.

	For $\EE\big(\cntwo \cnb\big)$, $\EE\big(\cntwo \cnc\big)$, $\EE\big(\cntwo \cnd\big)$, and $\EE\big(\cntwo \cne\big)$, they are all equal to 0, because $\cntwo$ is a function of $H$ alone, \cref{A-lem:eif-auxiliary-identities-3}, and the fact that $\EE\{A - \pi(H) \mid H \} = 0$.

	For $\EE\big(\cnthree \cna\big)$, $\EE\big(\cnthree \cnb\big)$, $\EE\big(\cnthree \cnc\big)$, $\EE\big(\cnthree \cnd\big)$, and $\EE\big(\cnthree \cne\big)$, they are all equal to 0, because $\cnthree$ is a constant, \cref{A-lem:eif-auxiliary-identities-3}, the fact that $\EE\{A - \pi(H) \mid H \} = 0$, and the fact that $\EE\{t(H)\} = 0$.

	Putting together all the calculation results gives exactly $\dfrac{\partial \theta_t^{11}(\vartheta_0)}{\partial \vartheta}$ given in \cref{A-eq:eif-aa-proofuse8.1}, and thus the proof is completed.
\end{proof}

\begin{lem}
	\label{A-lem:eif-theta-00-pathwise-diff}
	With $\dfrac{\partial \theta_t^{00}(\vartheta_0)}{\partial \vartheta}$ derived in \cref{A-eq:eif-aa-proofuse9.2}, $\varphi_t^{00}$ defined in \cref{A-eq:eif-aa-proofuse9.2}, and $S_\vartheta(\vartheta_0)$ derived in \cref{A-eq:eif-aa-proofuse4}, we have
	\begin{align*}
		\frac{\partial \theta_t^{00}(\vartheta_0)}{\partial \vartheta} = \EE \{\varphi_t^{00} S_\vartheta(\vartheta_0)\}.
	\end{align*}
\end{lem}

\begin{proof}[Proof of \cref{A-lem:eif-theta-00-pathwise-diff}]
	The proof is similar to (and simpler than) the proof of \cref{A-lem:eif-theta-11-pathwise-diff}, and is thus omitted.
\end{proof}

\subsection{EIF for $\theta_t^{ab}$ ($a \neq b$)}
\label{A-subsec:proof-eif-theta-t-ab}

The derivation of the EIF for $\theta_t^{ab}$ follows in a similar way as \cref{A-subsec:proof-eif-theta-t-aa}. We provide a sketch of the proof below.

First, we show that the function $\varphi_t^{ab}$, defined as follows, is the EIF for $\theta_t^{ab}$:
\begin{align}
	\varphi_t^{ab} & = \frac{\indic(A_t = d_t^a) ~ f_t(M_t | A_t = d_t^b, H_t) }{p_t(a | H_t) ~ f_t(M_t | A_t = d_t^a, H_t)} \Big\{ Y - \mu_t (a, H_t, M_t) \Big\} \nonumber \\
    & ~~~~ + \frac{\indic(A_t = d_t^b)}{p_t(b | H_t)} \Big\{\mu_t (a, H_t, M_t) - \nu_t(a, H_t)\Big\} + \nu_t(a, H_t), \label{A-eq:eif-ab-proofuse1}
\end{align}
where $f_t(M_t | A_t, H_t)$ is the conditional density of $M_t$ given $A_t, H_t$. (Note that $\varphi_t^{ab}$ is distinct from $\phi_t^{ab}$.) The fact that $\varphi_t^{ab}$ defined in \cref{A-eq:eif-ab-proofuse1} is the EIF can be shown in a similar fashion as the proof of Theorem 1 in \citet{tchetgen2012semiparametric}. The unique challenge beyond \citet{tchetgen2012semiparametric} is the handling of the eligibility indicator $I_t$, which can be handled in a similar way as \cref{A-subsec:proof-eif-theta-t-aa}.

Second, we show that $\phi_t^{ab} - \theta_t^{ab}$ with $\phi_t^{ab}$ defined in \cref{eq:def-phi-ab} of the main paper is equivalent to \cref{A-eq:eif-ab-proofuse1}. This is proven using Bayes theorem in a way similar to Lemma 3.2 in \citet{farbmacher2022causal}.

\section{Proofs for Asymptotic Theory and Multiple Robustness}
\label{A-sec:proof-asymptotic}

\subsection{Theorem Statement}
\label{A-subsec:asymptotic-theorem-statement}

Let $\gamma := (\alpha,\beta)$ denote the parameter of interest. Let $\zeta_t := (p_t, q_t, \eta_t, \mu_t, \nu_t)$ denote the nuisance functions corresponding to decision point $t$, and $\zeta := (\zeta_1, \zeta_2, \ldots, \zeta_T)$. Recall that the estimating function $\psi(\gamma, \zeta)$is defined as follows:
\begin{align}
	\psi(\gamma, \zeta) := \sum_{t=1}^T \omega(t)
	\left[\begin{matrix}
		\big\{ \phi_t^{10}(p_t, q_t, \mu_t, \nu_t) - \phi_t^{00}(p_t, \eta_t) - f(t)^T\alpha \big\} f(t) \\
		\big\{ \phi_t^{11}(p_t, \eta_t) - \phi_t^{10}(p_t, q_t, \mu_t, \nu_t) - f(t)^T\beta \big\} f(t)
	\end{matrix}\right]. \label{A-eq:def-psi}
\end{align}

We restate \cref{thm:estimator-normality} (asymptotic normality) separately for $\hat\gamma$ (the estimator without cross-fitting in Algorithm \ref{algo:estimator-ncf}) and $\tilde\gamma$ (the estimator with $K$-fold cross-fitting in Algorithm \ref{algo:estimator-cf}). We state explicitly the regularity conditions required for each estimator. \cref{thm:estimator-consistency} (consistency) is a by-product that will be established in \cref{A-subsec:asymptotic-proof-ncf}.

\begin{thm}[\cref{thm:estimator-normality} in main paper, $\hat\gamma$ part]
    \label{A-thm:rate-mr-ncf}
    \spacingset{1.5}
    Suppose \cref{asu:causal-assumptions} hold and consider $\gamma^\star = ((\alpha^\star)^T, (\beta^\star)^T)^T$ defined in \eqref{eq:alphabeta-star-def}. For each $t \in [T]$, suppose that the fitted nuisance function $\hat\zeta_t = (\hat{p}_t, \hat{q}_t, \hat\eta_t, \hat\mu_t, \hat\nu_t)$ converges in $L_2$ to limit $\zeta_t' = (p_t', q_t', \eta_t', \mu_t', \nu_t')$. Let $\zeta := (\zeta_1, \ldots, \zeta_T)$; similarly define $\hat\zeta$ and $\zeta'$.
    Suppose that for each $t \in [T]$, the nuisance function estimates satisfy the following rates:
    \begin{align}
    	\|\hat{p}_t - p_t^\star\| \cdot \|\hat\eta_t - \eta_t^\star\| = o_P(n^{-1/2}) \label{A-eq:nuisance-product-rate1} \\
    	\|\hat{p}_t - p_t^\star\| \cdot \|\hat\nu_t - \nu_t^\star\| = o_P(n^{-1/2}) \label{A-eq:nuisance-product-rate2} \\
    	\|\hat{q}_t - q_t^\star\| \cdot \|\hat\mu_t - \mu_t^\star\| = o_P(n^{-1/2}) \label{A-eq:nuisance-product-rate3}
    \end{align}
    Furthremore, suppose the following regularity conditions hold: \cref{A-asu:reg-general,A-asu:nuisance-conv-general,A-asu:donsker}. Then we have
    \begin{align}
        \sqrt{n} (\hat\gamma - \gamma^\star) \dto N(0, V) \text{ as } n\to\infty, \nonumber
    \end{align}
    with
    \begin{align}
    	V := \EE\{ \partial_\gamma \psi(\gamma^\star, \zeta') \}^{-1} \EE\{ \psi(\gamma^\star, \zeta') \psi(\gamma^\star, \zeta')^T\} \EE\{ \partial_\gamma \psi(\gamma^\star, \zeta') \}^{-1, T}. \label{A-eq:def-avar-V}
    \end{align}
    $V$ can be consistently estimated by $\PP_n\{ \partial_\gamma \psi(\hat\gamma, \hat\zeta) \}^{-1} \PP_n\{ \psi(\hat\gamma, \hat\zeta) \psi(\hat\gamma, \hat\zeta)^T\} \PP_n\{ \partial_\gamma \psi(\hat\gamma, \hat\zeta) \}^{-1, T}$.
\end{thm}

\begin{thm}[\cref{thm:estimator-normality} in main paper, $\tilde\gamma$ part]
    \label{A-thm:rate-mr-cf}
    \spacingset{1.5}
    Suppose \cref{asu:causal-assumptions} hold and consider $\gamma^\star = ((\alpha^\star)^T, (\beta^\star)^T)^T$ defined in \eqref{eq:alphabeta-star-def}. For each $t \in [T]$, suppose that $\hat\zeta_{kt} = (\hat{p}_{kt}, \hat{q}_{kt}, \hat\eta_{kt}, \hat\mu_{kt}, \hat\nu_{kt})$, the fitted nuisance functions using $B_k^c$, converges in $L_2$ to same limit $\zeta_t'$. Let $\zeta := (\zeta_1, \ldots, \zeta_T)$; similarly define $\hat\zeta_k$ and $\zeta'$.
    Suppose that for each $t \in [T]$, the nuisance function estimates satisfy the following rates for all $k \in [K]$:
    \begin{align}
    	\|\hat{p}_{kt} - p_t^\star\| \cdot \|\hat\eta_{kt} - \eta_t^\star\| = o_P(n^{-1/2}) \label{A-eq:nuisance-product-rate1-cf} \\
    	\|\hat{p}_{kt} - p_t^\star\| \cdot \|\hat\nu_{kt} - \nu_t^\star\| = o_P(n^{-1/2}) \label{A-eq:nuisance-product-rate2-cf} \\
    	\|\hat{q}_{kt} - q_t^\star\| \cdot \|\hat\mu_{kt} - \mu_t^\star\| = o_P(n^{-1/2}) \label{A-eq:nuisance-product-rate3-cf}
    \end{align}
    Furthremore, suppose the following regularity conditions hold: \cref{A-asu:reg-general,A-asu:nuisance-conv-general-cf,A-asu:reg-bounded-ee-deriv-and-meat}. Then we have
    \begin{align}
        \sqrt{n} (\tilde\gamma - \gamma^\star) \dto N(0, V) \text{ as } n\to\infty, \nonumber
    \end{align}
    with $V$ defined in \cref{A-eq:def-avar-V}. $V$ can be consistently estimated by 
	\begin{align*}
	    \bigg[\frac{1}{K}\sum_{k=1}^K \PP_{n,k}\big\{ \partial_\gamma \psi(\hat\gamma, \hat\zeta_k) \big\} \bigg]^{-1} \bigg[\frac{1}{K}\sum_{k=1}^K \PP_{n,k} \big\{ \psi(\hat\gamma, \hat\zeta_k) \psi(\hat\gamma, \hat\zeta_k)^T \big\} \bigg]
	    \bigg[\frac{1}{K}\sum_{k=1}^K \PP_{n,k}\big\{ \partial_\gamma \psi(\hat\gamma, \hat\zeta_k) \big\} \bigg]^{-1,T},
	\end{align*}
	where $\PP_{n,k}$ denotes the empirical average over observations from the $k$-th partition $B_k$.
\end{thm}

\subsection{Regularity Conditions}
\label{A-subsec:asymptotic-regularity-conditions}

We present the regularity conditions in establishing asymptotic normality. Here we use $\zeta'$ to denote the $L_2$ limit of the nuisance parameters, and $\cZ$ to denote the space of all possible $\zeta$.

\begin{asu}[Regularity conditions]
	\label{A-asu:reg-general}
	~
	\begin{asulist}
		\item \label{A-asu:unique-zero} Suppose that for fixed $\zeta$, $\PP\{\psi(\gamma, \zeta)\} = 0$ has a unique solution in terms of $\gamma$.
		\item \label{A-asu:reg-compact-param-space} Suppose the parameter space $\Theta$ of $\gamma$ is compact.
		\item \label{A-asu:reg-bounded-obs} Suppose the support of $O_i$ is bounded.
		\item \label{A-asu:reg-cont-PPee} Suppose $\PP\{\psi(\gamma, \zeta')\}$ is a continuous function in $\gamma$.
		\item \label{A-asu:reg-bounded-and-cont-differentiable-ee} Suppose $\psi(\gamma,\zeta)$ is continuously differentiable in $\gamma$, and the class $\{\psi(\gamma,\zeta): \gamma \in \Theta, \zeta \in \cZ\}$ is uniformly bounded and bounded by an integrable function.
		\item \label{A-asu:reg-dominated-ee-deriv} Suppose $\partial_\gamma \psi(\gamma,\zeta) := \frac{\partial \psi(\gamma,\zeta)}{\partial\gamma^T}$ is bounded by an integrable function.
		\item \label{A-asu:reg-dominated-ee-meat} Suppose $\psi(\gamma,\zeta) \psi(\gamma,\zeta)^T$ is bounded by an integrable function.
		\item \label{A-asu:reg-invertible-ee-deriv} Suppose $\PP\{\partial_\gamma \psi(\gamma^\star,\zeta')\}$ is invertible.
	\end{asulist}
\end{asu}

We note that the regularity conditions in \cref{A-asu:reg-general} can all be verified straightforwardly using the form of $\psi$ in \cref{A-eq:def-psi}. We keep them as assumptions instead of proven statements, so that the technical theory developed below can be potentially applicable to other estimation procedures.

The following are additional assumptions needed for establishing the asymptotics for the non-cross-fitted estimator $\hat\gamma$.

\begin{asu}[Convergence of nuisance parameter estimator]
	\label{A-asu:nuisance-conv-general}
	Suppose there exists $\zeta' \in \cZ$ such that the following hold.
	\begin{asulist}
		\item \label{A-asu:nuisance-conv-PPee-sup} $\sup_{\gamma \in \Theta} | \PP \psi(\gamma, \hat\zeta) - \PP \psi(\gamma, \zeta') | = o_P(1)$;
		\item \label{A-asu:nuisance-conv-ee-l2} $\| \psi(\gamma^\star, \hat\zeta) - \psi(\gamma^\star, \zeta') \|^2  := \int |\psi(\gamma^\star, \hat\zeta) - \psi(\gamma^\star, \zeta')|^2 dP = o_P(1)$;
		\item \label{A-asu:nuisance-conv-PPee-deriv} $| \PP \{\partial_\gamma \psi(\gamma^\star, \hat\zeta)\} - \PP \{\partial_\gamma \psi(\gamma^\star, \zeta')\} | = o_P(1)$;
		\item \label{A-asu:nuisance-conv-PPee-meat} $| \PP \{\psi(\gamma^\star, \hat\zeta) \psi(\gamma^\star, \hat\zeta)^T\} - \PP \{\psi(\gamma^\star, \zeta') \psi(\gamma^\star, \zeta')^T \} | = o_P(1)$.
	\end{asulist}
\end{asu}

\begin{asu}[Donsker condition]
	\label{A-asu:donsker}
	Suppose $\Psi := \{\psi(\gamma, \zeta): \gamma \in \Theta, \zeta \in \cZ\}$ and $\{\partial_\gamma \psi(\gamma,\zeta): \gamma \in \Theta, \zeta \in \cZ\}$ are $P$-Donsker classes.
\end{asu}

The following are additional assumptions needed for establishing the asymptotics for the cross-fitted estimator $\check\gamma$.

\begin{asu}[Convergence of nuisance parameter estimator (cross-fitting)]
	\label{A-asu:nuisance-conv-general-cf}
	Suppose there exists $\zeta' \in \cZ$ such that the following hold.
	\begin{asulist}
		\item \label{A-asu:nuisance-conv-PPee-sup-cf} For each $k\in[K]$, $\sup_{\gamma \in \Theta} | \PP \psi(\gamma, \hat\zeta_k) - \PP \psi(\gamma, \zeta') | = o_P(1)$;
		\item \label{A-asu:nuisance-conv-ee-l2-cf} For each $k\in[K]$, $ \| \psi(\gamma^\star, \hat\zeta_k)  - \psi(\gamma^\star, \zeta') \|^2  = o_P(1)$.
		\item \label{A-asu:nuisance-conv-PPee-deriv-cf} For each $k\in[K]$, $| \PP \{\partial_\gamma \psi(\gamma^\star, \hat\zeta_k)\} - \PP \{\partial_\gamma \psi(\gamma^\star, \zeta')\} | = o_P(1)$;
		\item \label{A-asu:nuisance-conv-PPee-meat-cf} For each $k\in[K]$, $| \PP \{\psi(\gamma^\star, \hat\zeta_k) \psi(\gamma^\star, \hat\zeta_k)^T\} - \PP \{\psi(\gamma^\star, \zeta') \psi(\gamma^\star, \zeta')^T \} | = o_P(1)$.
	\end{asulist}
\end{asu}

\begin{asu}[Additional regularity conditions (cross-fitting)]
	\label{A-asu:reg-bounded-ee-deriv-and-meat}
	Suppose $\partial_\gamma \psi(\gamma,\zeta) := \frac{\partial \psi(\gamma,\zeta)}{\partial\gamma^T}$ and $\psi(\gamma,\zeta) \psi(\gamma,\zeta)^T$ are uniformly bounded.
\end{asu}

\subsection{Lemmas on Multiple Robustness}
\label{A-subsec:lemmas-robustness}

We state and prove lemmas that elucidates the robustness properties of the estimating functions $\psi$ defined in \cref{A-eq:def-psi}.

\begin{lem}[$\phi_t^{aa}$ is doubly robust]
	\label{A-lem:phi-aa-dr}
	Fix any $t \in [T]$ and $a \in \{0,1\}$. Consider $\phi^{aa}_t$ defined in \cref{eq:def-phi-aa}, whose sample average is used to estimate $\theta_t^{aa}$ in Algorithm \ref{algo:estimator-ncf}:
	\begin{align*}
	    \phi_t^{aa}(p_t, \eta_t) := \frac{\indic(A_t = d_t^a)}{p_t(a | H_t)} Y - \frac{\indic(A_t = d_t^a) - p_t(a | H_t)}{p_t(a | H_t)} \eta_t(a, H_t).
	\end{align*}
	Suppose either $p_t = p_t^\star$ or $\eta_t = \eta_t^\star$, then $\PP \{\phi_t^{aa}(p_t, \eta_t)\} = \theta_t^{aa}$.
\end{lem}

\begin{proof}[Proof of \cref{A-lem:phi-aa-dr}]
	When $p_t = p_t^\star$, for any $\eta_t$, we use the identification of $\theta_t^{aa}$ in \cref{eq:identify-theta-aa-ipw} to get
	\begin{align*}
		\EE \{\phi_t^{aa}(p_t^\star, \eta_t)\} & = \EE \bigg\{\frac{\indic(A_t = d_t^a)}{p_t^\star(a | H_t)} Y \bigg\} - \EE \bigg[ \EE \bigg\{ \frac{\indic(A_t = d_t^a) - p_t^\star(a | H_t)}{p_t^\star(a | H_t)} ~\bigg|~ H_t \bigg\} \eta_t(a,H_t) \bigg] \\
		& = \theta_t^{aa} - 0 = \theta_t^{aa}.
	\end{align*}

	When $\eta_t = \eta_t^\star$, for any $p_t$, we use the identification of $\theta_t^{aa}$ in \cref{eq:identify-theta-aa-ie} to get
	\begin{align*}
		\EE \{\phi_t^{aa}(p_t, \eta_t^\star)\} & = \EE \bigg[ \frac{\indic(A_t = d_t^a)}{p_t(a | H_t)} \{Y - \eta_t^\star(a,H_t)\} \bigg] + \EE\{\eta_t^\star(a,H_t)\} \\
		& = 0 + \theta_t^{aa} = \theta_t^{aa}.
	\end{align*}

	This completes the proof.
\end{proof}

\bigskip

\begin{lem}[$\phi_t^{aa}$ is rate doubly robust]
	\label{A-lem:phi-aa-rate-dr}
	Fix any $t \in [T]$ and $a \in \{0,1\}$. Consider $\phi^{aa}_t$ defined in \cref{eq:def-phi-aa}. Suppose either $p_t' = p_t^\star$ or $\eta_t' = \eta_t^\star$, then
	\begin{align*}
		\Big|\PP \big\{\phi_t^{aa}(\hat{p}_t, \hat\eta_t) - \phi_t^{aa}(p_t', \eta_t') \big\} \Big| \lesssim \|\hat{p}_t - p_t^\star\| \cdot \|\hat\eta_t - \eta_t^\star\|.
	\end{align*}
\end{lem}

\begin{proof}[Proof of \cref{A-lem:phi-aa-rate-dr}]
	Using the law of iterated expectation, for any $p_t$ and $\eta_t$ we have
	\begin{align}
		\PP \bigg\{ \frac{\indic(A_t = d_t^a)}{p_t(a | H_t)}Y \bigg\} = \PP \bigg\{ \frac{\indic(A_t = d_t^a)}{p_t(a | H_t)} \eta_t^\star(a, H_t) \bigg\} = \PP \bigg\{ \frac{p_t^\star(a | H_t)}{p_t(a | H_t)} \eta_t^\star(a, H_t) \bigg\}, \label{A-eq:lem:phi-aa-rate-dr:proofuse1}
	\end{align}
	and
	\begin{align}
		\PP \bigg\{ \frac{\indic(A_t = d_t^a)}{p_t(a | H_t)} \eta(a, H_t) \bigg\} = \PP \bigg\{ \frac{p_t^\star(a | H_t)}{p_t(a | H_t)} \eta_t(a, H_t) \bigg\}. \label{A-eq:lem:phi-aa-rate-dr:proofuse2}
	\end{align}

	In the following, to reduce clutter, within this proof we use the shorthand notation: $1_a := \indic(A_t = d_t^a)$, $p_a := p_t(a | H_t)$, $\eta_a := \eta_t(a, H_t)$. Using this notation, we have
	\begin{align*}
		\phi_t^{aa}(p_t, \eta_t) = \frac{1_a}{p_a}Y - \frac{1_a}{p_a}\eta_a + \eta_a.
	\end{align*}
	Therefore,
	\begin{align}
		& ~~~~ \PP \big\{\phi_t^{aa}(\hat{p}_t, \hat\eta_t) - \phi_t^{aa}(p_t', \eta_t') \big\} \nonumber \\
		& = \PP \bigg( \frac{1_a}{\hat{p}_a}Y - \frac{1_a}{p_a'}Y - \frac{1_a}{\hat{p}_a} \hat\eta_a + \frac{1_a}{p_a'}\eta_a' + \hat\eta_a - \eta_a' \bigg) \nonumber \\
		& = \PP \bigg( \frac{p_a^\star}{\hat{p}_a}\eta_a^\star - \frac{p_a^\star}{p_a'}\eta_a^\star - \frac{p_a^\star}{\hat{p}_a} \hat\eta_a + \frac{p_a^\star}{p_a'}\eta_a' + \hat\eta_a - \eta_a' \bigg) \label{A-eq:lem:phi-aa-rate-dr:proofuse3} \\
		& = \PP \bigg\{ \frac{1}{\hat{p}_a p_a'} \Big( p_a' p_a^\star \eta_a^\star - \hat{p}_a p_a^\star \eta_a^\star - p_a' p_a^\star \hat\eta_a + \hat{p}_a p_a^\star \eta_a' + \hat{p}_a p_a' \hat\eta_a - \hat{p}_a p_a' \eta_a' \Big) \bigg\}, \label{A-eq:lem:phi-aa-rate-dr:proofuse4}
	\end{align}
	where \cref{A-eq:lem:phi-aa-rate-dr:proofuse3} follows from \cref{A-eq:lem:phi-aa-rate-dr:proofuse1,A-eq:lem:phi-aa-rate-dr:proofuse2}.

	When $p_a' = p_a^\star$, the sum inside the parentheses in \cref{A-eq:lem:phi-aa-rate-dr:proofuse4} becomes
	\begin{align*}
		(p_a^\star)^2\eta_a^\star - \hat{p}_a p_a^\star \eta_a^\star - (p_a^\star)^2 \hat\eta_a + \hat{p}_a p_a^\star \eta_a' + \hat{p}_a p_a^\star \hat\eta_a - \hat{p}_a p_a^\star \eta_a' = p_a^\star(p_a^\star - \hat{p}_a)(\eta_a^\star - \hat\eta_a).
	\end{align*}
	When $\eta_a' = \eta_a^\star$, the sum inside the parentheses in \cref{A-eq:lem:phi-aa-rate-dr:proofuse4} becomes
	\begin{align*}
		p_a' p_a^\star \eta_a^\star - \hat{p}_a p_a^\star \eta_a^\star - p_a' p_a^\star \hat\eta_a + \hat{p}_a p_a^\star \eta_a^\star + \hat{p}_a p_a' \hat\eta_a - \hat{p}_a p_a' \eta_a^\star = p_a'(p_a^\star - \hat{p}_a)(\eta_a^\star - \hat\eta_a).
	\end{align*}
	In both cases, \cref{A-eq:lem:phi-aa-rate-dr:proofuse4} is $\lesssim \|\hat{p}_t - p_t^\star\| \cdot \|\hat\eta_t - \eta_t^\star\|$. This completes the proof.
\end{proof}

\bigskip

\begin{lem}[Useful equalities for establishing multiple robustness of $\phi_t^{ab}$]
	\label{A-lem:useful-equalities}
	Fix any $t \in [T]$ and $a \neq b \in \{0,1\}$. Let functions with superscript $\star$ denote the true function, and functions without superscript $\star$ denote any generic function (those can also denote the estimated versions). We have
	\begin{itemize}
		\item[(i)] $\PP\left\{ \dfrac{\indic(A_t = d_t^a) ~ q_t(b|H_t,M_t)}{p_t(b|H_t) ~ q_t(a|H_t, M_t)} \mu_t(a, H_t, M_t) \right\} = \PP\left\{ \dfrac{q_t^\star(a|H_t, M_t)~q_t(b|H_t,M_t)}{p_t(b|H_t)~q_t(a|H_t, M_t)} \mu_t(a, H_t, M_t) \right\}$
		\item[(ii)] $\PP\left\{ \dfrac{\indic(A_t = d_t^a) ~ q_t(b|H_t,M_t)}{p_t(b|H_t) ~ q_t(a|H_t, M_t)} Y \right\} = \PP\left\{ \dfrac{\indic(A_t = d_t^a) ~ q_t(b|H_t,M_t)}{p_t(b|H_t) ~ q_t(a|H_t, M_t)} \mu_t^\star(a, H_t, M_t) \right\}$
		\item[(iii)] $\PP\left\{ \dfrac{\indic(A_t = d_t^b)}{p_t(b|H_t)} \mu_t(a, H_t, M_t) \right\} = \PP\left\{ \dfrac{q_t^\star(b|H_t,M_t)}{p_t(b|H_t)} \mu_t(a, H_t, M_t) \right\}$
		\item[(iv)] $\PP\left\{ \dfrac{q_t^\star(b | H_t, M_t)}{p_t(b|H_t)} \mu_t^\star(a, H_t, M_t)\right\} = \PP\left\{ \dfrac{p_t^\star(b | H_t)}{p_t(b|H_t)} \nu_t^\star(a, H_t)\right\}$
		\item[(v)] $\PP\left\{ \dfrac{\indic(A_t = d_t^b)}{p_t(b|H_t)} \nu_t(a, H_t) \right\} = \PP\left\{ \dfrac{p_t^\star(b|H_t)}{p_t(b|H_t)} \nu_t(a, H_t) \right\}$
	\end{itemize}
	Using the shorthand notation in the proof of \cref{A-lem:phi-ab-dr,A-lem:phi-ab-rate-dr}, these equalities can be written as (i) $\PP\Big( \dfrac{1_a q_b}{p_b q_a}\mu_a \Big) = \PP \Big( \dfrac{q_a^\star q_b}{p_b q_a} \mu_a \Big)$; (ii) $\PP \Big( \dfrac{1_a q_b}{p_b q_a} Y \Big) = \PP \Big( \dfrac{1_a q_b}{p_b q_a} \mu_a^\star \Big)$; (iii) $\PP \Big( \dfrac{1_b}{p_b} \mu_a \Big) = \PP \Big( \dfrac{q_b^\star}{p_b} \mu_a^\star \Big)$; (iv) $\PP \Big( \dfrac{q_b^\star}{p_b} \mu_a^\star \Big) = \PP \Big( \dfrac{p_b^\star}{p_b} \nu_a^\star \Big)$; (v) $\PP \Big( \dfrac{1_b}{p_b} \nu_a \Big) = \PP \Big( \dfrac{p_b^\star}{p_b} \nu_a \Big)$.
\end{lem}

\begin{proof}[Proof of \cref{A-lem:useful-equalities}]
	All equalities follow immediately from the law of iterated expectations and the definition of the true nuisance functions in \cref{eq:def-true-nuisance-function}. We present the detailed proof below.

	\underline{Proof of (i):} We have
	\begin{align*}
		& ~~~~ \PP\left\{ \dfrac{\indic(A_t = d_t^a) ~ q_t(b|H_t,M_t)}{p_t(b|H_t) ~ q_t(a|H_t, M_t)} \mu_t(a, H_t, M_t) \right\} \\
		& = \PP \left[ \PP\left\{ \dfrac{\indic(A_t = d_t^a) ~ q_t(b|H_t,M_t)}{p_t(b|H_t) ~ q_t(a|H_t, M_t)} \mu_t(a, H_t, M_t) ~\bigg|~ H_t, M_t \right\} \right] \\
		& = \PP\left\{ \dfrac{q_t^\star(a|H_t, M_t) ~ q_t(b|H_t,M_t)}{p_t(b|H_t) ~ q_t(a|H_t, M_t)} \mu_t(a, H_t, M_t) \right\}.
	\end{align*}

	\underline{Proof of (ii):} We have
	\begin{align*}
		& ~~~~ \PP\left\{ \dfrac{\indic(A_t = d_t^a) ~ q_t(b|H_t,M_t)}{p_t(b|H_t) ~ q_t(a|H_t, M_t)} Y \right\} \\
		& = \PP \left[ \PP\left\{ \dfrac{\indic(A_t = d_t^a) ~ q_t(b|H_t,M_t)}{p_t(b|H_t) ~ q_t(a|H_t, M_t)} Y ~\bigg|~ H_t, A_t, M_t \right\} \right] \\
		& = \PP \left[ \dfrac{\indic(A_t = d_t^a) ~ q_t(b|H_t,M_t)}{p_t(b|H_t) ~ q_t(a|H_t, M_t)} ~\PP\left\{ Y \mid H_t, A_t, M_t \right\} \right] \\
		& = \PP \left[ \dfrac{\indic(A_t = d_t^a) ~ q_t(b|H_t,M_t)}{p_t(b|H_t) ~ q_t(a|H_t, M_t)} ~\PP\left\{ Y \mid H_t, A_t = d_t^a, M_t \right\} \right] \\
		& = \PP\left\{ \dfrac{\indic(A_t = d_t^a) ~ q_t(b|H_t,M_t)}{p_t(b|H_t) ~ q_t(a|H_t, M_t)} \mu_t^\star(a, H_t, M_t) \right\}.
	\end{align*}

	\underline{Proof of (iii):} We have
	\begin{align*}
		\PP\left\{ \dfrac{\indic(A_t = d_t^b)}{p_t(b|H_t)} \mu_t(a, H_t, M_t) \right\} & = \PP \left[ \PP\left\{ \dfrac{\indic(A_t = d_t^b)}{p_t(b|H_t)} \mu_t(a, H_t, M_t) ~\bigg|~ H_t, M_t \right\} \right] \\
		& = \PP\left\{ \dfrac{q_t^\star(b|H_t, M_t)}{p_t(b|H_t)} \mu_t(a, H_t, M_t) \right\}.
	\end{align*}

	\underline{Proof of (iv):} we have
	\begin{align*}
		& ~~~~ \PP\left\{ \dfrac{q_t^\star(b | H_t, M_t)}{p_t(b|H_t)} \mu_t^\star(a, H_t, M_t)\right\} \\
		& = \PP\left\{ \dfrac{\indic(A_t = d_t^b)}{p_t(b|H_t)} \mu_t^\star(a, H_t, M_t)\right\} & \text{[because of (iii)]} \\
		& = \PP\left[ \PP\left\{ \dfrac{\indic(A_t = d_t^b)}{p_t(b|H_t)} \mu_t^\star(a, H_t, M_t) ~\bigg|~ H_t, A_t = d_t^b \right\} p_t^\star(b|H_t) \right] \\
		& ~~~~ + \PP\left[ \PP\left\{ \dfrac{\indic(A_t = d_t^b)}{p_t(b|H_t)} \mu_t^\star(a, H_t, M_t) ~\bigg|~ H_t, A_t \neq d_t^b \right\} \{1 - p_t^\star(b|H_t)\} \right] \\
		& = \PP \left[ \frac{p_t^\star(b|H_t)}{p_t(b|H_t)} ~ \PP \Big\{\mu_t^\star(a,H_t,M_t) ~\Big|~ H_t, A_t = d_t^b \Big\} \right] + 0 \\
		& = \PP\left\{ \dfrac{p_t^\star(b | H_t)}{p_t(b|H_t)} \nu_t^\star(a, H_t)\right\}. & \text{[definition of $\nu_t^\star(a, H_t)$]}
	\end{align*}

	\underline{Proof of (v):} We have
	\begin{align*}
		\PP\left\{ \dfrac{\indic(A_t = d_t^b)}{p_t(b|H_t)} \nu_t(a, H_t) \right\} & = \PP \left[ \PP\left\{ \dfrac{\indic(A_t = d_t^b)}{p_t(b|H_t)} \nu_t(a, H_t) ~\bigg|~ H_t \right\} \right] \\
		& = \PP\left\{ \dfrac{p_t^\star(b|H_t)}{p_t(b|H_t)} \nu_t(a, H_t) \right\}.
	\end{align*}

	This completes the proof.
\end{proof}

\bigskip

\begin{lem}[$\phi_t^{ab}$ is multiply robust]
	\label{A-lem:phi-ab-dr}
	Fix any $t \in [T]$ and $a \neq b \in \{0,1\}$. Consider $\phi^{ab}_t$ defined in \cref{eq:def-phi-ab}, whose sample average is used to estimate $\theta_t^{ab}$ in Algorithm \ref{algo:estimator-ncf}:
	\begin{align*}
	    \phi_t^{ab}(p_t, q_t, \mu_t, \nu_t) &:= \frac{\indic(A_t = d_t^a) ~q_t(b | H_t, M_t)}{p_t(b | H_t) ~q_t(a | H_t, M_t)} \Big\{ Y - \mu_t (a, H_t, M_t) \Big\} \nonumber \\
    	& ~~~~ + \frac{\indic(A_t = d_t^b)}{p_t(b | H_t)} \Big\{\mu_t (a, H_t, M_t) - \nu_t(a, H_t)\Big\} + \nu_t(a, H_t).
	\end{align*}
	Suppose one of the four conditions hold: (i) $p_t = p_t^\star$ and $q_t = q_t^\star$; or (ii) $p_t = p_t^\star$ and $\mu_t = \mu_t^\star$; or (iii) $\nu_t = \nu_t^\star$ and $q_t = q_t^\star$; or (iv) $\nu_t = \nu_t^\star$ and $\mu_t = \mu_t^\star$. Then $\PP \{\phi_t^{ab}(p_t, q_t, \mu_t, \nu_t)\} = \theta_t^{ab}$.
\end{lem}

\begin{proof}[Proof of \cref{A-lem:phi-ab-dr}]
	To reduce clutter, within this proof we use the shorthand notation: $1_a := \indic(A_t = d_t^a)$, $p_a := p_t(a | H_t)$, $q_a := q_t(a | H_t, M_t)$, $\mu_a := \mu_t(a, H_t, M_t)$, $\nu_a := \nu_t(a, H_t)$. Similarly for those with subscript $b$. Using this notation, we have
	\begin{align*}
		\phi_t^{ab}(p_t, q_t, \mu_t, \nu_t) & = \frac{1_a q_b}{p_b q_a} (Y - \mu_a) + \frac{1_b}{p_b} (\mu_a - \nu_a) + \nu_a \\
		& = \frac{1_a q_b}{p_b q_a} Y - \frac{1_a q_b}{p_b q_a} \mu_a + \frac{1_b}{p_b} \mu_a - \frac{1_b}{p_b} \nu_a + \nu_a.
	\end{align*}

	\underline{Case (i):} $p_t = p_t^\star$ and $q_t = q_t^\star$. In this case,
	\begin{align}
		\PP(\phi_t^{ab}) & = \PP\Big(\frac{1_a q_b^\star}{p_b^\star q_a^\star} Y\Big) - \PP\Big(\frac{1_a q_b^\star}{p_b^\star q_a^\star} \mu_a\Big) + \PP\Big(\frac{1_b}{p_b^\star} \mu_a\Big) - \PP\Big(\frac{1_b}{p_b^\star} \nu_a\Big) + \PP(\nu_a) \nonumber \\
		& = \theta_t^{ab} - \PP\Big(\frac{q_a^\star q_b^\star}{p_b^\star q_a^\star}\mu_a\Big) + \PP\Big(\frac{q_b^\star}{p_b^\star} \mu_a\Big) - \PP\Big(\frac{p_b^\star}{p_b^\star} \nu_a\Big) + \PP(\nu_a) \label{A-eq:lem:phi-ab-dr:proofuse1} \\
		& = \theta_t^{ab} - 0 - 0 = \theta_t^{ab}. \nonumber
	\end{align}
	Here, the first four terms in \cref{A-eq:lem:phi-ab-dr:proofuse1} follow from the cross-world weighting identification \cref{eq:identify-theta-ab-ipw} and \cref{A-lem:useful-equalities}(i)(iii)(v), respectively.

	\underline{Case (ii):} $p_t = p_t^\star$ and $\mu_t = \mu_t^\star$. In this case,
	\begin{align}
		\PP(\phi_t^{ab}) & = \PP\Big(\frac{1_a q_b}{p_b^\star q_a} Y\Big) - \PP\Big(\frac{1_a q_b}{p_b^\star q_a} \mu_a^\star\Big) + \PP\Big(\frac{1_b}{p_b^\star} \mu_a^\star\Big) - \PP\Big(\frac{1_b}{p_b^\star} \nu_a\Big) + \PP(\nu_a). \label{A-eq:lem:phi-ab-dr:proofuse2}
	\end{align}
	The first term in \cref{A-eq:lem:phi-ab-dr:proofuse2} is $\PP\Big(\dfrac{1_a q_b}{p_b^\star q_a} Y\Big) = \PP\Big(\dfrac{1_a q_b}{p_b^\star q_a} \mu_a^\star\Big)$ because of \cref{A-lem:useful-equalities}(ii), and thus it cancels with the second term in \cref{A-eq:lem:phi-ab-dr:proofuse2}. The third term in \cref{A-eq:lem:phi-ab-dr:proofuse2} is $\PP\Big(\dfrac{1_b}{p_b^\star} \mu_a^\star\Big) = \PP\Big(\dfrac{q_b^\star}{p_b^\star} \mu_a^\star\Big) = \PP(\nu_a^\star) = \theta_t^{ab}$ because of \cref{A-lem:useful-equalities}(iii), (iv), and the mediation g-formula \cref{eq:identify-theta-ab-ie}. The fourth term in \cref{A-eq:lem:phi-ab-dr:proofuse2} is $\PP\Big(\dfrac{1_b}{p_b^\star} \nu_a\Big) = \PP(\nu_a)$ because of \cref{A-lem:useful-equalities}(v), and thus it cancels with the last term in \cref{A-eq:lem:phi-ab-dr:proofuse2}. Therefore, $\PP(\phi_t^{ab}) = \theta_t^{ab}$ in Case (ii).

	\underline{Case (iii):} $\nu_t = \nu_t^\star$ and $q_t = q_t^\star$. In this case,
	\begin{align}
		\PP(\phi_t^{ab}) & = \PP\Big(\frac{1_a q_b^\star}{p_b q_a^\star} Y\Big) - \PP\Big(\frac{1_a q_b^\star}{p_b q_a^\star} \mu_a\Big) + \PP\Big(\frac{1_b}{p_b} \mu_a\Big) - \PP\Big(\frac{1_b}{p_b} \nu_a^\star\Big) + \PP(\nu_a^\star). \label{A-eq:lem:phi-ab-dr:proofuse3}
	\end{align}
	The first term in \cref{A-eq:lem:phi-ab-dr:proofuse3} is $\PP\Big(\dfrac{1_a q_b^\star}{p_b q_a^\star} Y\Big) = \PP\Big(\dfrac{q_a^\star q_b^\star}{p_b q_a^\star} \mu_a^\star\Big) = \PP\Big( \dfrac{p_b^\star}{p_b} \nu_a^\star \Big)$ because of \cref{A-lem:useful-equalities}(ii) and (iv). The second term in \cref{A-eq:lem:phi-ab-dr:proofuse3} is $\PP\Big(\dfrac{1_a q_b^\star}{p_b q_a^\star} \mu_a\Big) = \PP\Big(\dfrac{q_b^\star}{p_b} \mu_a\Big)$ due to \cref{A-lem:useful-equalities}(i). The third term in \cref{A-eq:lem:phi-ab-dr:proofuse3} is $\PP\Big(\dfrac{1_b}{p_b} \mu_a\Big) = \PP\Big(\dfrac{q_b^\star}{p_b} \mu_a\Big)$ because of \cref{A-lem:useful-equalities}(iii), and thus it cancels with the second term. The fourth term in \cref{A-eq:lem:phi-ab-dr:proofuse3} is $\PP\Big(\dfrac{1_b}{p_b} \nu_a^\star\Big) = \PP\Big(\dfrac{p_b^\star}{p_b} \nu_a^\star\Big)$ because of \cref{A-lem:useful-equalities}(v), and thus it cancels with the first term. The last term in \cref{A-eq:lem:phi-ab-dr:proofuse3} is $\PP(\nu_a^\star) = \theta_t^{ab}$ because of the mediation g-formula \cref{eq:identify-theta-ab-ie}. Therefore, $\PP(\phi_t^{ab}) = \theta_t^{ab}$ in Case (iii).

	\underline{Case (iv):} $\nu_t = \nu_t^\star$ and $\mu_t = \mu_t^\star$. In this case,
	\begin{align}
		\PP(\phi_t^{ab}) & = \PP\Big(\frac{1_a q_b}{p_b q_a} Y\Big) - \PP\Big(\frac{1_a q_b}{p_b q_a} \mu_a^\star\Big) + \PP\Big(\frac{1_b}{p_b} \mu_a^\star\Big) - \PP\Big(\frac{1_b}{p_b} \nu_a^\star\Big) + \PP(\nu_a^\star). \label{A-eq:lem:phi-ab-dr:proofuse4}
	\end{align}
	The first term in \cref{A-eq:lem:phi-ab-dr:proofuse4} is $\PP\Big(\frac{1_a q_b}{p_b q_a} Y\Big) = \PP\Big(\frac{1_a q_b}{p_b q_a} \mu_a^\star\Big)$ because of \cref{A-lem:useful-equalities}(ii), and thus it cancels with the second term in \cref{A-eq:lem:phi-ab-dr:proofuse3}. The third term in \cref{A-eq:lem:phi-ab-dr:proofuse4} is $\PP\Big(\dfrac{1_b}{p_b} \mu_a^\star\Big) = \PP\Big(\dfrac{q_b^\star}{p_b} \mu_a^\star\Big) = \PP\Big(\dfrac{p_b^\star}{p_b} \nu_a^\star\Big)$ because of \cref{A-lem:useful-equalities}(iii) and (iv). The fourth term in \cref{A-eq:lem:phi-ab-dr:proofuse4} is $\PP\Big(\frac{1_b}{p_b} \nu_a^\star\Big) = \PP\Big(\frac{p_b^\star}{p_b} \nu_a^\star\Big)$ because of \cref{A-lem:useful-equalities}(v), and thus it cancels with the third term. The last term in \cref{A-eq:lem:phi-ab-dr:proofuse4} is $\PP(\nu_a^\star) = \theta_t^{ab}$ because of the mediation g-formula \cref{eq:identify-theta-ab-ie}. Therefore, $\PP(\phi_t^{ab}) = \theta_t^{ab}$ in Case (iv).

	Therefore, we have verified that $\PP \{\phi_t^{ab}(p_t, q_t, \mu_t, \nu_t)\} = \theta_t^{ab}$ in each of the four cases. This completes the proof.
\end{proof}

\bigskip

\begin{lem}[$\phi_t^{ab}$ is rate multiply robust]
	\label{A-lem:phi-ab-rate-dr}
	Fix any $t \in [T]$ and $a \neq b \in \{0,1\}$. Consider $\phi^{ab}_t$ defined in \cref{eq:def-phi-ab}.	Suppose one of the four conditions hold: (i) $p_t' = p_t^\star$ and $q_t' = q_t^\star$; or (ii) $p_t' = p_t^\star$ and $\mu_t' = \mu_t^\star$; or (iii) $\nu_t' = \nu_t^\star$ and $q_t' = q_t^\star$; or (iv) $\nu_t' = \nu_t^\star$ and $\mu_t' = \mu_t^\star$. Then
	\begin{align*}
		\Big|\PP \big\{\phi_t^{ab}(\hat{p}_t, \hat{q}_t, \hat\mu_t, \hat\nu_t) - \phi_t^{ab}(p_t', q_t', \mu_t', \nu_t') \big\} \Big| \lesssim \|\hat{p}_t - p_t^\star\| \cdot \|\hat\nu_t - \nu_t^\star\| + \|\hat{q}_t - q_t^\star\| \cdot \|\hat\mu_t - \mu_t^\star\|.
	\end{align*}
\end{lem}

\begin{proof}[Proof of \cref{A-lem:phi-ab-rate-dr}]
	To reduce clutter, within this proof we use the same shorthand notation as in the proof of \cref{A-lem:phi-ab-dr}: $1_a := \indic(A_t = d_t^a)$, $p_a := p_t(a | H_t)$, $q_a := q_t(a | H_t, M_t)$, $\mu_a := \mu_t(a, H_t, M_t)$, $\nu_a := \nu_t(a, H_t)$. Similarly for those with subscript $b$. Using this notation, we have
	\begin{align*}
		\phi_t^{ab}(p_t, q_t, \mu_t, \nu_t) & = \frac{1_a q_b}{p_b q_a} (Y - \mu_a) + \frac{1_b}{p_b} (\mu_a - \nu_a) + \nu_a \\
		& = \frac{1_a q_b}{p_b q_a} Y - \frac{1_a q_b}{p_b q_a} \mu_a + \frac{1_b}{p_b} \mu_a - \frac{1_b}{p_b} \nu_a + \nu_a. 
	\end{align*}
	This implies that
	\begin{align}
		& ~~~~ \PP \big\{\phi_t^{ab}(\hat{p}_t, \hat{q}_t, \hat\mu_t, \hat\nu_t) - \phi_t^{ab}(p_t', q_t', \mu_t', \nu_t') \big\} \nonumber \\
		& = \PP \Big( \frac{1_a \hat{q}_b}{\hat{p}_b \hat{q}_a} Y \Big) - \PP \Big( \frac{1_a \hat{q}_b}{\hat{p}_b \hat{q}_a} \hat\mu_a \Big) + \PP \Big( \frac{1_b}{\hat{p}_b} \hat\mu_a \Big) - \PP \Big( \frac{1_b}{\hat{p}_b} \hat\nu_a \Big) + \PP (\hat\nu_a) \nonumber \\
		& ~~~~ - \PP \Big( \frac{1_a q_b'}{p_b' q_a'} Y \Big) + \PP \Big( \frac{1_a q_b'}{p_b' q_a'} \mu_a' \Big) - \PP \Big( \frac{1_b}{p_b'} \mu_a' \Big) + \PP \Big( \frac{1_b}{p_b'} \nu_a' \Big) - \PP (\nu_a') \nonumber \\
		& = \PP \Big( \frac{q_a^\star \hat{q}_b}{\hat{p}_b \hat{q}_a} \mu_a^\star \Big) - \PP \Big( \frac{q_a^\star \hat{q}_b}{\hat{p}_b \hat{q}_a} \hat\mu_a \Big) + \PP \Big( \frac{q_b^\star}{\hat{p}_b} \hat\mu_a \Big) - \PP \Big( \frac{p_b^\star}{\hat{p}_b} \hat\nu_a \Big) + \PP (\hat\nu_a) \nonumber \\
		& ~~~~ - \PP \Big( \frac{q_a^\star q_b'}{p_b' q_a'} \mu_a^\star \Big) + \PP \Big( \frac{q_a^\star q_b'}{p_b' q_a'} \mu_a' \Big) - \PP \Big( \frac{q_b^\star}{p_b'} \mu_a' \Big) + \PP \Big( \frac{p_b^\star}{p_b'} \nu_a' \Big) - \PP (\nu_a'), \label{A-eq:lem:phi-ab-rate-dr:proofuse1}
	\end{align}
	where \cref{A-eq:lem:phi-ab-rate-dr:proofuse1} follows from \cref{A-lem:useful-equalities}(i), (ii), (iii), and (v).

	We now add and subtract three terms into \cref{A-eq:lem:phi-ab-rate-dr:proofuse1}: $\PP\Big(- \dfrac{q_b^\star}{\hat{p}_b}\mu_a^\star \Big)$, $\PP\Big(\dfrac{p_b^\star}{\hat{p}_b}\nu_a^\star \Big)$, $\PP(-\nu_a^\star)$. This gives
	\begin{align}
		\PP \big\{\phi_t^{ab}(\hat{p}_t, \hat{q}_t, \hat\mu_t, \hat\nu_t) - \phi_t^{ab}(p_t', q_t', \mu_t', \nu_t') \big\} = \lambda_1 + \lambda_2 + \lambda_3, \label{A-eq:lem:phi-ab-rate-dr:proofuse2}
	\end{align}
	with
	\begin{align*}
		\lambda_1 & := \PP \bigg( \frac{q_a^\star \hat{q}_b}{\hat{p}_b \hat{q}_a}\mu_a^\star - \frac{q_a^\star \hat{q}_b}{\hat{p}_b \hat{q}_a}\hat\mu_a - \frac{q_b^\star}{\hat{p}_b} \mu_a^\star + \frac{q_b^\star}{\hat{p}_b} \hat\mu_a \bigg), \\
		\lambda_2 & := \PP \bigg( \frac{p_b^\star}{\hat{p}_b} \nu_a^\star - \frac{p_b^\star}{\hat{p}_b} \hat\nu_a - \nu_a^\star + \hat\nu_a \bigg), \\
		\lambda_3 & := \PP \bigg( -\frac{q_a^\star q_b'}{p_b' q_a'}\mu_a^\star + \frac{q_a^\star q_b'}{p_b' q_a'}\mu_a' + \frac{q_b^\star}{\hat{p}_b}\mu_a^\star - \frac{q_b^\star}{p_b'}\mu_a' + \frac{p_b^\star}{p_b'} \nu_a' - \frac{p_b^\star}{\hat{p}_b} \nu_a^\star + \nu_a^\star - \nu_a' \bigg).
	\end{align*}
	Note that only $\lambda_3$ involves the quantities $p', q', \mu', \nu'$.

	First, we show that $|\lambda_1| \lesssim \|\hat{q}_t - q_t^\star\| \cdot \|\hat\mu_t - \mu_t^\star\|$. We have
	\begin{align}
		\lambda_1 & = \PP \bigg\{ \frac{q_a^\star \hat{q}_b}{\hat{p}_b \hat{q}_a}(\mu_a^\star - \hat\mu_a) - \frac{q_b^\star}{\hat{p}_b} (\mu_a^\star - \hat\mu_a) \bigg\} \nonumber \\
		& = \PP \bigg\{ \frac{1}{\hat{p}_b \hat{q}_a} (q_a^\star \hat{q}_b - \hat{q}_a q_b^\star ) (\mu_a^\star - \hat\mu_a) \bigg\} \nonumber \\
		& = \PP \bigg\{ \frac{1}{\hat{p}_b \hat{q}_a} (q_a^\star \hat{q}_b - q_a^\star q_b^\star + q_a^\star q_b^\star - \hat{q}_a q_b^\star ) (\mu_a^\star - \hat\mu_a) \bigg\} \nonumber \\
		& = \PP \bigg[ \frac{1}{\hat{p}_b \hat{q}_a}\Big\{ q_a^\star(\hat{q}_b - q_b^\star) - q_b^\star(\hat{q}_a - q_a^\star) \Big\} (\mu_a^\star - \hat\mu_a) \bigg]. \nonumber \\
		\implies |\lambda_1| & \lesssim \|\hat{q}_t - q_t^\star\| \cdot \|\hat\mu_t - \mu_t^\star\|. \label{A-eq:lem:phi-ab-rate-dr:proofuse3}
	\end{align}

	Second, we show that $|\lambda_2| \lesssim \|\hat{p}_t - p_t^\star\| \cdot \|\hat\nu_t - \nu_t^\star\|$. We have
	\begin{align}
		\lambda_2 & = \PP \bigg\{ \frac{p_b^\star}{\hat{p}_b}(\nu_a^\star - \hat\nu_a) - (\nu_a^\star - \hat\nu_a) \bigg\} \nonumber \\
		& = \PP \bigg\{ \bigg( \frac{p_b^\star}{\hat{p}_b} - 1 \bigg) (\nu_a^\star - \hat\nu_a) \bigg\} \nonumber \\
		& = \PP \bigg\{ \frac{1}{\hat{p}_b} (p_b^\star - \hat{p}_b) (\nu_a^\star - \hat\nu_a) \bigg\} \nonumber \\
		\implies |\lambda_2| & \lesssim \|\hat{p}_t - p_t^\star\| \cdot \|\hat\nu_t - \nu_t^\star\|. \label{A-eq:lem:phi-ab-rate-dr:proofuse4}
	\end{align}

	Lastly, we show that $\lambda_3 = 0$ under any of the four conditions in the lemma statement.

	\underline{Case (i):} $p_t' = p_t^\star$ and $q_t' = q_t^\star$. In this case,
	\begin{align}
		\lambda_3 & = \PP \bigg( & -\frac{q_a^\star q_b^\star}{p_b^\star q_a^\star}\mu_a^\star & & + \frac{q_a^\star q_b^\star}{p_b^\star q_a^\star}\mu_a' & & + \frac{q_b^\star}{\hat{p}_b}\mu_a^\star & & - \frac{q_b^\star}{p_b^\star}\mu_a' & & + \frac{p_b^\star}{p_b^\star} \nu_a' & & - \frac{p_b^\star}{\hat{p}_b} \nu_a^\star & & + \nu_a^\star & & - \nu_a' & \bigg) \nonumber \\
		& = \PP \bigg(& -\frac{q_b^\star}{p_b^\star}\mu_a^\star & & + \frac{q_b^\star}{p_b^\star}\mu_a' & & + \frac{q_b^\star}{\hat{p}_b}\mu_a^\star & & - \frac{q_b^\star}{p_b^\star}\mu_a' & & + \nu_a' & & - \frac{p_b^\star}{\hat{p}_b} \nu_a^\star & & + \nu_a^\star & & - \nu_a' & \bigg) \nonumber \\
		& = \PP \bigg(& -\frac{q_b^\star}{p_b^\star}\mu_a^\star & & & & + \frac{q_b^\star}{\hat{p}_b}\mu_a^\star & & & & & & - \frac{p_b^\star}{\hat{p}_b} \nu_a^\star & & + \nu_a^\star & & & \bigg) \nonumber \\
		& = \PP \bigg(& -\frac{p_b^\star}{p_b^\star}\nu_a^\star & & & & + \frac{p_b^\star}{\hat{p}_b}\nu_a^\star & & & & & & - \frac{p_b^\star}{\hat{p}_b} \nu_a^\star & & + \nu_a^\star & & & \bigg) \label{A-eq:lem:phi-ab-rate-dr:proofuse5} \\
		& = 0, \nonumber
	\end{align}
	where \cref{A-eq:lem:phi-ab-rate-dr:proofuse5} follows from \cref{A-lem:useful-equalities}(iv).

	\underline{Case (ii):} $p_t' = p_t^\star$ and $\mu_t' = \mu_t^\star$. In this case,
	\begin{align}
		\lambda_3 & = \PP \bigg( & -\frac{q_a^\star q_b'}{p_b^\star q_a'}\mu_a^\star & & + \frac{q_a^\star q_b'}{p_b^\star q_a'}\mu_a^\star & & + \frac{q_b^\star}{\hat{p}_b}\mu_a^\star & & - \frac{q_b^\star}{p_b^\star}\mu_a^\star & & + \frac{p_b^\star}{p_b^\star} \nu_a' & & - \frac{p_b^\star}{\hat{p}_b} \nu_a^\star & & + \nu_a^\star & & - \nu_a' & \bigg) \nonumber \\
		& = \PP \bigg( & & & & & + \frac{q_b^\star}{\hat{p}_b}\mu_a^\star & & - \frac{q_b^\star}{p_b^\star}\mu_a^\star & & & & - \frac{p_b^\star}{\hat{p}_b} \nu_a^\star & & + \nu_a^\star & & & \bigg) \nonumber \\
		& = \PP \bigg( & & & & & + \frac{p_b^\star}{\hat{p}_b}\nu_a^\star & & - \frac{p_b^\star}{p_b^\star}\nu_a^\star & & & & - \frac{p_b^\star}{\hat{p}_b} \nu_a^\star & & + \nu_a^\star & & & \bigg) \label{A-eq:lem:phi-ab-rate-dr:proofuse6} \\
		& = 0, \nonumber
	\end{align}
	where \cref{A-eq:lem:phi-ab-rate-dr:proofuse6} follows from \cref{A-lem:useful-equalities}(iv).

	\underline{Case (iii):} $\nu_t' = \nu_t^\star$ and $q_t' = q_t^\star$. In this case,
	\begin{align}
		\lambda_3 & = \PP \bigg( & -\frac{q_a^\star q_b^\star}{p_b' q_a^\star}\mu_a^\star & & + \frac{q_a^\star q_b^\star}{p_b' q_a^\star}\mu_a' & & + \frac{q_b^\star}{\hat{p}_b}\mu_a^\star & & - \frac{q_b^\star}{p_b'}\mu_a' & & + \frac{p_b^\star}{p_b'} \nu_a^\star & & - \frac{p_b^\star}{\hat{p}_b} \nu_a^\star & & + \nu_a^\star & & - \nu_a^\star & \bigg) \nonumber \\
		& = \PP \bigg( & -\frac{q_b^\star}{p_b'}\mu_a^\star & & + \frac{q_b^\star}{p_b'}\mu_a' & & + \frac{q_b^\star}{\hat{p}_b}\mu_a^\star & & - \frac{q_b^\star}{p_b'}\mu_a' & & + \frac{p_b^\star}{p_b'} \nu_a^\star & & - \frac{p_b^\star}{\hat{p}_b} \nu_a^\star & & & & & \bigg) \nonumber \\
		& = \PP \bigg( & -\frac{q_b^\star}{p_b'}\mu_a^\star & & & & + \frac{q_b^\star}{\hat{p}_b}\mu_a^\star & & & & + \frac{p_b^\star}{p_b'} \nu_a^\star & & - \frac{p_b^\star}{\hat{p}_b} \nu_a^\star & & & & & \bigg) \nonumber \\
		& = \PP \bigg( & -\frac{p_b^\star}{p_b'}\nu_a^\star & & & & + \frac{p_b^\star}{\hat{p}_b}\nu_a^\star & & & & + \frac{p_b^\star}{p_b'} \nu_a^\star & & - \frac{p_b^\star}{\hat{p}_b} \nu_a^\star & & & & & \bigg) \label{A-eq:lem:phi-ab-rate-dr:proofuse7} \\
		& = 0, \nonumber
	\end{align}
	where \cref{A-eq:lem:phi-ab-rate-dr:proofuse7} follows from \cref{A-lem:useful-equalities}(iv).

	\underline{Case (iv):} $\nu_t' = \nu_t^\star$ and $\mu_t' = \mu_t^\star$. In this case,
	\begin{align}
		\lambda_3 & = \PP \bigg( & -\frac{q_a^\star q_b'}{p_b' q_a'}\mu_a^\star & & + \frac{q_a^\star q_b'}{p_b' q_a'}\mu_a^\star & & + \frac{q_b^\star}{\hat{p}_b}\mu_a^\star & & - \frac{q_b^\star}{p_b'}\mu_a^\star & & + \frac{p_b^\star}{p_b'} \nu_a^\star & & - \frac{p_b^\star}{\hat{p}_b} \nu_a^\star & & + \nu_a^\star & & - \nu_a^\star & \bigg) \nonumber \\
		& = \PP \bigg( & & & & & + \frac{q_b^\star}{\hat{p}_b}\mu_a^\star & & - \frac{q_b^\star}{p_b'}\mu_a^\star & & + \frac{p_b^\star}{p_b'} \nu_a^\star & & - \frac{p_b^\star}{\hat{p}_b} \nu_a^\star & & & & & \bigg) \nonumber \\
		& = \PP \bigg( & & & & & + \frac{p_b^\star}{\hat{p}_b}\nu_a^\star & & - \frac{p_b^\star}{p_b'}\nu_a^\star & & + \frac{p_b^\star}{p_b'} \nu_a^\star & & - \frac{p_b^\star}{\hat{p}_b} \nu_a^\star & & & & & \bigg) \label{A-eq:lem:phi-ab-rate-dr:proofuse8} \\
		& = 0, \nonumber
	\end{align}
	where \cref{A-eq:lem:phi-ab-rate-dr:proofuse8} follows from \cref{A-lem:useful-equalities}(iv).

	Therefore, we have shown that $\lambda_3 = 0$ under any of the four conditions in the lemma statement. This combined with \cref{A-eq:lem:phi-ab-rate-dr:proofuse2,A-eq:lem:phi-ab-rate-dr:proofuse3,A-eq:lem:phi-ab-rate-dr:proofuse4} gives the desired lemma result. This completes the proof.
\end{proof}

\bigskip

\begin{lem}[$\psi$ is multiply robust]
	\label{A-lem:psi-multiply-robust}
	Consider $\psi(\gamma, \zeta)$, the estimating function for $\gamma = (\alpha,\beta)$, defined in \cref{A-eq:def-psi}. For each $t \in [T]$, suppose that the nuisance function $\zeta_t = (p_t, q_t, \eta_t, \mu_t, \nu_t)$ satisfies one of the four conditions: (i) $p_t = p_t^\star$ and $q_t = q_t^\star$; or (ii) $p_t = p_t^\star$ and $\mu_t = \mu_t^\star$; or (iii) $\eta_t = \eta_t^\star$, $\nu_t = \nu_t^\star$ and $q_t = q_t^\star$; or (iv) $\eta_t = \eta_t^\star$, $\nu_t = \nu_t^\star$ and $\mu_t = \mu_t^\star$. Then $\PP\{\psi(\gamma^\star, \zeta)\} = 0$.
\end{lem}

\begin{proof}[Proof of \cref{A-lem:psi-multiply-robust}]
	The definition of $\gamma^\star = (\alpha^\star, \beta^\star)$ in \cref{eq:alphabeta-star-def} implies that
	\begin{align}
		\PP\left(\sum_{t=1}^T \omega(t)
		\left[\begin{matrix}
			\big\{ \theta_t^{10} - \theta_t^{00} - f(t)^T\alpha^\star \big\} f(t) \\
			\big\{ \theta_t^{11} - \theta_t^{10} - f(t)^T\beta^\star \big\} f(t)
		\end{matrix}\right] \right) = 0. \label{A-eq:lem:psi-multiply-robust:proofuse1}
	\end{align}
	Under any of the four conditions in the lemma statement, we have the following due to the double robustness of $\phi_t^{aa}$ and the multiple robustness of $\phi_t^{ab}$ (\cref{A-lem:phi-aa-dr,A-lem:phi-ab-dr}):
	\begin{align}
		\PP\{\phi_t^{10}(p_t,q_t,\mu_t,\nu_t)\} = \theta_t^{10},
		& & \PP\{\phi_t^{00}(p_t,\eta_t)\} = \theta_t^{00}, \nonumber \\
		\PP\{\phi_t^{01}(p_t,q_t,\mu_t,\nu_t)\} = \theta_t^{01},
		& & \PP\{\phi_t^{11}(p_t,\eta_t)\} = \theta_t^{11}.
	\end{align}
	Therefore, using the definition of $\psi$ in \cref{A-eq:def-psi}, under any of the four conditions in the lemma statement we have
	\begin{align}
		\PP\{\psi(\gamma^\star, \zeta)\} = \PP\left(\sum_{t=1}^T \omega(t)
		\left[\begin{matrix}
			\big\{ \theta_t^{10} - \theta_t^{00} - f(t)^T\alpha^\star \big\} f(t) \\
			\big\{ \theta_t^{11} - \theta_t^{10} - f(t)^T\beta^\star \big\} f(t)
		\end{matrix}\right] \right) = 0, \nonumber
	\end{align}
	where the last equality follows from \cref{A-eq:lem:psi-multiply-robust:proofuse1}. This completes the proof.
\end{proof}

\bigskip

\begin{lem}[$\psi$ is rate multiply robust]
	\label{A-lem:psi-rate-multiply-robust}
	Consider $\psi(\gamma, \zeta)$, the estimating function for $\gamma = (\alpha,\beta)$, defined in \cref{A-eq:def-psi}. For each $t \in [T]$, suppose that the fitted nuisance function $\hat\zeta_t = (\hat{p}_t, \hat{q}_t, \hat\eta_t, \hat\mu_t, \hat\nu_t)$ converges in $L_2$ to limit $\zeta_t' = (p_t', q_t', \eta_t', \mu_t', \nu_t')$. Suppose that for each $t \in [T]$, the limit $\zeta_t'$ satisfies one of the four conditions: (i) $p_t' = p_t^\star$ and $q_t' = q_t^\star$; or (ii) $p_t' = p_t^\star$ and $\mu_t' = \mu_t^\star$; or (iii) $\eta_t' = \eta_t^\star$, $\nu_t' = \nu_t^\star$ and $q_t' = q_t^\star$; or (iv) $\eta_t' = \eta_t^\star$, $\nu_t' = \nu_t^\star$ and $\mu_t' = \mu_t^\star$. Then
	\begin{align*}
		\Big|\PP\{\psi(\gamma^\star, \hat\zeta) - \psi(\gamma^\star, \zeta')\} \Big| \lesssim \max_{t \in [T]} \bigg\{\|\hat{p}_t - p_t^\star\| \cdot \Big(\|\hat\eta_t - \eta_t^\star\| + \|\hat\nu_t - \nu_t^\star\| \Big) + \|\hat{q}_t - q_t^\star\| \cdot \|\hat\mu_t - \mu_t^\star\| \bigg\}.
	\end{align*}
\end{lem}

\begin{proof}[Proof of \cref{A-lem:psi-rate-multiply-robust}]
	By the form of $\psi$ in \cref{A-eq:def-psi}, We have
	\begin{align}
		& ~~~~ \PP\{\psi(\gamma^\star, \hat\zeta) - \psi(\gamma^\star, \zeta')\} \nonumber \\
		& = \PP \sum_{t=1}^T \omega(t)
			\left[\begin{matrix}
				\Big( \big\{ \phi_t^{10}(\hat{p}_t, \hat{q}_t, \hat\mu_t, \hat\nu_t) - \phi_t^{10}(p_t', q_t', \mu_t', \nu_t') \big\} - \big\{\phi_t^{00}(\hat{p}_t, \hat\eta_t) - \phi_t^{00}(p_t', \eta_t') \big\} \Big) f(t) \\
				\Big( \big\{ \phi_t^{11}(\hat{p}_t, \hat\eta_t) - \phi_t^{11}(p_t', \eta_t') \big\} - \big\{ \phi_t^{10}(\hat{p}_t, \hat{q}_t, \hat\mu_t, \hat\nu_t) + \phi_t^{10}(p_t', q_t', \mu_t', \nu_t') \big\} \Big) f(t)
			\end{matrix}\right]. \nonumber
	\end{align}
	Therefore, the Lemma conclusion follows immediately from \cref{A-lem:phi-aa-rate-dr,A-lem:phi-ab-rate-dr}. This completes the proof.
\end{proof}

\subsection{Additional Technical Lemmas for Non-Cross-Fitted Estimator}
\label{A-subsec:asymptotic-additional-lemmas-ncf}

We state additional technical lemmas used in establishing asymptotic normality and prove them. These lemmas are generalizations of \citet[][Section B.3]{cheng2023efficient} to estimating functions that are not necessarily globally robust, and they can potentially be applied to other statistical problems and estimation procedures.

We use $|\cdot|$ to denote the absolute value of a scalar of the Euclidean norm of a vector depending on the argument.

\begin{lem}[Well-separated zero]
	\label{A-lem:well-separated-zero}
	Suppose the following regularity conditions hold: \cref{A-asu:unique-zero,A-asu:reg-compact-param-space,A-asu:reg-cont-PPee}. Then the unique zero $\gamma^\star$ of the function $\PP\{\psi(\gamma, \zeta')\}$ is well-separated. That is, for any $\epsilon > 0$, there exists $\delta > 0$ such that $|\gamma - \gamma^\star| > \epsilon$ implies $|\PP\{\psi(\gamma, \zeta')\}| > \delta$.
\end{lem}

\begin{proof}[Proof of \cref{A-lem:well-separated-zero}]
	For any $\epsilon > 0$, consider $\delta := \inf_{\gamma \in \Theta: |\gamma - \gamma^\star| \geq \epsilon}|\PP\{\psi(\gamma, \zeta')\}|$. Because the parameter space $\Theta$ is compact (\cref{A-asu:reg-compact-param-space}), $\{\gamma: |\gamma - \gamma^\star| \geq \epsilon\} \cap \Theta$ is also compact. This combined with the fact that $\PP\{\psi(\gamma, \zeta')\}$ is a continuous function in $\gamma$ (\cref{A-asu:reg-cont-PPee}) implies that the infimum is attained, i.e., $\inf_{\gamma \in \Theta: |\gamma - \gamma^\star| \geq \epsilon}|\PP\{\psi(\gamma, \zeta')\}| = \min_{\gamma \in \Theta: |\gamma - \gamma^\star| \geq \epsilon}|\PP\{\psi(\gamma, \zeta')\}|$. Because $\gamma^\star$ is the unique zero of $\PP\{\psi(\gamma, \zeta')\}$ (\cref{A-asu:unique-zero}), $\min_{\gamma \in \Theta: |\gamma - \gamma^\star| \geq \epsilon}|\PP\{\psi(\gamma, \zeta')\}| > 0$. Therefore, we proved the lemma by constructing a particular $\delta := \min_{\gamma \in \Theta: |\gamma - \gamma^\star| \geq \epsilon}|\PP\{\psi(\gamma, \zeta')\}| > 0$.
\end{proof}

\bigskip

\begin{lem}[Consistency of $\hat\gamma$]
	\label{A-lem:consistency}
	Suppose $\PP \psi(\gamma^\star, \zeta') = 0$. Suppose the following regularity conditions hold: \cref{A-asu:unique-zero,A-asu:reg-compact-param-space,A-asu:reg-bounded-obs,A-asu:reg-cont-PPee,A-asu:nuisance-conv-PPee-sup,A-asu:donsker}. Then $\hat\gamma \pto \gamma^\star$ as $n\to\infty$.
\end{lem}

\begin{proof}[Proof of \cref{A-lem:consistency}]
	Consider an arbitrary $\epsilon > 0$. We need to prove $\lim_{n\to\infty} P(| \hat\gamma - \gamma^\star | > \epsilon) = 0$. Because \cref{A-asu:unique-zero,A-asu:reg-compact-param-space} hold, by \cref{A-lem:well-separated-zero} there exists $\delta > 0$ such that
	\begin{align*}
		P(| \hat\gamma - \gamma^\star | > \epsilon) \leq P[ |\PP\{\psi(\hat\gamma, \zeta')\}| > \delta].
	\end{align*}
	Therefore, it suffices to prove that $|\PP\{\psi(\hat\gamma, \zeta')\}|$ converges in probability to 0.

	Because $\PP_n\{\psi(\hat\gamma, \hat\zeta)\} = 0$, we have
	\begin{align}
		|\PP\{\psi(\hat\gamma, \zeta')\}|
		& = | \PP\{\psi(\hat\gamma, \zeta')\} - \PP_n\{\psi(\hat\gamma, \hat\zeta)\} | \nonumber \\
		& \leq | \PP\{\psi(\hat\gamma, \zeta')\} - \PP\{\psi(\hat\gamma, \hat\zeta)\} | + | \PP\{\psi(\hat\gamma, \hat\zeta)\} - \PP_n\{\psi(\hat\gamma, \hat\zeta)\} | \label{A-eq:lem:general-consistency:proofuse1}
	\end{align}
	Next we show that both terms in \cref{A-eq:lem:general-consistency:proofuse1} are $o_P(1)$.

	For the first term in \cref{A-eq:lem:general-consistency:proofuse1}, by \cref{A-asu:nuisance-conv-PPee-sup} we have
	\begin{align}
		| \PP\{\psi(\hat\gamma, \zeta')\} - \PP\{\psi(\hat\gamma, \hat\zeta)\} | \leq \sup_{\gamma \in \Theta} | \PP\{\psi(\gamma, \zeta')\} - \PP\{\psi(\gamma, \hat\zeta)\} | = o_P(1), \label{A-eq:lem:general-consistency:proofuse2}
	\end{align}
	
	For the second term in \cref{A-eq:lem:general-consistency:proofuse1}, because $\Psi$ is a $P$-Donsker class (\cref{A-asu:donsker}) and thus a $P$-Glivenko-Cantelli class, $\sup_{\gamma \in \Theta, \zeta \in \cZ}| (\PP_n - \PP) \psi(\gamma, \zeta)| = o_P(1)$. Therefore, $| \PP\{\psi(\hat\gamma, \hat\zeta)\} - \PP_n\{\psi(\hat\gamma, \hat\zeta)\} | = | (\PP_n - \PP) \psi(\hat\gamma, \hat\zeta)| = o_P(1)$.

	Thus, we showed that both terms in \cref{A-eq:lem:general-consistency:proofuse1} are $o_P(1)$. This completes the proof.	
\end{proof}

\bigskip

\begin{lem}[Convergence of the derivative.]
	\label{A-lem:conv-ee-deriv}
	Suppose the following regularity conditions hold: \cref{A-asu:reg-dominated-ee-deriv,A-asu:nuisance-conv-PPee-deriv,A-asu:donsker}. If a sequence of random variables $\bar\gamma_n$ satisfies $\bar\gamma_n \pto \gamma^\star$, then $\PP_n \{ \partial_\gamma \psi(\bar\gamma_n, \hat\zeta) \} \pto \PP \left\{ \partial_\gamma \psi(\gamma^\star, \zeta') \right\}$.
\end{lem}

\begin{proof}[Proof of \cref{A-lem:conv-ee-deriv}]
	We have
	\begin{align}
		& ~~~~ \PP_n \{ \partial_\gamma \psi(\bar\gamma_n, \hat\zeta) \} - \PP \big\{ \partial_\gamma \psi(\gamma^\star, \zeta') \big\} \nonumber \\
		& =  \PP_n \{ \partial_\gamma \psi(\bar\gamma_n, \hat\zeta) \} - \PP \big\{ \partial_\gamma \psi(\bar\gamma_n, \hat\zeta) \big\} + \PP \big\{ \partial_\gamma \psi(\bar\gamma_n, \hat\zeta) \big\} -  \PP \big\{ \partial_\gamma \psi(\gamma^\star, \zeta') \big\} \nonumber \\
		& \leq \sup_{\gamma \in \Theta, \zeta \in \cZ} | (\PP_n - \PP) \partial_\gamma \psi(\gamma, \zeta) | \nonumber \\
		& ~~~~ + \bigg[\PP \big\{ \partial_\gamma \psi(\bar\gamma_n, \hat\zeta) \big\} -  \PP \big\{ \partial_\gamma \psi(\gamma^\star, \hat\zeta) \big\}\bigg] + \bigg[\PP \big\{ \partial_\gamma \psi(\gamma^\star, \hat\zeta) \big\} -  \PP \big\{ \partial_\gamma \psi(\gamma^\star, \zeta') \big\}\bigg].
		\label{A-eq:lem:conv-ee-deriv:proofuse1}
	\end{align}
	We now control each of the three terms in \cref{A-eq:lem:conv-ee-deriv:proofuse1}.

	For the first term in \cref{A-eq:lem:conv-ee-deriv:proofuse1}, \cref{A-asu:donsker} implies that $\{\partial_\gamma \psi(\gamma,\zeta): \gamma \in \Theta, \zeta \in \cZ\}$ is a $P$-Glivenko-Cantelli class and thus
	\begin{align}
		\sup_{\gamma \in \Theta, \zeta \in \cZ} | (\PP_n - \PP) \partial_\gamma \psi(\gamma, \zeta) | = o_P(1). \label{A-eq:lem:conv-ee-deriv:proofuse2}
	\end{align}	

	For the second term in \cref{A-eq:lem:conv-ee-deriv:proofuse1}, it follows from the fact that $\bar\gamma_n \pto \gamma^\star$ (lemma assumption), the dominatedness of $\partial_\gamma \psi(\gamma,\zeta)$ (\cref{A-asu:reg-dominated-ee-deriv}), and the dominated convergence theorem for convergence in probability (see, e.g., \citet[][Result (viii) of Chapter 3.2]{chung2001course}) that
	\begin{align}
		\PP \big\{ \partial_\gamma \psi(\bar\gamma_n, \hat\zeta) \big\} - \PP \big\{ \partial_\gamma \psi(\gamma^\star, \hat\zeta) \big\} = o_P(1). \label{A-eq:lem:conv-ee-deriv:proofuse3}
	\end{align}

	For the third term in \cref{A-eq:lem:conv-ee-deriv:proofuse1}, \cref{A-asu:nuisance-conv-PPee-deriv} states that
	\begin{align}
		\PP \big\{ \partial_\gamma \psi(\gamma^\star, \hat\zeta) \big\} - \PP \big\{ \partial_\gamma \psi(\gamma^\star, \hat\zeta) \big\} = o_P(1). \label{A-eq:lem:conv-ee-deriv:proofuse4}
	\end{align}

	Plugging \cref{A-eq:lem:conv-ee-deriv:proofuse4,A-eq:lem:conv-ee-deriv:proofuse3,A-eq:lem:conv-ee-deriv:proofuse2} into \cref{A-eq:lem:conv-ee-deriv:proofuse1} yields
	\begin{align*}
		\PP_n \big\{ \partial_\gamma \psi(\bar\gamma_n, \hat\zeta) \big\} - \PP \big\{ \partial_\gamma \psi(\gamma^\star, \zeta') \big\} = o_P(1).
	\end{align*}
	The proof is thus completed.
\end{proof}

\bigskip

\begin{lem}[Convergence of the ``meat'' term.]
	\label{A-lem:conv-ee-meat}
	Suppose the following regularity conditions hold: \cref{A-asu:reg-bounded-and-cont-differentiable-ee,A-asu:reg-dominated-ee-meat,A-asu:nuisance-conv-PPee-meat,A-asu:donsker}. If a sequence of random variables $\bar\gamma_n$ satisfies $\bar\gamma_n \pto \gamma^\star$, then $\PP_n \{ \psi(\bar\gamma_n, \hat\zeta) \psi(\bar\gamma_n, \hat\zeta)^T \} \pto \PP \{ \psi(\gamma^\star, \zeta') \psi(\gamma^\star, \zeta')^T \}$.
\end{lem}

\begin{proof}[Proof of \cref{A-lem:conv-ee-meat}]
	The proof is identical to the proof of \cref{A-lem:conv-ee-deriv}, except that we need to establish the following:
	\begin{itemize}
	 	\item $\psi(\gamma, \zeta) \psi(\gamma, \zeta)^T$ is continuous: this is \cref{A-asu:reg-bounded-and-cont-differentiable-ee};
	 	\item $\psi(\gamma, \zeta) \psi(\gamma, \zeta)^T$ is bounded by an integrable function: this is \cref{A-asu:reg-dominated-ee-meat};
	 	\item $\PP \{\psi(\gamma^\star, \zeta) \psi(\gamma^\star, \hat\zeta)^T\} - \PP \{\psi(\gamma^\star, \zeta) \psi(\gamma^\star, \zeta')^T\} = o_P(1)$: this is \cref{A-asu:nuisance-conv-PPee-meat};
	 	\item $\psi(\gamma, \zeta) \psi(\gamma, \zeta)^T$ takes value in a $P$-Glivenko-Cantelli class.
	\end{itemize}
	We establish the last bullet point now. We know $\psi(\gamma,\zeta)$ is uniformly bounded (\cref{A-asu:reg-bounded-and-cont-differentiable-ee}), $\Psi$ is a $P$-Donsker class (\cref{A-asu:donsker}), and the product operation $f \cdot g$ is a Lipschitz transformation. These three statements imply that $\Psi\times\Psi$ is a $P$-Donsker class (see, e.g., \citet[][Example 19.20]{van2000asymptotic}). Thus, $\{\psi(\gamma, \zeta) \psi(\gamma, \zeta)^T: \gamma\in\Theta, \zeta \in \cZ\} \subset \Psi\times\Psi$ is a $P$-Donsker class and thus a $P$-Glivenko-Cantelli class. This completes the proof.
\end{proof}

\bigskip

\begin{lem}[Lemma 19.24 of \citet{van2000asymptotic}]
	\label{A-lem:vdv19.24}
	Suppose that $\mathcal{F}$ is a $P$-Donsker class of measurable functions and $\hat{f}_n$ is a sequence of random functions that take their values in $\mathcal{F}$ such that $\int\{ \hat{f}_n(x) - f_0(x) \}^2 dP(x)$ converges in probability to 0 for some $f_0 \in L_2(P)$. Then $\mathbb{G}_n(\hat{f}_n - f_0) \pto 0$ and hence $\mathbb{G} \hat{f}_n \rightsquigarrow \mathbb{G}_P(f_0)$.
\end{lem}

\subsection{Proof of \texorpdfstring{\cref{A-thm:rate-mr-ncf}}{Theorem B.X}}
\label{A-subsec:asymptotic-proof-ncf}

We are now ready to establish the asymptotic theory. We first prove \cref{thm:estimator-consistency}.

\begin{proof}[Proof of \cref{thm:estimator-consistency}]
	The theorem follows immediately from \cref{A-lem:psi-multiply-robust} and \cref{A-lem:consistency}.
\end{proof}

We now prove \cref{A-thm:rate-mr-ncf}, which is \cref{thm:estimator-normality} in the main paper for the non-cross-fitted estimator $\hat\gamma$.

\begin{proof}[Proof of \cref{A-thm:rate-mr-ncf}]
	Because $\psi(\gamma,\zeta)$ is continuously differentiable in $\gamma$ (\cref{A-asu:reg-bounded-and-cont-differentiable-ee}), the Lagrange mean value theorem implies that
	\begin{align}
		0 = \PP_n\{ \psi(\hat\gamma, \hat\zeta)\} & = \PP_n\{ \psi(\gamma^\star, \hat\zeta) \} + \left[ \frac{\partial}{\partial \gamma^T} \PP_n\{\psi(\bar\gamma, \hat\zeta)\} \right](\hat\gamma - \gamma^\star) \nonumber \\
		& = \PP_n\{ \psi(\gamma^\star, \hat\zeta) \} + \PP_n\{\partial_\gamma \psi(\bar\gamma, \hat\zeta)\} (\hat\gamma - \gamma^\star), \label{A-eq:thm:normality:proofuse1}
	\end{align}
	where $\bar\gamma$ is between $\hat\gamma$ and $\gamma^\star$.
	Because \cref{A-lem:consistency} implies that $\hat\gamma \pto \gamma^\star$ and thus $\bar\gamma \pto \gamma^\star$, \cref{A-lem:conv-ee-deriv} implies that
	\begin{align}
		\PP_n\{ \partial_\gamma \psi(\bar\gamma, \hat\zeta)\} \pto \PP\{\partial_\gamma \psi(\gamma^\star, \zeta')\}. \label{A-eq:thm:normality:proofuse2}
	\end{align}
	Because $\PP\{\partial_\gamma \psi(\gamma^\star,\zeta')\}$ is invertible (\cref{A-asu:reg-invertible-ee-deriv}), \cref{A-eq:thm:normality:proofuse2} implies that $\PP_n\{ \partial_\gamma \psi(\bar\gamma, \hat\zeta)\}$ is invertible with probability approaching 1. Therefore, \cref{A-eq:thm:normality:proofuse1,A-eq:thm:normality:proofuse2} imply that with probability approaching 1 we have
	\begin{align}
		\sqrt{n} (\hat\gamma - \gamma^\star) = - [ \PP\{\partial_\gamma \psi(\gamma^\star, \zeta')\} ]^{-1} [ \sqrt{n} \PP_n\{ \psi(\gamma^\star, \hat\zeta)\} ]. \label{A-eq:thm:normality:proofuse3}
	\end{align}
	
	We now show that the second term in \cref{A-eq:thm:normality:proofuse3} is asymptotically normal. The assumed convergence rates for the nuisance function estimates [\cref{A-eq:nuisance-product-rate1,A-eq:nuisance-product-rate2,A-eq:nuisance-product-rate3}] imply that for each $t \in [T]$, the limit $\zeta_t'$ satisfies one of the four conditions: (i) $p_t' = p_t^\star$ and $q_t' = q_t^\star$; or (ii) $p_t' = p_t^\star$ and $\mu_t' = \mu_t^\star$; or (iii) $\eta_t' = \eta_t^\star$, $\nu_t' = \nu_t^\star$ and $q_t' = q_t^\star$; or (iv) $\eta_t' = \eta_t^\star$, $\nu_t' = \nu_t^\star$ and $\mu_t' = \mu_t^\star$. Therefore, \cref{A-lem:psi-multiply-robust} implies that $\PP \{\psi(\gamma^\star, \zeta')\} = 0$, and thus we have
	\begin{align}
		\PP_n \{ \psi(\gamma^\star, \hat\zeta) \} & = \PP_n\{ \psi(\gamma^\star, \hat\zeta) \} - \PP\{ \psi(\gamma^\star, \hat\zeta) \} + \PP\{ \psi(\gamma^\star, \hat\zeta) \} - \PP \{\psi(\gamma^\star, \zeta')\} \nonumber \\
		& = (\PP_n - \PP)\{ \psi(\gamma^\star, \hat\zeta) \} + \PP\{ \psi(\gamma^\star, \hat\zeta) - \psi(\gamma^\star, \zeta')\}. \label{A-eq:thm:normality:proofuse4}
	\end{align}

	For the term $(\PP_n - \PP)\{ \psi(\gamma^\star, \hat\zeta) \}$ in \cref{A-eq:thm:normality:proofuse4}, because of \cref{A-asu:nuisance-conv-ee-l2}, we can invoke \cref{A-lem:vdv19.24} to get
	\begin{align*}
		\mathbb{G}_n\{\psi(\gamma^\star, \hat\zeta) - \psi(\gamma^\star, \zeta')\} \pto 0,
	\end{align*}
	or equivalently
	\begin{align}
		(\PP_n - \PP)\{ \psi(\gamma^\star, \hat\zeta) \} = (\PP_n - \PP)\{ \psi(\gamma^\star, \zeta') \} + o_P(n^{-1/2}). \label{A-eq:thm:normality:proofuse5}
	\end{align}

	For the term $\PP\{ \psi(\gamma^\star, \hat\zeta) - \psi(\gamma^\star, \zeta')\}$ in \cref{A-eq:thm:normality:proofuse4}, we have
	\begin{align}
		& ~~~~ \Big|\PP\{ \psi(\gamma^\star, \hat\zeta) - \psi(\gamma^\star, \zeta')\} \Big| \nonumber \\
		& \lesssim \max_{t \in [T]} \bigg\{\|\hat{p}_t - p_t^\star\| \cdot \Big(\|\hat\eta_t - \eta_t^\star\| + \|\hat\nu_t - \nu_t^\star\| \Big) + \|\hat{q}_t - q_t^\star\| \cdot \|\hat\mu_t - \mu_t^\star\| \bigg\} \label{A-eq:thm:normality:proofuse6} \\
		& = o_P(n^{-1/2}), \label{A-eq:thm:normality:proofuse7}
	\end{align}
	where \cref{A-eq:thm:normality:proofuse6} follows from \cref{A-lem:psi-rate-multiply-robust} and \cref{A-eq:thm:normality:proofuse7} follows from the assumed convergence rates for the nuisance function estimates [\cref{A-eq:nuisance-product-rate1,A-eq:nuisance-product-rate2,A-eq:nuisance-product-rate3}].

	Plugging \cref{A-eq:thm:normality:proofuse5,A-eq:thm:normality:proofuse7} into \cref{A-eq:thm:normality:proofuse4} and we get
	\begin{align}
		\sqrt{n} \PP_n \{ \psi(\gamma^\star, \hat\zeta) \} & = \sqrt{n}(\PP_n - \PP)\{ \psi(\gamma^\star, \zeta') \} + o_P(1) \nonumber \\
		& \dto N \Big(0, ~\PP\{ \psi(\gamma^\star, \zeta') \psi(\gamma^\star, \zeta')^T \} \Big), \label{A-eq:thm:normality:proofuse8}
	\end{align}
	where the last step follows from the Lindeberg-Feller Central Limit Theorem and the fact that $\PP\{ \psi(\gamma^\star, \zeta')\} = 0$ (\cref{A-lem:psi-multiply-robust}).

	Therefore, plugging \cref{A-eq:thm:normality:proofuse8} into \cref{A-eq:thm:normality:proofuse3}, it follows from Slutsky's Theorem that $\sqrt{n} (\hat\gamma - \gamma^\star) \dto N(0, V)$ with $V$ defined in the theorem statement. Consistency of the variance estimator follows immediately from \cref{A-lem:conv-ee-deriv,A-lem:conv-ee-meat} and the continuous mapping theorem. This completes the proof.
\end{proof}

\subsection{Proof of \texorpdfstring{\cref{A-thm:rate-mr-cf}}{Theorem B.X}}
\label{A-subsec:asymptotic-proof-cf}

The proof for the asymptotic normality of the cross-fitting estimator is through verifying the conditions for Theorems 3.1 and 3.2 in \citet{chernozhukov2018double}. The technical details are similar to the proof of Theorem 4.2 in \citet{farbmacher2022causal}, and much of the technical results reuse the proof of \cref{A-thm:rate-mr-ncf}. Thus the proof is omitted.

\section{Additional Details and Results for Simulation}
\label{A-sec:additional_details_and_results_for_simulation}

\subsection{Deriving the True Nuisance Functions for GM-1}
\label{A-subsec:deriving_the_true_nuisance_functions_for_gm_1}

Because $P(A_t = a, M_t = m \mid H_t) = s_{am} / s$, by marginalizing over $m = \{0,1\}$ we obtain $p_t^\star$:
\begin{align*}
	p_t^\star(1 | H_t) & = P(A_t = 1 \mid H_t) \\
	& = P(A_t = 1, M_t = 1 \mid H_t) + P(A_t = 1, M_t = 0 \mid H_t) \\
	& = \frac{s_{10} + s_{11}}{s_{10} + s_{11} + s_{00} + s_{01}} \\
	& = \frac{e^{\kappa_1 + h_1(t,X_t)} + e^{\kappa_0 + \kappa_1 + \kappa_2 + h_1(t, X_t) + h_2(t, X_t)}}{1 + e^{\kappa_1 + h_1(t,X_t)} + e^{\kappa_2 + h_2(t,X_t)} + e^{\kappa_0 + \kappa_1 + \kappa_2 + h_1(t, X_t) + h_2(t, X_t)}},
\end{align*}
and thus $p_t^\star(0 | H_t) = 1 - p_t^\star(1 | H_t)$.

To derive $q_t^\star$, we first derive
\begin{align*}
	P(A_t = 1 \mid M_t = 0, X_t) & = \frac{e^{\kappa_1 + h_1(t, X_t)}}{1 + e^{\kappa_1 + h_1(t, X_t)}} = \expit\{\kappa_1 + h_1(t, X_t)\}, \\
	P(A_t = 1 \mid M_t = 1, X_t) & = \frac{e^{\kappa_0 + \kappa_1 + \kappa_2 + h_1(t, X_t) + h_2(t, X_t)}}{e^{\kappa_2 + h_2(t, X_t)} + e^{\kappa_0 + \kappa_1 + \kappa_2 + h_1(t, X_t) + h_2(t, X_t)}} = \expit\{\kappa_0 + \kappa_1 + h_1(t, X_t)\}.
\end{align*}
Thus,
\begin{align*}
	q_t^\star(1 | H_t, M_t) = P(A_t = 1 \mid H_t, M_t) = \expit\{\kappa_0 M_t + \kappa_1 + h_1(t, X_t)\},
\end{align*}
and thus $q_t^\star(0 | H_t, M_t) = 1 - q_t^\star(1 | H_t, M_t)$.

Before deriving $\eta_t^\star, \mu_t^\star, \nu_t^\star$, we derive a few intermediate results. 
We have
\begin{align*}
	P(M_t = 1 \mid A_t = 0, X_t) & = \frac{e^{\kappa_2 + h_2(t, X_t)}}{1 + e^{\kappa_2 + h_2(t, X_t)}} = \expit\{\kappa_2 + h_2(t, X_t)\}, \\
	P(M_t = 1 \mid A_t = 1, X_t) & = \frac{e^{\kappa_0 + \kappa_1 + \kappa_2 + h_1(t, X_t) + h_2(t, X_t)}}{e^{\kappa_1 + h_1(t, X_t)} + e^{\kappa_0 + \kappa_1 + \kappa_2 + h_1(t, X_t) + h_2(t, X_t)}} = \expit\{\kappa_0 + \kappa_2 + h_2(t, X_t)\}.
\end{align*}
Thus,
\begin{align*}
	P(M_t = 1 \mid H_t, A_t) = \expit\{\kappa_0 A_t + \kappa_2 + h_2(t, X_t)\}.
\end{align*}
Furthermore, for each $t \in [T]$, we compute the following through numerical integration (by interagrating over the the distribution of $X_t$ which is Gaussian): $\EE(A_t), \EE(M_t), \EE(A_t M_t)$, and we define $\delta_t := \xi_t \EE(X_t) + \rho_t \EE(M_t) + \lambda_t \EE(A_t) + \tau_t \EE(A_t M_t)$.

Next we derive $\eta_t^\star, \mu_t^\star, \nu_t^\star$. We have
\begin{align*}
	Y = \sum_{t=1}^T (\xi_t X_t + \rho_t M_t + \lambda_t A_t + \tau_t A_t M_t) + \epsilon.
\end{align*}
Then for each $s \in [T]$, we have
\begin{align*}
	\mu_s^\star(A_s, H_s, M_s) & = \EE(Y \mid H_s, A_s, M_s) \\
	& = \sum_{t=1}^s (\xi_t X_t + \rho_t M_t + \lambda_t A_t + \tau_t A_t M_t) \\
	& ~~~~ + \EE\bigg\{ \sum_{t=s+1}^T (\xi_t X_t + \rho_t M_t + \lambda_t A_t + \tau_t A_t M_t) \mid H_s, A_s, M_s\bigg\} \\
	& = \sum_{t=1}^s (\xi_t X_t + \rho_t M_t + \lambda_t A_t + \tau_t A_t M_t) \\
	& ~~~~ + \sum_{t=s+1}^T \{(\xi_t \EE(X_t) + \rho_t \EE(M_t) + \lambda_t \EE(A_t) + \tau_t \EE(A_t M_t)\} \\
	& = \sum_{t=1}^s (\xi_t X_t + \rho_t M_t + \lambda_t A_t + \tau_t A_t M_t) + \sum_{t=s+1}^T \delta_t,
\end{align*}
and
\begin{align*}
	\nu_s^\star(a) & = \EE \{ \EE (Y \mid H_s, A_s = a, M_s) \mid H_s, A_s = 1-a\} \\
	& = \EE \bigg\{ \sum_{t=1}^{s-1} (\xi_t X_t + \rho_t M_t + \lambda_t A_t + \tau_t A_t M_t) \\
	& ~~~~ + \xi_s X_s + \rho_s M_s + \lambda_s a + \tau_s a M_s + \sum_{t=s+1}^T \delta_t ~\bigg|~ H_s, A_s = 1-a \bigg\} \\
	& = \sum_{t=1}^{s-1} (\xi_t X_t + \rho_t M_t + \lambda_t A_t + \tau_t A_t M_t) \\
	& ~~~~ + \xi_s X_s + \lambda_s a + (\rho_s + \tau_s a) \EE(M_s \mid H_s, A_s = 1-a) + \sum_{t=s+1}^T \delta_t \\
	& = \sum_{t=1}^{s-1} (\xi_t X_t + \rho_t M_t + \lambda_t A_t + \tau_t A_t M_t) \\
	& ~~~~ + \xi_s X_s + \lambda_s a + (\rho_s + \tau_s a) ~\expit\{\kappa_0(1-a) + \kappa_2 + h_2(s, X_s)\} + \sum_{t=s+1}^T \delta_t,
\end{align*}
and similarly
\begin{align*}
	\eta_s^\star(A_s) & = \EE(Y \mid H_s, A_s) = \EE\{\mu_s^\star(A_s, H_s, M_s) \mid H_s, A_s\} \\
	& = \sum_{t=1}^{s-1} (\xi_t X_t + \rho_t M_t + \lambda_t A_t + \tau_t A_t M_t) \\
	& ~~~~ + \xi_s X_s + \lambda_s A_s + (\rho_s + \tau_s A_s) ~\expit\{\kappa_0 A_s + \kappa_2 + h_2(s, X_s)\} + \sum_{t=s+1}^T \delta_t.
\end{align*}

\subsection{Additional Results from Simulation Study 1}
\label{A-subsec:additional_results_from_simulation_study_1}

Here we present additional results from simulation study 1. \cref{A-fig:simulation-1-marginal-cp} and \cref{A-fig:simulation-1-marginal-bias} are expanded versions of \cref{fig:simulation-1} with more $n$'s and the value of each cell displayed. \cref{A-fig:simulation-1-moderated-cp} and \cref{A-fig:simulation-1-moderated-bias} are parallel versions for when the effect of interest is in the mediation effect moderated by the decision point index $t$, i.e., when $f(t) = (1,t)^T$. In this case, $\alpha_1$ and $\beta_1$ represent the intercepts for NDEE and NIEE, and $\alpha_2$ and $\beta_2$ represent the slopes.

\begin{figure}[htbp]
    \includegraphics[width = \textwidth]{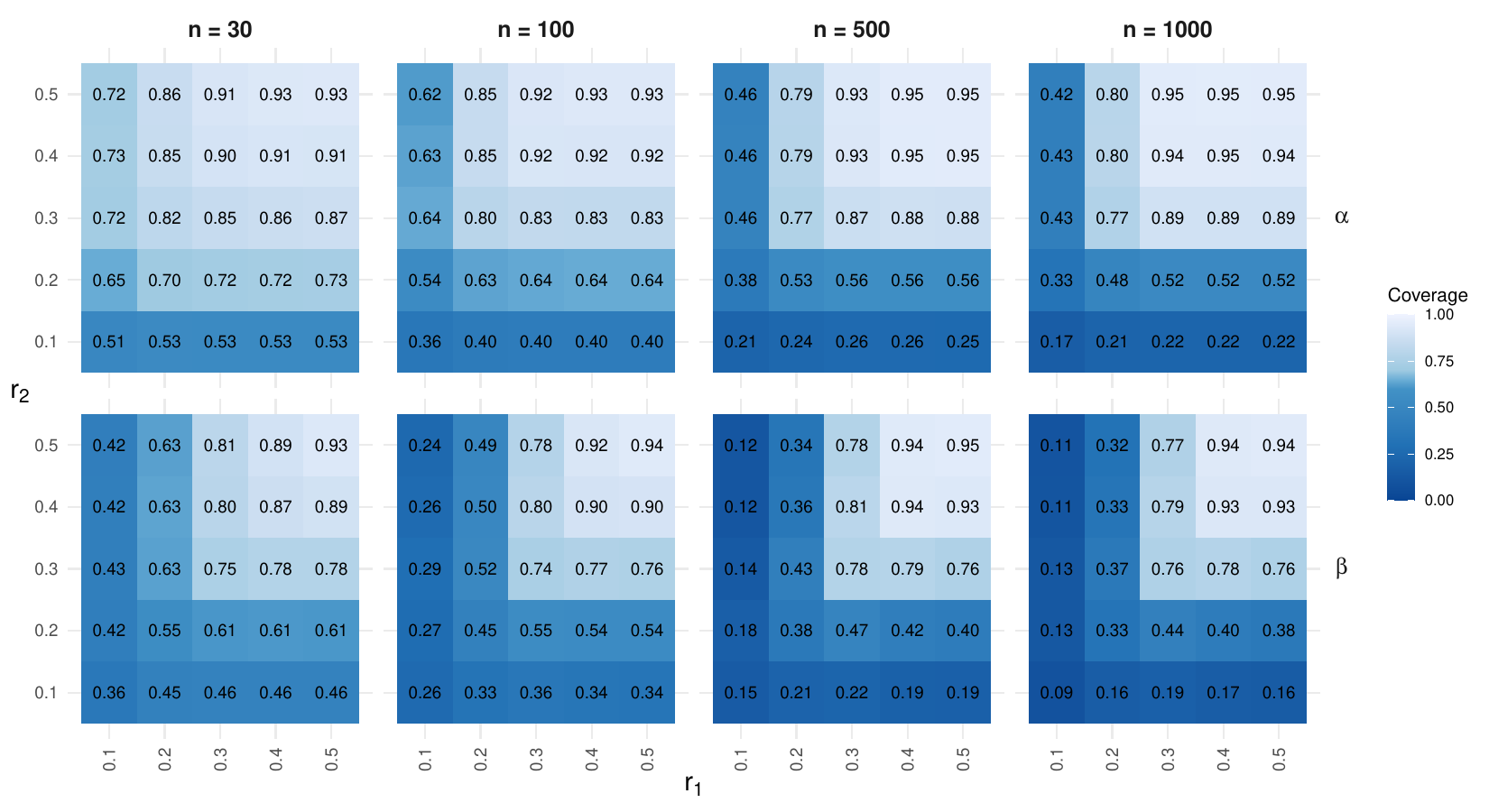}
    \caption{Additional simulation results under GM-1 for fully marginal effect with $f(t) = 1$: Coverage of 95\% confidence intervals for $\hat\alpha$ and $\hat\beta$ under various convergence rates of the nuisance parameters, $r_1$ and $r_2$.}
    \label{A-fig:simulation-1-marginal-cp}
\end{figure}

\begin{figure}[htbp]
    \includegraphics[width = \textwidth]{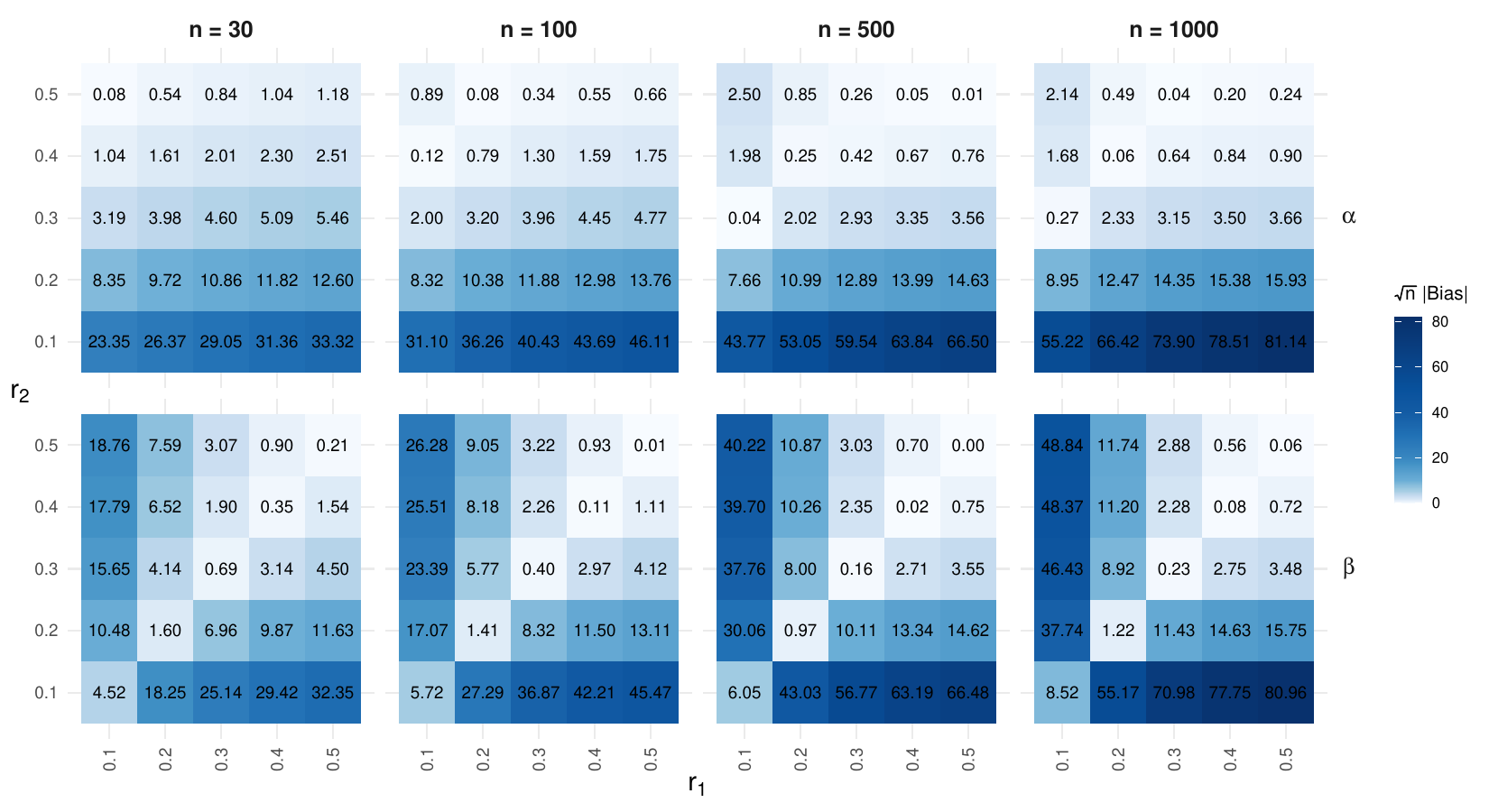}
    \caption{Additional simulation results under GM-1 for fully marginal effect with $f(t) = 1$: $\sqrt{n}~|\text{Bias}|$ for $\hat\alpha$ and $\hat\beta$ under various convergence rates of the nuisance parameters, $r_1$ and $r_2$.}
    \label{A-fig:simulation-1-marginal-bias}
\end{figure}

\begin{figure}[htbp]
    \includegraphics[width = \textwidth]{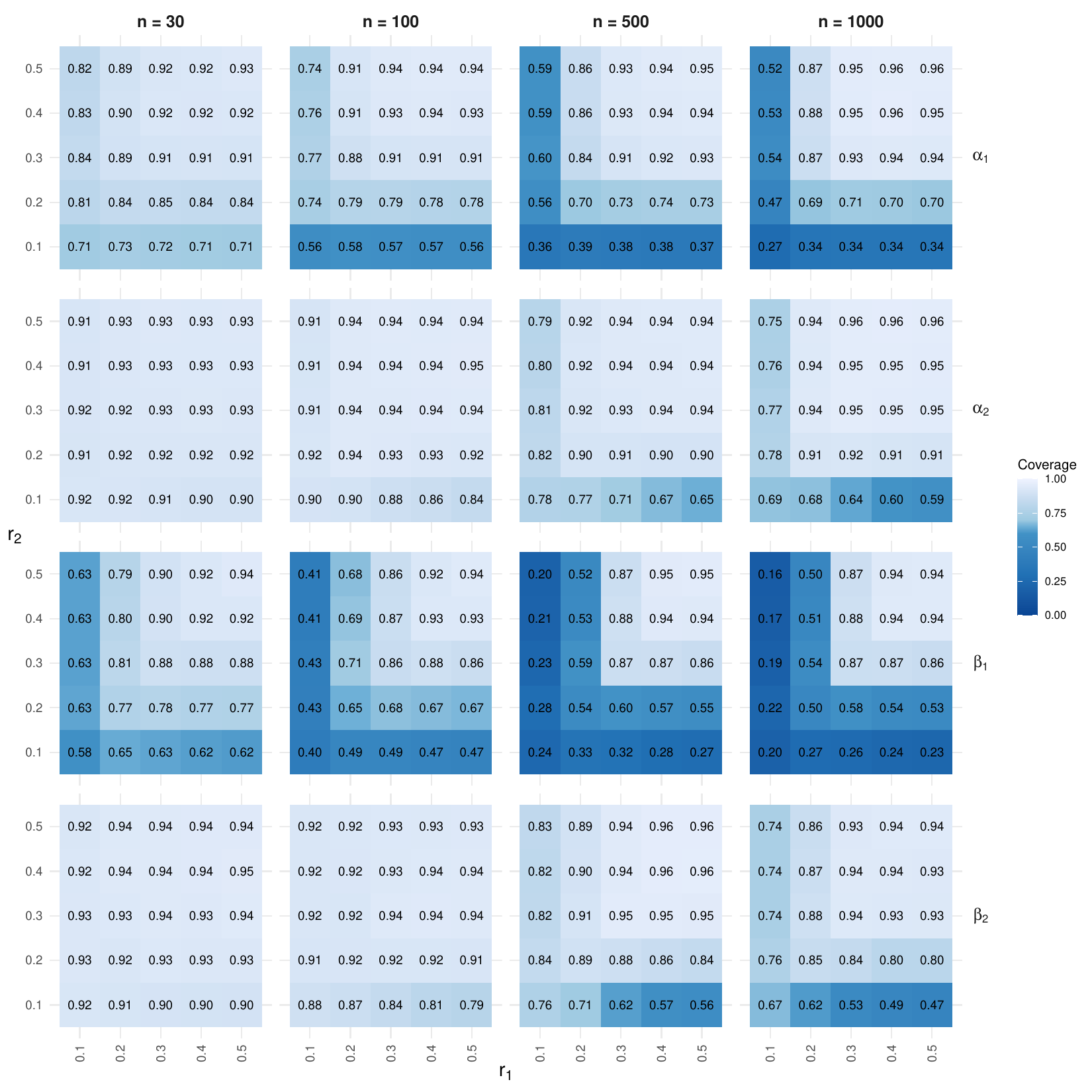}
    \caption{Additional simulation results under GM-1 for effect moderated by $t$ with $f(t) = (1,t)^T$: Coverage of 95\% confidence intervals for $\hat\alpha_1, \hat\alpha_2,\hat\beta_1, \hat\beta_2$ under various convergence rates of the nuisance parameters, $r_1$ and $r_2$.}
    \label{A-fig:simulation-1-moderated-cp}
\end{figure}

\begin{figure}[htbp]
    \includegraphics[width = \textwidth]{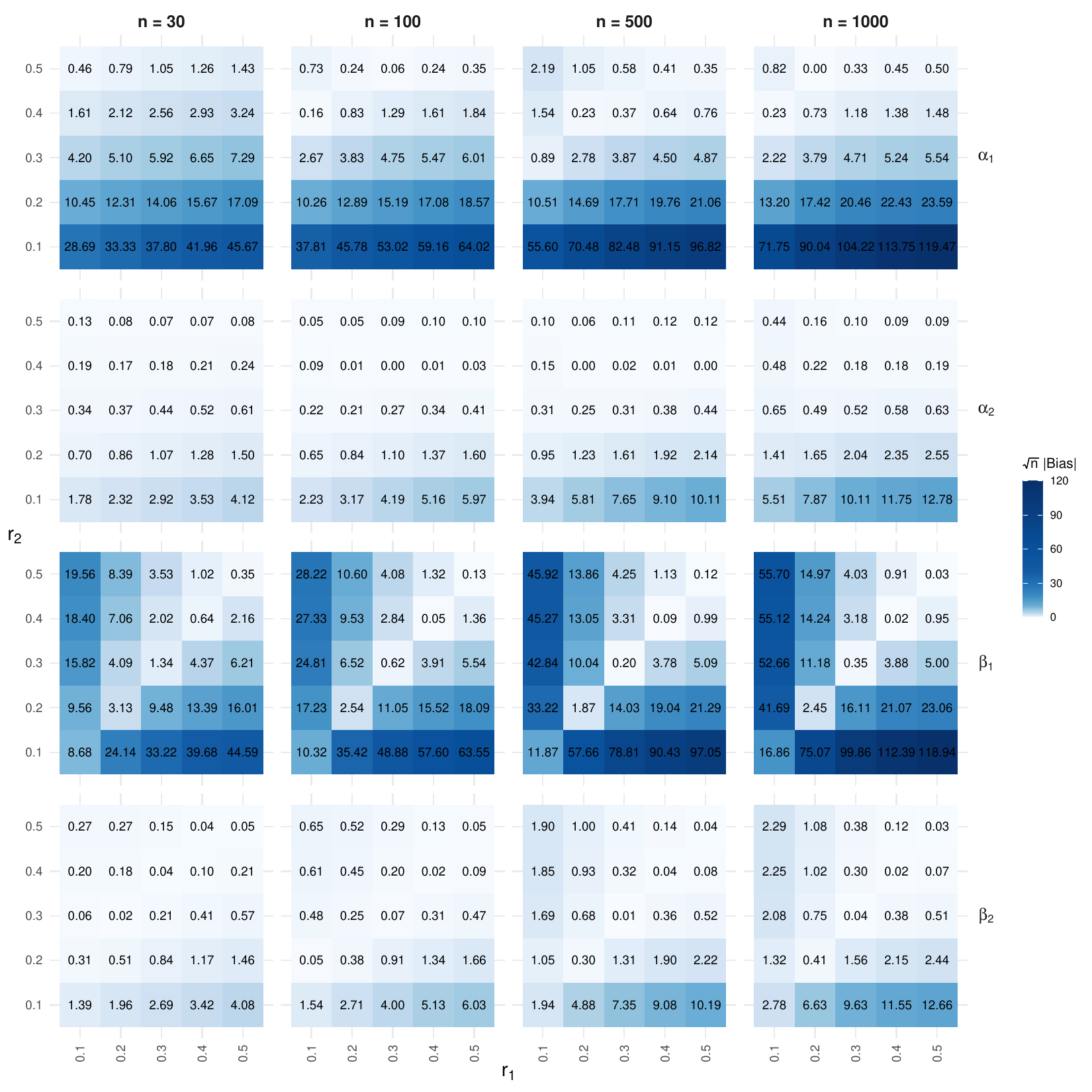}
    \caption{Additional simulation results under GM-1 for effect moderated by $t$ with $f(t) = (1,t)^T$: $\sqrt{n}~|\text{Bias}|$ for $\hat\alpha_1, \hat\alpha_2,\hat\beta_1, \hat\beta_2$ under various convergence rates of the nuisance parameters, $r_1$ and $r_2$.}
    \label{A-fig:simulation-1-moderated-bias}
\end{figure}

\section{SleepHealth: Data Preparation and Covariate Selection}
\label{A-sec:sleephealth}

\subsection{Data Sources and Participant Eligibility}
We combined three streams of self-reported, sensor, and questionnaire data for each study participant:
\begin{itemize}
  \item \textbf{Daily surveys}: sleepiness ratings, sleep-quality assessments, and PM check-ins.
  \item \textbf{Baseline questionnaires}: four instruments (\textit{About Me}, \textit{My Health}, \textit{Sleep Assessment}, \textit{Sleep Habits}) administered at enrollment.
  \item \textbf{HealthKit step counts}: minute-level step intervals recorded by participants' Apple devices.
\end{itemize}
Participants were deemed eligible if they completed the sleepiness, sleep-quality, and PM check-in surveys on at least six consecutive calendar days, and if that first 6-day common engagement streak began within 30 days of their first survey entry. We identified streaks by inner-joining the date-stamped survey records, grouping by participant, and computing run lengths of consecutive dates. The resulting list of participant IDs and the start date of their qualifying streak comprised our eligibility anchor.

\subsection{Data Integration and Processing}
\subsubsection{Baseline Covariates}
For each eligible participant, we loaded all four baseline questionnaires and selected one record per instrument, choosing the entry whose timestamp was closest to the participant's 6-day streak start. This ensured consistent temporal alignment between baseline covariates and the first sustained engagement period.

\subsubsection{Longitudinal Daily Summaries}
We restricted both the sleepiness and sleep-quality checkers to each participant's inclusive 6-day window. Within each day, multiple sleep-quality checkers were averaged to produce a single daily mean and a corresponding observation count. (For sleepiness, most often there are two observations per day. For sleep-quality checkers, most often there are one observations per day.) The two daily summaries were then full-joined by participant and date, yielding a complete longitudinal dataset over days 1--6.

\subsubsection{Step-Count Aggregation}
From the AppleHealthKit data, we retained only intervals overlapping the same 6-day window. Each record's duration and date span were computed, and total daily steps were summed by participant and start-time date. We then generated a full grid of all six days per participant and imputed zero for any missing day.

\subsubsection{Final Data Assembly}
A master file was created by left-joining the daily longitudinal summaries and baseline covariates, and then left-joining the daily-steps table.

\subsection{Covariate Screening and Selection}

\subsubsection{Initial Screening}

We began with all numeric variables in the assembled dataset (excluding identifiers, the binary exposure sleep-quality checkers, and the mediator sleepiness) and applied the following sequential filters:
\begin{enumerate}
  \item \textbf{Missingness filter}: retained only variables observed in $\geq 70\%$ of person-days (this only matters for baseline covariates).
  \item \textbf{Variance filter}:
  \begin{itemize}
    \item Continuous or categorical variables required standard deviation $\geq 0.01$.
    \item Binary variables required proportion of ones between 1\% and 99\%.
  \end{itemize}
  \item \textbf{Collinearity filter}: eliminated one of each pair with Pearson $|r| > 0.8$ \\(using \texttt{caret::findCorrelation}).
\end{enumerate}


\subsubsection{Dual LASSO}
We performed two separate LASSO regressions (via \texttt{glmnet}) with one-standard-error lambda ($\lambda_{\mathrm{1se}}$):
\begin{enumerate}
  \item \textbf{Outcome model}: Gaussian to predict continuous sleep-quality score.
  \item \textbf{Exposure model}: Binomial to predict the binary exposure.
\end{enumerate}
The union of nonzero-coefficient predictors from both models formed our initial multivariable covariate set.

\subsubsection{Discretization and Manual Recoding}

Continuous predictors were discretized by selecting cutoffs that maximized the absolute $t$-statistic for the final day outcome. A suite of categorical items (e.g., income, education, smoking status, work schedule, fatigue frequency, total steps) was recoded into binary or grouped levels (e.g., ``$\leq 6$ hours'' vs. ``$>6$ hours'', ``zero steps'' vs. ``1--4500 steps'' vs. ``$> 4500$ steps''). Below is the detailed list of covariates and discretization used.

\spacingset{1.5}
\scriptsize
\begin{longtable}{@{}p{0.18\textwidth}p{0.41\textwidth}p{0.41\textwidth}@{}}
\toprule
\textbf{Variable} & \textbf{Question} & \textbf{Discretization} \\
\midrule
\endhead

basic\_expenses &
\emph{How hard is it for you (and your family) to pay for very basics like food, rent, mortgage, or heating, etc.?} &
Unknown if missing or code 6; Not very hard or better if response $>$ 3; Somewhat hard or worse if $\leq$ 3. \\

daily\_smoking &
\emph{Do you now smoke every day, some days, or not at all?} &
Unknown if missing or code 4; Not at all if response $>$ 2; Some days or every day if $\leq$ 2. \\

education &
\emph{What is the highest degree or level of school you have completed?} &
Unknown if missing or code 7; College graduate or more if response $>$ 4; Some college or less if $\leq$ 4. \\

gender &
\emph{What is your sex?} &
Unknown if missing or code 3; Female if code 1; Male if code 2. \\

good\_life &
\emph{I have a good life.} &
Unknown if missing; Strongly agree if response $>$ 4; Agree or less if $\leq$ 4. \\

hispanic &
\emph{Are you of Hispanic, Latino or Spanish origin or ancestry? Please check all that apply.} &
Unknown if missing or codes 7/8; No if code 1; Yes if $>$ 1. \\

income &
\emph{Last year, what was your total household income from all sources, before taxes?} &
Unknown if missing or codes 8/9; 49k or less if response $\leq$ 2; 50k or more if $>$ 2. \\

race &
\emph{What is your race? Please check all that apply.} &
Unknown if missing or response $\geq$ 8; White if code 5; Non-white if codes 1-4,6-7. \\

smoking\_status &
\emph{How would you describe your smoking status?} &
Unknown if missing or codes 6/9/10; Never smoker if code 4; Current or former smoker if codes 1-3,5,7-8. \\

work\_schedule &
\emph{Thinking about the past 3 months, which of the following best describes your work schedule?} &
Unknown if missing or code 7; Regular day shifts if code 1; Other shifts if codes 2-6. \\

exercise &
\emph{In the past 7 days…how often did you have enough energy to exercise strenuously?} &
Unknown if missing; Sometimes or more if response $>$ 2; Rarely or never if $\leq$ 2. \\

fatigue\_limit &
\emph{In the past 7 days…how often did fatigue limit you at work (include work at home)?} &
Unknown if missing; Sometimes or more if response $>$ 2; Rarely or never if $\leq$ 2. \\

feel\_tired\_frequency &
\emph{In the past 7 days…how often did you feel tired?} &
Unknown if missing; Sometimes or more if response $>$ 2; Rarely or never if $\leq$ 2. \\

felt\_alert &
\emph{In the past 7 days…I felt alert when I woke up.} &
Unknown if missing; Somewhat or more if response $>$ 2; A little bit or less if $\leq$ 2. \\

had\_problem &
\emph{In the past 7 days…I had a problem with my sleep.} &
Unknown if missing; Quite a bit or more if response $>$ 3; Somewhat or less if $\leq$ 3. \\

hard\_times &
\emph{In the past 7 days…I had a hard time getting things done because I was sleepy.} &
Unknown if missing; Somewhat or more if response $>$ 2; A little bit or less if $\leq$ 2. \\

tired\_easily &
\emph{In the past 7 days…I got tired easily.} &
Unknown if missing; Almost never or more if response $>$ 1; Never if $\leq$ 1. \\

alarm\_dependency &
\emph{How much do you depend on an alarm clock?} &
Unknown if missing; Very much if response $>$ 3; Somewhat or less if $\leq$ 3. \\

sleep\_needed &
\emph{How many hours of sleep do you need to function your best during the day?} &
Unknown if missing; $>$ 6 hours if response $>$ 6; $\leq$ 6 hours if $\leq$ 6. \\

sleep\_time\_weekend &
\emph{How many hours of sleep not including naps do you get per night on workdays?} &
Unknown if missing; $>$ 6 hours if response $>$ 6; $\leq$ 6 hours if $\leq$ 6. \\

wake\_up\_choices &
\emph{Approximately what time would you wake up if you were free to plan your day?} &
Unknown if missing; 8--9AM or later if response $>$ 4; 7--8AM or earlier if $\leq$ 4. \\

medication\_by\_doctor &
\emph{Medication prescribed by a doctor. How frequently do you use this treatment to help you sleep?} &
Unknown if missing or codes 6/7; Never if code 5; Rarely or more if $\leq$ 4. \\

total\_steps &
\emph{(Constructed from HealthKit step counts)} &
zero (effectively missing) if steps = 0; $\leq$ 4500 steps recorded if $\leq$ 4500; $>$ 4500 recorded if $>$ 4500. \\

\bottomrule
\end{longtable}

\end{appendices}

\end{document}